\def\llncs{0}
\def\fullpage{1}
\def\anonymous{0}
\def\authnote{0}
\def\notxfont{0}
\def\submission{0}
\def\reply{0}
\def\cameraready{0}
\def\noaux{0}
\def\showlabels{0}
\ifnum\submission=1
\def\anonymous{1}
\def\llncs{1}
\else
\fi

\ifnum\cameraready=1
\def\llncs{1}
\def\anonymous{0}
\def\authnote{0}
\else
\fi

\ifnum\anonymous=1
\def\authnote{0}
\else
\fi

\ifnum\llncs=1
	\documentclass[envcountsect,a4paper,runningheads]{llncs}
\else
	\documentclass[letterpaper,hmargin=1.05in,vmargin=1.05in,usenames]{article}
			\ifnum\fullpage=1
		\usepackage{fullpage}
		\fi
\fi

\ifnum\reply=1
\usepackage{ulem}
\renewcommand{\emph}{\textit}
\else
\fi

\def\makeuppercase#1{
\expandafter\newcommand\csname sf#1\endcsname{\mathsf{#1}}
\expandafter\newcommand\csname frak#1\endcsname{\mathfrak{#1}}
\expandafter\newcommand\csname bb#1\endcsname{\mathbb{#1}}
\expandafter\newcommand\csname bf#1\endcsname{\textbf{#1}}
}

\def\makelowercase#1{
\expandafter\newcommand\csname frak#1\endcsname{\mathfrak{#1}}
\expandafter\newcommand\csname bf#1\endcsname{\textbf{#1}}
}

\newcounter{char}
\loop
	\edef\letter{\alph{char}}
	\edef\Letter{\Alph{char}}
	\expandafter\makelowercase\letter
	\expandafter\makeuppercase\Letter
	\stepcounter{char}
	\unless\ifnum\thechar>26
\repeat

\usepackage[%
pdfpagelabels=true,
linktocpage=false,
bookmarks=true,bookmarksnumbered=false,bookmarkstype=toc,
pagebackref=true,
colorlinks=true,linkcolor=darkblue,urlcolor=darkblue,citecolor=darkviolet
]
{hyperref}%  dvipdfmx, 

\usepackage{amsmath, amsfonts, amssymb, mathtools,amscd}

\usepackage{amsthm}

\usepackage{lmodern}
\usepackage[T1]{fontenc}
\usepackage[utf8]{inputenc}
\usepackage{etex}
\usepackage{arydshln} % In order to use \hdashline
\usepackage{url}
\usepackage{ifthen}
\usepackage{bm}
\usepackage{multirow}
\usepackage[dvips]{graphicx}
\usepackage{xcolor,colortbl} % Leave out in case the usepackage cannot be found. Not critical
\usepackage{tikz}
\usetikzlibrary{matrix,arrows}
\usetikzlibrary{positioning}
\usetikzlibrary{decorations,decorations.text}
\usetikzlibrary{decorations.pathmorphing}
\usetikzlibrary{shapes}
\usepackage{threeparttable}
\usepackage{centernot}
\usepackage{dashbox}
\usepackage{comment}
\usepackage{paralist,verbatim}
\usepackage{cases}
\usepackage{booktabs}
\usepackage{fancybox}
\usepackage{braket}
\usepackage{cancel} 
\usepackage{ascmac} 
\usepackage{framed}
\usepackage{authblk}
\usepackage{pifont}
\usepackage{physics}

\definecolor{darkblue}{rgb}{0,0,0.6}
\definecolor{darkgreen}{rgb}{0,0.5,0}
\definecolor{maroon}{rgb}{0.5,0.1,0.1}
\definecolor{dpurple}{rgb}{0.2,0,0.65}
\definecolor{chocolate}{rgb}{0.8,0.4,0.1}
\definecolor{darkviolet}{RGB}{130,95,141}

\usepackage[capitalise,noabbrev]{cleveref}
\usepackage[absolute]{textpos}
\usepackage[final]{microtype}
\usepackage[absolute]{textpos}
\usepackage{everypage}
\usepackage{autonum}

\usepackage{dsfont}
\DeclareMathAlphabet{\mathpzc}{OT1}{pzc}{m}{it}

\ifnum\submission=1
\renewcommand*{\backref}[1]{}
\def\notxfont{1}
\else
\fi

\ifnum\llncs=1
\renewcommand{\subparagraph}{\paragraph}
\else
\ifnum\notxfont=1
\else
\usepackage{mathpazo}
\usepackage{newtxtext}
\usepackage{helvet}
\fi
\fi

\ifnum\showlabels=1
\usepackage{showkeys}
\else
\fi

\usepackage{fontawesome}

\newtheoremstyle{thicktheorem}%
{\topsep}
{\topsep}
{\itshape}{}%
{\bfseries}%
{.}
{ }%
{\thmname{#1}\thmnumber{ #2}%
		\thmnote{ (#3)}%
}

\newtheoremstyle{remark}%name
{\topsep}
{\topsep}
	{}%body font
	{}%indent amount
	{}%theorem head font
	{.}%punctuation after theorem head
	{ }%space after theorem head
	{\textit{\thmname{#1}}\thmnumber{ #2}%theorem head specs
			\thmnote{ (#3)}%
	}

\ifnum\llncs=0
	\theoremstyle{thicktheorem}
	\newtheorem{theorem}{Theorem}[section]
	\newtheorem{lemma}[theorem]{Lemma}
	\newtheorem{corollary}[theorem]{Corollary}
	
	\newtheorem{definition}[theorem]{Definition}

	\theoremstyle{remark}
	
	\newtheorem{remark}[theorem]{Remark}

\fi
	\crefname{theorem}{Theorem}{Theorems}
	\crefname{assumption}{Assumption}{Assumptions}
	\crefname{construction}{Construction}{Constructions}
	\crefname{corollary}{Corollary}{Corollaries}
	\crefname{conjecture}{Conjecture}{Conjectures}
	\crefname{definition}{Definition}{Definitions}
	\crefname{exmaple}{Example}{Examples}
	\crefname{experiment}{Experiment}{Experiments}
	\crefname{counterexample}{Counterexample}{Counterexamples}
	\crefname{lemma}{Lemma}{Lemmata}
	\crefname{observation}{Observation}{Observations}
	\crefname{proposition}{Proposition}{Propositions}
	\crefname{remark}{Remark}{Remarks}
	\crefname{claim}{Claim}{Claims}
	\crefname{fact}{Fact}{Facts}
	\crefname{note}{Note}{Notes}

\ifnum\llncs=1
 \crefname{appendix}{App.}{Appendices}
 \crefname{section}{Sec.}{Sections}
\else
\fi

\ifnum\llncs=1
\renewcommand*{\backref}[1]{}
\else
	\renewcommand*{\backref}[1]{(Cited on page~#1.)}
	\ifnum\notxfont=1
	\else
		\usepackage{newtxtext}
	\fi
\fi

\newcommand{\maskval}{y_{\mathsf{mask}}}
\newcommand{\maskeduk}{\ct_{\uk}}
\newcommand{\comrand}{y_{\mathsf{comr}}}
\newcommand{\uekeylen}{\ell_{\mathsf{uk}}}
\newcommand{\comrandlen}{\ell_{\mathsf{comr}}}
\newcommand{\pfinplen}{\ell_{\mathsf{inp}}}

\newcommand{\PFone}{\calP\calF_{\pfinplen}^1}
\newcommand{\PFm}{\calP\calF^{\mathsf{mlt}}}
\newcommand{\PFs}{\calP\calF^1}
\newcommand{\Resamp}{\mathtt{Resamp}}

\newcommand{\Dresamp}[1]{D_{#1\textrm{-}\mathsf{resamp}}}

\newcommand{\Com}{\mathsf{Com}}
\newcommand{\EqSetup}{\mathsf{EqSetup}}
\newcommand{\Open}{\mathsf{Open}}

\newcommand{\numa}{n}

\newcommand{\CPABE}{\mathsf{CPABE}}

\newcommand{\UPE}{\mathsf{UPE}}
\newcommand{\upe}{\mathsf{upe}}

\newcommand{\qABESim}{\qalgo{ABESim}}
\newcommand{\qSimEnc}{\qalgo{SimEnc}}
\newcommand{\qSimKG}{\qalgo{SimKG}}

\newcommand{\ptriv}{p^{\mathtt{triv}}}

\newcommand{\CCObf}{\mathsf{CC}.\mathsf{Obf}}

\newcommand{\lock}{\mathsf{lock}}

\newcommand{\CC}{\mathsf{CC}}

\newcommand{\expo}{\mathsf{exp}}
\newcommand{\MoEBB}{\mathsf{MoEBB84}}
\newcommand{\MoEBBLeak}{\mathsf{MoEBB84Leak}}

\newcommand{\QFHE}{\mathsf{QFHE}}
\newcommand{\qEval}{\qalgo{Eval}}

\newcommand{\uk}{\mathsf{uk}}
\newcommand{\CP}{\mathsf{CP}}
\newcommand{\qCopyProtect}{\qalgo{CopyProtect}}
\newcommand{\qSim}{\qalgo{Sim}}

\newcommand{\qSDKG}{\SDE.\qalgo{KG}}
\newcommand{\SDEnc}{\SDE.\algo{Enc}}
\newcommand{\qSDDec}{\SDE.\qalgo{Dec}}
\newcommand{\sdct}{\mathsf{sdct}}

\newcommand{\BadDec}{\mathsf{BadDec}}

\newcommand{\CPFE}{\mathsf{CPFE}}

\newcommand{\HE}{\mathsf{HE}}
\newcommand{\HKG}{\mathsf{HKG}}
\newcommand{\HEnc}{\mathsf{HEnc}}
\newcommand{\HDec}{\mathsf{HDec}}
\newcommand{\hct}{\mathsf{hct}}

\newcommand{\UE}{\mathsf{UE}}

\newcommand{\Leak}{\mathsf{Leak}}

\newcommand{\ue}{\keys{ue}}

\newcommand{\qfhe}{\mathsf{qfhe}}

\newcommand{\Drev}{D^{\mathtt{rev}}}
\newcommand{\cPrev}{\cP^{\mathtt{rev}}}

\newcommand{\pp}{\mathsf{pp}}

\newcommand{\evalct}{\mathsf{evct}}

\newcommand{\SDE}{\algo{SDE}}

\newcommand{\pirateD}{\qalgo{D}}

\newcommand{\sde}{\algo{sde}}
\newcommand{\qdk}{\qstate{dk}}
\newcommand{\fe}{\algo{fe}}

\newcommand{\share}{s}

\newcommand{\qKG}{\qalgo{KG}}
\newcommand{\qDec}{\qalgo{Dec}}

\newcommand{\qencrypt}{\qalgo{Enc}}

\newcommand{\qdecrypt}{\qalgo{Dec}}
\newcommand{\qct}{\qalgo{ct}}

\newcommand{\cnc}[2]{\mathbf{CC}[#1,#2]}

\newcommand{\qB}{\qalgo{B}}
\newcommand{\qC}{\qalgo{C}}
\newcommand{\qD}{\qalgo{D}}
\newcommand{\inplen}{n}

\newcommand{\msglen}{k}

\newcommand{\Oracle}[1]{O_{\mathtt{#1}}}

\makeatletter
\newcounter{game}
\def\newGames#1#2#3{%
  \xdef\gameNS{#1}\xdef\gamePrefix{#2}\setcounter{game}{#3}\addtocounter{game}{-1}%
  \immediate\write\@auxout{\string\expandafter\string\xdef\noexpand\csname game-prefix-#1\string\endcsname{#2}}%
}
\def\newGame#1{%
  \xdef\prevGame{\gamePrefix\arabic{game}}\stepcounter{game}\xdef\thisGame{\gamePrefix\arabic{game}}%
  \immediate\write\@auxout{\string\expandafter\string\xdef\noexpand\csname game-\gameNS-#1\string\endcsname{\arabic{game}}}%
}
\makeatother
\def\safecsname#1{\expandafter\ifx\csname#1\endcsname\relax\else\csname#1\endcsname\fi}

\renewcommand{\Game}[1][]{\mathcmd{\textrm{Game\if!#1!\else~\ensuremath{#1}\fi}}}

\newcommand{\td}{\mathsf{td}}

\usepackage[font=small,
format=plain,
labelformat=simple, 
singlelinecheck=false,
labelfont=bf,
up
]{caption}
\DeclareCaptionFormat{myformat}{#1#2#3\hrulefill}
\newcommand{\qA}{\qalgo{A}}
\newcommand{\qS}{\qalgo{S}}

\DeclareFontFamily{U}{skulls}{}
\DeclareFontShape{U}{skulls}{m}{n}{ <-> skull }{}

\newcommand{\qstateq}{\qstate{q}}

\usepackage{mathtools}
\newcommand{\chosen}{\leftarrow}

\renewcommand{\gets}{\leftarrow}

\newcommand{\la}{\leftarrow}
\newcommand{\ra}{\rightarrow}

\newcommand{\seteq}{\coloneqq}

\newcommand{\tensor}{\otimes}
\newcommand{\concat}{\|}

\newcommand{\setbk}[1]{\{#1\}}

\newcommand{\cC}{\mathcal{C}}
\newcommand{\cD}{\mathcal{D}}
\newcommand{\cE}{\mathcal{E}}

\newcommand{\cH}{\mathcal{H}}
\newcommand{\cI}{\mathcal{I}}
\newcommand{\cK}{\mathcal{K}}

\newcommand{\cM}{\mathcal{M}}
\newcommand{\cN}{\mathcal{N}}

\newcommand{\cP}{\mathcal{P}}

\newcommand{\cR}{\mathcal{R}}
\newcommand{\cS}{\mathcal{S}}

\newcommand{\cX}{\mathcal{X}}

\newcommand{\tlr}{\widetilde{r}}

\newcommand{\tlP}{\tilde{P}}

\newcommand{\tlC}{\widetilde{C}}
\newcommand{\tlD}{\widetilde{D}}

\newcommand{\tlp}{\widetilde{p}}

\newcommand{\R}{\mathbb{R}}

\newcommand{\Ms}{\mathcal{M}}

\newcommand{\coin}{\keys{coin}}

\newcommand{\Params}{\algo{\Params}}

\newcommand{\secp}{\lambda}

\newcommand{\ctlen}{{\ell_{\mathsf{ct}}}}

\newcommand{\IND}{\textrm{IND}}

\newcommand{\ind}{\mathsf{ind}}

\newcommand{\adva}[2]{\mathsf{Adv}_{#1}^{\mathsf{#2}}}
\newcommand{\advb}[3]{\mathsf{Adv}_{#1}^{\mathsf{#2} \mbox{-} \mathsf{#3}}}
\newcommand{\advc}[4]{\mathsf{Adv}_{#1}^{\mathsf{#2} \mbox{-} \mathsf{#3} \mbox{-} \mathsf{#4}}}
\newcommand{\advd}[5]{\mathsf{Adv}_{#1}^{\mathsf{#2} \mbox{-} \mathsf{#3} \mbox{-} \mathsf{#4} \mbox{-} \mathsf{#5}}}

\newcommand{\expt}[2]{\mathsf{Expt}_{#1}^{\mathsf{#2}}}
\newcommand{\expa}[2]{\mathsf{Expt}_{#1}^{\mathsf{#2}}}
\newcommand{\expb}[3]{\mathsf{Exp}_{#1}^{ \mathsf{#2} \mbox{-} \mathsf{#3}}}
\newcommand{\expc}[4]{\mathsf{Exp}_{#1}^{ \mathsf{#2} \mbox{-} \mathsf{#3} \mbox{-} \mathsf{#4}}}
\newcommand{\expd}[5]{\mathsf{Exp}_{#1}^{\mathsf{#2} \mbox{-} \mathsf{#3} \mbox{-} \mathsf{#4} \mbox{-} \mathsf{#5}}}

\newcommand{\hybi}[1]{\mathsf{Hyb}_{#1}}
\newcommand{\hybij}[2]{\mathsf{Hyb}_{#1}^{#2}}

\newcommand*{\sk}{\keys{sk}}
\newcommand*{\pk}{\keys{pk}}

\newcommand*{\msk}{\keys{msk}}
\newcommand*{\mpk}{\keys{mpk}}

\newcommand*{\MPK}{\keys{MPK}}

\newcommand*{\ck}{\keys{ck}}
\newcommand*{\ek}{\keys{ek}}

\newcommand*{\hk}{\keys{hk}}

\newcommand*{\ct}{\keys{ct}}

\newcommand*{\msg}{\keys{m}}

\newcommand*{\keys}[1]{\mathsf{#1}}
\newcommand*{\qstate}[1]{\mathpzc{#1}}
\newcommand*{\qreg}[1]{{\color{gray}{\mathsf{#1}}}}

\newcommand*{\algo}[1]{\ensuremath{\mathsf{#1}}}
\newcommand*{\qalgo}[1]{\ensuremath{\mathpzc{#1}}}

\newenvironment{boxfig}[2]{\begin{figure}[#1]\fbox{\begin{minipage}{0.97\linewidth}
                        \vspace{0.2em}
                        \makebox[0.025\linewidth]{}
                        \begin{minipage}{0.95\linewidth}
            {{
                        #2 }}
                        \end{minipage}
                        \vspace{0.2em}
                        \end{minipage}}}{\end{figure}}

\newcommand{\pprotocol}[4]{
\begin{boxfig}{t!}{\footnotesize 
\centering{\textbf{#1}}
    #4
} \caption{#2}
\label{#3}
\end{boxfig}
}

\newcommand{\protocol}[4]{
\pprotocol{#1}{#2}{#3}{#4} }
 \newcounter{expitem}[table]

\newcommand{\bit}{\{0,1\}}

\newcommand{\mat}[1]{\boldsymbol{#1}}

\newcommand{\Setup}{\algo{Setup}}

\newcommand{\KG}{\algo{KG}}
\newcommand{\Enc}{\algo{Enc}}
\newcommand{\Dec}{\algo{Dec}}

\newcommand{\qEnc}{\qalgo{Enc}}

\newcommand\FE{\algo{FE}}

\newcommand{\GC}{\algo{GC}}
\newcommand{\Garble}{\algo{Grbl}}

\newcommand{\PRG}{\algo{PRG}}

\newcommand{\Eval}{\algo{Eval}}

\newcommand{\qExtract}{\qalgo{Extract}}
\newcommand{\API}{\qalgo{API}}

\newcommand{\shiftdis}[1]{\Delta_{\mathsf{Shift}}^{#1}}

\newcommand{\projimp}{\algo{ProjImp}}

\newcommand{\BadExt}{\mathsf{BadExt}}

\newcommand{\Commit}{\algo{Com}}

\newcommand{\com}{\keys{com}}

\newcommand{\cind}{\stackrel{\mathsf{c}}{\approx}}

\newcommand{\negl}{{\mathsf{negl}}}

\newcommand{\poly}{{\mathrm{poly}}}

\newcommand{\zo}[1]{\{0,1\}^{#1}}
\newcommand{\bin}{\{0,1\}}

\newcommand{\xor}{\oplus}

\newcommand{\class}[1]{\mathsf{#1}}
\newcommand{\classPpoly}{\class{P}/\class{poly}}

\newcommand{\Ppoly}{\classPpoly}

\newcommand{\calC}{\mathcal{C}}

\newcommand{\calF}{\mathcal{F}}

\newcommand{\calP}{\mathcal{P}}

\newcommand{\SUC}{{\tt SUC}}

\newcommand{\Sim}{\algo{Sim}}

\newcommand{\SKE}{\algo{SKE}}

\newcommand{\E}{\algo{E}}
\newcommand{\D}{\algo{D}}

\newcommand{\lbl}{\mathsf{label}}

\newcommand{\Hash}{\algo{Hash}}

\newcommand{\state}{\mathsf{st}}

\newcommand{\authornote}[3]{\textcolor{#3}{[\textbf{#1:} {#2}]}}
\ifnum\authnote=1
\newcommand{\fuyuki}[1]{\authornote{fuyuki}{#1}{red}}
\newcommand{\ryo}[1]{\authornote{ryo}{#1}{darkblue}}
\else
\newcommand{\fuyuki}[1]{}
\newcommand{\ryo}[1]{}
\fi

\ifnum\llncs=1
\let\oldvec\vec% Store \vec in \oldvec
\let\vec\oldvec% Restore \vec from \oldvec
\makeatletter
\renewcommand*\l@author[2]{}
\renewcommand*\l@title[2]{}
\makeatletter
\fi    
\theoremstyle{remark}

\ifnum\submission=1
\title{
\textbf{One-out-of-Many Unclonable Cryptography:\\
Definitions, Constructions, and More}\thanks{{\color{red}{\emph{We attached the full version of this paper as a supplementary material}}}.}
}
\else
\title{
\textbf{One-out-of-Many Unclonable Cryptography:\\
Definitions, Constructions, and More}
}
\fi

\begin{document}
%\author{}
%\institute{}

\ifnum\anonymous=1
\author{\empty}
\ifnum\llncs=1
\institute{\empty}
\else
\fi
\else
%
%  For camera ready version.
%
\ifnum\llncs=1
\author{
	Fuyuki Kitagawa\inst{1} \and Ryo Nishimaki\inst{1}
}
\institute{
	NTT Social Informatics Laboratories
}
\else
%
%   For full/eprint version, etc.
%
\author[$\dagger$]{Fuyuki Kitagawa}
\author[$\dagger$]{Ryo Nishimaki}
\affil[$\dagger$]{{\small NTT Social Informatics Laboratories, Tokyo, Japan}\authorcr{\small \{fuyuki.kitagawa.yh,ryo.nishimaki.zk\}@hco.ntt.co.jp}}
\renewcommand\Authands{, }
\fi %%%%% END OF LNCS branch
\fi

\ifnum\llncs=1
\date{}
\else
\date{\today}
\fi

\maketitle

\begin{abstract}
 The no-cloning principle of quantum mechanics enables us to achieve amazing unclonable cryptographic primitives, which is impossible in classical cryptography.
However, the security definitions for unclonable cryptography are tricky.
Achieving desirable security notions for unclonability is a challenging task.
In particular, there is no indistinguishable-secure unclonable encryption and quantum copy-protection for single-bit output point functions in the standard model. To tackle this problem, we introduce and study relaxed but meaningful security notions for unclonable cryptography in this work. We call the new security notion \emph{one-out-of-many} unclonable security.

 We obtain the following results.
\begin{itemize}
\item We show that one-time strong anti-piracy secure secret key single-decryptor encryption (SDE) implies one-out-of-many indistinguishable-secure unclonable encryption.
\item We construct a one-time strong anti-piracy secure secret key SDE scheme in the standard model from the LWE assumption.
\item We construct one-out-of-many copy-protection for single-bit output point functions from one-out-of-many indistinguishable-secure unclonable encryption and the LWE assumption.
\item We construct one-out-of-many unclonable predicate encryption (PE) from one-out-of-many indistinguishable-secure unclonable encryption and the LWE assumption.
\end{itemize}
Thus, we obtain one-out-of-many indistinguishable-secure unclonable encryption, one-out-of-many copy-protection for single-bit output point functions, and one-out-of-many unclonable PE in the standard model from the LWE assumption.
In addition, our one-time SDE scheme is the first SDE scheme that does not rely on any oracle heuristics and strong assumptions such as indistinguishability obfuscation and witness encryption.
\end{abstract}

\ifnum\llncs=1
\else
\newpage
\setcounter{tocdepth}{2}
\tableofcontents

\newpage
\fi

\newcommand{\qaux}{\qstate{aux}}
\renewcommand{\share}{r}
\newcommand{\Can}[1]{\mathsf{Can}_{#1}}

\renewcommand{\msglen}{\ell}

%%% Main body %%%
\renewcommand{\Commit}{\mathsf{Commit}}
% !TEX root = main.tex

\newcommand{\subEvs}{\class{subEVS}}
\newcommand{\CandC}{\textrm{C\&C}}

\section{Introduction}\label{sec:intro}

\subsection{Background}\label{sec:background}
\paragraph{Unclonable encryption and quantum copy-protection.}
Quantum information enables us to achieve new cryptographic primitives beyond classical cryptography.
Especially the no-cloning principle of quantum information has given rise to amazing unclonable cryptographic primitives.
This includes quantum money~\cite{SIGACT:Wiesner83}, quantum copy-protection~\cite{CCC:Aaronson09}, unclonable encryption~\cite{TQC:BroLor20}, one-shot signatures \cite{STOC:AGKZ20}, single-decryptor encryption~\cite{EPRINT:GeoZha20,C:CLLZ21}, and many more.
In this work, we mainly focus on unclonable encryption and quantum copy-protection.

Broadbent and Lord~\cite{TQC:BroLor20} introduced unclonable encryption.
Unclonable encryption is a one-time secure secret key encryption where a plaintext is encoded into a quantum ciphertext that is impossible to clone.
More specifically, an unclonable encryption scheme encrypts a plaintext $\msg$ into a quantum ciphertext $\qct$.
The user who has the secret key can recover $\msg$ from $\qct$.
The security notion of unclonable encryption ensures that it is impossible to convert $\qct$ into possibly entangled bipartite states $\qct_1$ and $\qct_2$, both of which can be used to recover $\msg$ when the secret key is given.
Ananth and Kaleoglu~\cite{TCC:AnaKal21} later introduced unclonable public key encryption.
Unclonable encryption has interesting applications, such as preventing cloud storage attacks where an adversary steals ciphertexts from cloud storage with the hope that they can be decrypted if the secret key is leaked later.

Quantum copy-protection~\cite{CCC:Aaronson09} is a cryptographic primitive that prevents users from creating pirated copies of a program.
More specifically, a quantum copy-protection scheme transforms a classical program $C$ into a quantum program $\rho$ that is impossible to copy.
We can compute $C(x)$ for any input $x$ using $\rho$.
The security notion of copy-protection ensures that it is impossible to convert $\rho$ into possibly entangled bipartite states $\rho_1$ and $\rho_2$, both of which can be used to compute $C$.
As shown by Ananth and La Placa~\cite{EC:AnaLaP21}, it is impossible to have quantum copy-protection for general unlearnable functions.
For this reason, recent works have been studying quantum copy-protection for a simple class of functions such as point functions~\cite{ARXIV:ColMajPor20,TCC:AnaKal21,C:AKLLZ22,EPRINT:AnaKal22}.\footnote{
Some lines of works~\cite{C:CLLZ21,myTCC:LLQZ22} studied quantum copy-protection for cryptographic functionalities that are not captured by $\CandC$ programs.
Quantum copy-protections for cryptographic functionalities have different names, such as unclonable decryption or single decryptor encryption.
In this work, unless stated otherwise, we use the term quantum copy-protection to indicate quantum copy-protection for point functions.
For the previous works on quantum copy-protection for cryptographic functionalities, see \cref{sec:more_related_work}.
}
Moreover, Coladangelo, Majenz, and Poremba~\cite{ARXIV:ColMajPor20} show that quantum copy-protection for point functions can be transformed into quantum copy-protection for a more general class of compute-and-compare programs ($\CandC$ programs) that includes conjunctions with wildcards, affine testers, plaintext testers, and so on.
We focus on quantum copy-protection for point functions in this work.

\paragraph{Definition of unclonability: one-wayness and indistinguishability.}
To describe our research questions and contributions, we first explain a general template for unclonable security games played by a tuple of three adversaries $(\qA_0,\qA_1,\qA_2)$. The template is common to unclonable encryption and quantum copy-protection.
In the first stage, the challenger sends a challenge copy-protected object (such as a quantum ciphertext in unclonable encryption and copy-protected program in quantum copy-protection) to an adversary $\qA_0$.
Then, $\qA_0$ generates possibly entangled bipartite states $\rho_1$ and $\rho_2$ and sends $\rho_\alpha$ to $\qA_\alpha$ for $\alpha \in \setbk{1,2}$. In the second phase, the challenger sends extra information (such as secret keys in unclonable encryption and some inputs for the program in quantum copy-protection) to $\qA_1$ and $\qA_2$, and they try to compute target information about the copy-protected object (such as plaintexts in unclonable encryption and computation results on the given inputs in quantum copy-protection). Here, $\qA_1$ and $\qA_2$ are not allowed to communicate. If both $\qA_1$ and $\qA_2$ succeed in computing target information, the adversaries win. Note that if $\qA_\alpha$ has the original objects, computing target information is easy.

Using the above game, we can define both one-wayness-based notion and indistinguishability-based notion depending on which one the task of $\qA_1$ and $\qA_2$ is, recovering entire bits of high min-entropy information or detecting $1$-bit information.
Similarly to standard security notions, indistinguishability-based one is more general and enables us to have a wide range of applications.
Also, indistinguishability-based unclonability usually implies standard cryptographic security notions, but one-wayness-based unclonability does not necessarily imply them.\footnote{
For example, indistinguishability-based unclonability for unclonable encryption implies (one-time) IND-CPA security, but one-wayness-based unclonability does not.
}
In this work, we focus on indistinguishability-based unclonability notions.

\paragraph{Toward indistinguishability-based unclonability in the standard model.}
Unclonable encryption and quantum copy-protection have been studied actively and there are many constructions.
Although we have constructions with one-wayness-based unclonability from standard assumptions in the standard model~\cite{TQC:BroLor20,TCC:AnaKal21}, we have constructions with indistinguishability-based unclonability \emph{only in the oracle model}, in both unclonable encryption and quantum copy-protection~\cite{C:AKLLZ22}.
Ananth, Kaleoglu, Li, Liu, and Zhandry~\cite{C:AKLLZ22} proposed the only indistinguishability-based secure unclonable encryption and quantum copy-protection schemes.
Their proof technique is highly specific to the oracle model.
Thus, it still remains elusive to achieve unclonable encryption and quantum copy-protection with indistinguishability-based unclonability in the standard model.

Given the above situation, it is natural and reasonable to explore relaxed but meaningful indistinguishability-based unclonability and ask whether the notion can be achieved in the standard model.
Such a standard model construction with a relaxed notion would provide new insights toward achieving full-fledged indistinguishability-based unclonability in the standard model.

\subsection{Our Result}\label{sec:result}
Our contributions are proposing new definitions for unclonability and constructions satisfying them under the LWE assumption in the standard model.
\paragraph{New definitions: one-out-of-many unclonability notions.}
We introduce a relaxed indistinguishability-based unclonability for unclonable encryption and copy-protection, called \emph{one-out-of-many} unclonable security.
This notion captures a meaningful unclonability, as we argue below.
It guarantees that no adversary can generate $\numa$ copies with probability significantly better than $\frac{1}{\numa}$ for any $\numa$.
Thus, roughly speaking, it guarantees that the expected number of successful target objects generated by any copying adversary is less than $1$.
We define one-out-of-many unclonability by extending the unclonable game played by a tuple of three adversaries into a game played by $\numa+1$ adversaries, where $\numa \ge 2$ is arbitrary.

Although one-out-of-many unclonable security looks weaker than existing unclonable security, it is useful in some applications. For example, suppose we publish many quantum objects, say $\ell$ objects. Then, one-out-of-many security guarantees that no matter what copying attacks are applied to those objects, there are expected to be only $\ell$ objects on average in this world.
Another nice property of one-out-of-many security is that it implies standard cryptographic security notions.
For example, one-out-of-many unclonability for unclonable encryption implies (one-time) IND-CPA security.
This result contrasts one-wayness-based unclonability notions that do not necessarily imply standard indistinguishability notions.

\paragraph{Unclonable encryption in the standard model via single decryptor encryption.} 
We provide unclonable encryption satisfying one-out-of-many unclonability under the LWE assumption in the standard model.
We obtain this result as follows.

We first define one-out-of-many unclonability for (one-time) single decryptor encryption (SDE)~\cite{EPRINT:GeoZha20,C:CLLZ21}.
One-time SDE is a dual of unclonable encryption in the sense that one-time SDE is a one-time secret key encryption scheme where a secret key is encoded into a quantum state, and its security notion guarantees that any adversary cannot copy the quantum secret key.
Under appropriate definitions, it is possible to back and forth between unclonable encryption and one-time SDE, as shown by Georgiou and Zhandry~\cite{EPRINT:GeoZha20}.
We show that we can transform any one-time SDE with one-out-of-many unclonability to unclonable encryption with one-out-of-many unclonability.

We then show that we can obtain one-time SDE with one-out-of-many unclonability from the LWE assumption.
More specifically, assuming the LWE assumption, we construct one-time SDE satisfying strong anti-piracy introduced by Coladangelo, Liu, Liu, and Zhandry~\cite{C:CLLZ21}, and show that strong anti-piracy implies one-out-of-many unclonability.
Combining this result with the above transformation, we obtain unclonable encryption with one-out-of-many security under the LWE assumption.
\begin{theorem}[informal]
Assuming the LWE assumption holds, there exists strong anti-piracy secure one-time SDE.
\end{theorem}
\begin{theorem}[informal]
Assuming the LWE assumption holds, there exists one-out-of-many indistinguishable-secure unclonable encryption.
\end{theorem}
To achieve one-time SDE satisfying strong anti-piracy, we develop a technique enabling us to use a BB84~\cite{TCS:BenBra14} state as a copy-protected secret key. Our crucial tool is single-key ciphertext-policy functional encryption (CPFE) with a succinct key, which we introduce in this work. We instantiate it with hash encryption (HE)~\cite{PKC:DGHM18} implied by the LWE assumption. The technique of the post-quantum watermarking by Kitagawa and Nishimaki~\cite{EC:KitNis22} inspired our proof technique.
We emphasize that \emph{our one-time SDE scheme is the first SDE scheme that does not require either oracle heuristic or strong assumptions such as indistinguishability obfuscation and witness encryption.}

\paragraph{Quantum copy-protection in the standard model via unclonable encryption.}

We propose quantum copy protection for single-bit output point functions based on unclonable encryption and the LWE assumption.
Known constructions from unclonable encryption~\cite{ARXIV:ColMajPor20,TCC:AnaKal21} support only multi-bit output point functions.\footnote{
One might think that copy protection for multi-bit output point functions implies that for single-bit output point functions.
However, this is not the case.
This is because the security of copy protection for multi-bit output point functions usually relies on the high min-entropy of the multi-bit output string, and it is broken if the output string does not have enough entropy as in the case of single-bit output.
Realizing copy protection for single-bit output point function is challenging in the sense that we have to achieve the security without relying on the entropy of the output string.
}
Although we formally prove this result with our new one-out-of-many security notion, our construction also works under standard indistinguishability-based unclonability definitions defined using three adversaries $(\qA_0,\qA_1,\qA_2)$.
\begin{theorem}[informal]
Assuming (resp. one-out-of-many) indistinguishable-secure unclonable encryption and the LWE assumption holds, there exits (resp. one-out-of-many) copy-protection for single-bit output point functions.
\end{theorem}

\paragraph{Unclonable predicate encryption.}
Using the technique proposed by Ananth and Kaleoglu~\cite{TCC:AnaKal21}, we can convert our one-out-of-many secure unclonable encryption into one-out-of-many secure unclonable public-key encryption.
We construct a one-out-of-many unclonable predicate encryption (PE) scheme from one-out-of-many unclonable encryption and the LWE assumption.
PE is a stronger variant of attribute-based encryption (ABE).
ABE is an advanced public-key encryption system where we can generate a user secret key for an attribute $x$ and a ciphertext of a message $\msg$ under a policy $P$.\footnote{We focus on ciphertext-policy ABE in this work.}
We can decrypt a ciphertext and obtain $\msg$ if $P(x)=1$. In PE, ciphertexts hide not only plaintexts but also policies.
\begin{theorem}[informal]
Assuming (resp. one-out-of-many) indistinguishable-secure unclonable encryption and the LWE assumption holds, there exits (resp. one-out-of-many) unclonable PE.
\end{theorem}

\subsection{Concurrent and Independent Work}
Ananth, Kaleoglu, and Liu~\cite{ARXIV:AnaKalLiu23} introduce a new framework called cloning games to study unclonable cryptography. They obtain many implications to unclonable cryptography thanks to the framework. In particular, they obtain information-theoretically secure one-time SDE in the standard model. The scheme is \emph{single-bit} encryption while our computationally secure scheme is multi-bit encryption. \emph{Note that we do not know how to obtain multi-bit encryption from single-bit one via parallel repetition and the standard hybrid argument in unclonable cryptography.} 
Thus, their results on SDE is incomparable with ours.
We also note that it is not clear whether we can obtain one-out-of-many indistinguishable-secure unclonable encryption from their SDE, since it seems that we need strong anti-piracy secure one-time SDE to obtain one-out-of-many secure one, but their scheme is only proved to satisfy (non-strong) indistinguishability-based security.
For the detailed reason on the need of strong anti-piracy for the implication, see the ``SDE from LWE'' paragraph in \cref{sec:tech_overview}.

% !TEX root = main.tex

\subsection{Technical Overview}\label{sec:tech_overview}
We provide a high-level overview of our techniques in this subsection.
%Note that when we refer to SDE in this overview, it is one-time secure secret key SDE unless stated otherwise.
In this paper, standard math or sans serif font stands for classical algorithms and classical variables. The calligraphic font stands for quantum algorithms and the calligraphic font and/or the bracket notation for (mixed) quantum states.

\paragraph{Relaxed definition of unclonable cryptography.}
%One of our goals is achieving unclonable encryption satisfying indistinguishability-based unclonability (we call unclonable-indistinguishable security) in the standard model.
%To this end, we propose a different approach from that by Ananth et al.~\cite{C:AKLLZ22} who obtained unclonable indistinguishable secure unclonable encryption using the quantum random oracle.
We introduce relaxed security notions for unclonable cryptography called \emph{one-out-of-many} unclonable security that roughly guarantees that no adversary can generate $\numa$ copies with a probability significantly better than $\frac{1}{\numa}$ for any $\numa$.
The one-out-of-many unclonability is defined by extending the unclonable game played by $(\qA_0,\qA_1,\qA_2)$ explained in~\cref{sec:background}.
The one-out-of-many unclonability game is played by a tuple of $\numa+1$ adversaries $(\qA_0,\qA_1,\cdots,\qA_\numa)$, where $\numa$ is arbitrary.
At the first stage of the game, $\qA_0$ is given a single quantum object, generates possibly entangled $\numa$-partite states $\rho_1,\ldots,\rho_\numa$, and sends $\rho_k$ to $\qA_k$ for $k\in \setbk{1,\ldots,\numa}$. At the second stage, the challenger selects one of $(\qA_1,\ldots,\qA_\numa)$ by a random $\alpha\chosen \setbk{1,\ldots, \numa}$ and sends additional information \emph{only to $\qA_\alpha$}, and only $\qA_\alpha$ tries to detect the target $1$-bit information.
Recall that we focus on indistinguishability-based setting.
The one-out-of-many unclonability guarantees that the adversary cannot win this game with probability significantly better than the trivial winning probability $\frac{1}{\numa}\cdot 1+\frac{\numa-1}{\numa}\cdot\frac{1}{2}=\frac{1}{2}+\frac{1}{2\numa}$.
The definition captures the above intuition because if $\qA_0$ could make $\numa$ copies with probability $\frac{1}{\numa}+\delta$ for some noticeable $\delta$, the adversary would win the game with probability at least $(\frac{1}{\numa}+\delta)\cdot 1+(1-\frac{1}{\numa}-\delta)\cdot\frac{1}{2}=\frac{1}{2}+\frac{1}{2\numa}+\frac{\delta}{2}$.

We consider one-out-of-many unclonable security for the following security notions in this work:
(1) (one-time) unclonable-indistinguishable security for unclonable encryption,
(2) copy-protection security for single-bit output point functions,
(3) (one-time) indistinguishability-based security for one-time SDE,
(4) unclonable-simulation security for PE, which is introduced in this work.

The nice property of one-out-of-many security are as follows. 
One-out-of-many security implies standard cryptographic security notions.
For example, one-out-of-many unclonability for unclonable encryption implies (one-time) IND-CPA security, and one for copy-protection implies distributional indistinguishability as virtual black-box obfuscation.
Moreover, we can use one-out-of-many secure unclonable cryptographic primitives as a drop-in-replacement of standard indistinguishability-based unclonable cryptographic primitives if our goal is constructing a one-out-of-many secure unclonable cryptographic primitive (and vice versa).
For example, the transformation from unclonable encryption to unclonable public-key encryption proposed by Ananth and Kaleoglu~\cite{TCC:AnaKal21} works also in the one-out-of-many setting.
Moreover, all of the generic constructions from an unclonable primitive to another unclonable primitive that we propose work in both standard (three adversary style) setting and one-out-of-many setting.

In addition to the above nice properties, we can prove that \emph{one-out-of-many indistinguishability-based secure one-time SDE is equivalent to one-out-of-many unclonable-indistinguishable secure unclonable encryption}.
In this work, we first obtain one-out-of-many indistinguishability-based secure one-time SDE from the LWE assumption, and using this equivalence, we obtain one-out-of-many secure unclonable encryption, copy protection for single-bit output point functions, and unclonable PE.

Before our work, Georgiou and Zhandry~\cite{EPRINT:GeoZha20} showed a transformation from one-time SDE to unclonable encryption under the standard three adversary style setting.
Informed readers may think that by combining their result with the result by Coladangelo et al.~\cite{C:CLLZ21} (and the result by Culf and Vidick~\cite{Quantum:CulVid22}), we can obtain indistinguishability-based unclonable encryption in the standard model based on indistinguishability obfuscation.
However, this is not the case due to the fact that those two works used different definitions of the indistinguishability-based security for SDE, and we do not know any relation between them.
For more details, see \cref{rem:SDE_def}.

\paragraph{SDE from LWE.}
We next explain how to obtain one-out-of-many indistinguishability-based secure one-time SDE.
In fact, we obtain one-time SDE satisfying much stronger security notion called strong anti-piracy~\cite{C:CLLZ21} from the LWE assumption, and prove that \emph{strong anti-piracy implies one-out-of-many indistinguishability-based security}.
%Hence, if we achieve strong anti-piracy secure SDE in the standard model from the LWE assumption, we obtain one-out-of-many unclonable-indistinguishable secure unclonable encryption.
% simulation-based copy-protection for $\CandC$ programs and one-out-of-many secure unclonable PE in the standard model.
Below, we first introduce the definition of strong anti-piracy for one-time SDE, briefly explain the intuition of the implication, and finally present the high level ideas on how to realize strong anti-piracy secure one-time SDE from the LWE assumption.

Recall the general template of the security game for unclonability played by three adversaries $(\qA_0,\qA_1,\qA_2)$ explained in \cref{sec:background}.
This template also captures the security game of strong anti-piracy for SDE.
In strong anti-piracy security for SDE, $\qA_0$ receives a copy-protected decryption key $\qdk$ in the first stage.
In the second stage, $\qA_1$ and $\qA_2$ outputs quantum decryptors $\qD_1$ and $\qD_2$, respectively, and the challenger tests whether both $\qD_1$ and $\qD_2$ are ``good'' (or ``live'') quantum decryptors~\cite{TCC:Zhandry20,C:CLLZ21}. Intuitively, good quantum decryptors can distinguish encryption of $m_0$ from that of $m_1$ with probability $\frac{1}{2}+\frac{1}{\poly(\secp)}$, and do not lose the decryption capability even after its goodness was tested. Strong anti-piracy guarantees that the probability that both $\qD_1$ and $\qD_2$ are tested as good is negligible. See~\cref{def:otsanti_piracy} for the precise definition.

The intuition behind the implication from strong anti-piracy to one-out-of-many security is as follows.
The one-out-of-many security game for a one-time SDE scheme played by a tuple of $\numa+1$ adversaries $(\qA_0,\cdots,\qA_\numa)$ is defined as follows.
The first stage adversary $\qA_0$ is given a quantum decryption key $\qdk$, generates possibly entangled $\numa$-partite states $\rho_1,\cdots\rho_\numa$, and sends $\rho_k$ to the second stage adversary $\qA_k$ for every $k\in\{1,\cdots,\numa\}$.
In the second stage, only randomly chosen single second stage adversary $\qA_\alpha$ is given the challenge ciphertext and required to guess the challenge bit.
The $\numa$-partite state $(\rho_1,\cdots,\rho_\numa)$ generated by $\qA_0$ can be regarded as a tuple of $\numa$ quantum decryptors.
If the one-time SDE scheme is strong anti-piracy secure, all $\numa$ quantum decryptors \emph{except one} must have success probabilities at most $1/2 +\negl(\secp)$. Hence, the success probability of $(\qA_0,\cdots,\qA_{\numa+1})$ in the one-out-of-many security game is at most $1/n \cdot 1 + (n-1)/n \cdot (1/2+\negl(\secp))=1/2+1/2n +\negl(\secp)$, which proves the one-out-of-many security. 
It seems that ``strong'' anti piracy is required for this argument and it is difficult to prove a similar implication from (non-strong) indistinguishability-based security defined by Coladangelo et al.~\cite{C:CLLZ21}.
To formally prove the implication, we have to construct a reduction algorithm that finds two ``good'' decryptors from $\numa$ decryptors output by $\qA$.
If the reduction attacks strong anti-piracy, it is sufficient to randomly pick two decryptors out of $\numa$ since the reduction's goal is to output two ``good'' decryptors with inverse polynomial probability.
However, if the reduction attacks (non-strong) indistinguishability-based security, the reduction cannot use such random guessing and needs to detect whether each decryptor is ``good'' since the reduction's goal is to make a distinguishing gap.
We are considering the one-time setting where the adversaries are not given the encryption key. Thus, it seems difficult to perform such detection of ``good'' decryptors.
%The issue is similar to that prevents us from proving the implication from (non-strong) indistinguishability-based security to one-wayness-based security pointed out by Coladangelo et al.~\cite{C:CLLZ21}.

We next explain how to achieve strong anti-piracy secure one-time SDE based on the LWE assumption in the standard model.
We use the monogamy of entanglement property of BB84 states~\cite{NJP:TFKW13} differently from the previous work on SDE~\cite{C:CLLZ21} that used the monogamy of entanglement property of coset states.

We combine BB84 states and ciphertext-policy FE with succinct key to achieve strong anti-piracy.
We first explain the definition of single-key CPFE.
A single-key CPFE scheme $\CPFE$ consists of three algorithms $(\FE.\Setup,\FE.\Enc,\FE.\Dec)$.
$\FE.\Setup$ takes as input a string $x$ and outputs a public key $\pk$ and a decryption key $\sk_x$.\footnote{We omit the security parameter for simplicity in this overview. The same is applied to other cryptographic primitives.} Here, we assume that $x$ itself works as a decryption key $\sk_x$ for $x$, thus $\FE.\Setup$ outputs only $\pk$. We can achieve such a CPFE (we will explain later).
$\FE.\Enc$ takes as input $\pk$ and a circuit $C$, and outputs a ciphertext $\ct$.
We can decrypt $\ct$ with $\sk_x$ using $\FE.\Dec$, and obtain $C(x)$.
The single-key security of CPFE guarantees that $\FE.\Enc(\pk,C_0)$ and $\FE.\Enc(\pk,C_1)$ are computationally indistinguishable for an adversary who has a decryption key $\sk_x$ for $x$ of its choice as long as $C_0(x)=C_1(x)$ holds.

Let $s[i]$ is the $i$-th bit of a string $s\in\zo{n}$.
Our one-time SDE scheme is as follows. The key generation algorithm generate a BB84 state $\ket{x^\theta} \seteq H^{\theta[1]}\ket{x[1]} \tensor \cdots \tensor H^{\theta[n]}\ket{x[n]}$, where $H$ is the Hadamard gate, and a public key $\pk \gets \FE.\Setup(x)$ of $\CPFE$. It outputs an encryption key $\ek \seteq (\theta,\pk)$ and decryption key $\qdk \seteq \ket{x^\theta}$. Note that although our one-time SDE scheme is secret key encryption, an encryption key and a decryption key are different.
The encryption algorithm takes as input the encryption key and a plaintext $\msg$, and generates a ciphertext $\fe.\ct \gets \FE.\Enc(\pk,C[\msg])$, where $C[\msg]$ is a constant circuit that outputs $\msg$ for all inputs. It outputs a ciphertext $(\theta,\fe.\ct)$.
We can decrypt $\fe.\ct$ and obtain $\msg$ by recovering $x$ from $\ket{x^\theta}$ and $\theta$ since $x$ works as a decryption key of $\CPFE$ as we assumed.
Intuitively, it is hard to copy $\qdk = \ket{x^\theta}$ by the monogamy of entanglement property.
The monogamy of entanglement property can be explained by the template of unclonable cryptography.
In the first stage, $\qA_0$ is given $\ket{x^\theta}$. In the second stage, $\qA_1$ and $\qA_2$ receive $\theta$ and try to output $x$.
It is proved that the winning probability of the adversaries is exponentially small without any assumptions~\cite{NJP:TFKW13}.

To prove the strong anti-piracy security, we need to extract $x$ both from good decryptors $\qD_1$ and $\qD_2$ respectively output by $\qA_1$ and $\qA_2$ to reduce the SDE security to the monogamy of entanglement property.
The idea for the extraction is as follows.
Let $\tlC[b,\msg_0,\msg_1,i]$ be a circuit that takes as input $x$ and outputs $\msg_{b \xor x[i]}$.
We estimate the probability that a good decryptor outputs the correct $b$ when we feed $\FE.\Enc(\pk,\tlC[b,\msg_0,\msg_1,i])$ to it.
The security of $\CPFE$ guarantees that $\FE.\Enc(\pk,\tlC[b,\msg_0,\msg_1,i])$ is indistinguishable from $\FE.\Enc(\pk,C[\msg_{b \xor x[i]}])$ since $\tlC[b,\msg_0,\msg_1,i](x)=\msg_{b\xor x[i]}=C[\msg_{b\xor x[i]}](x)$.
Hence, we can analyze the probability as follows.
\begin{itemize}
\item If $x[i]=0$, the distinguishing probability should be greater than $\frac{1}{2}$ since a good decryptor receives $\FE.\Enc(\pk,C[\msg_b])$ in its view and correctly guesses $b$ with probability $\frac{1}{2}+\frac{1}{\poly(\secp)}$.
\item If $x[i]=1$, the distinguishing probability should be smaller than $\frac{1}{2}$ since a good decryptor receives $\FE.\Enc(\pk,C[\msg_{1\xor b}])$ in its view and outputs the flipped bit $1\xor b$ with probability $\frac{1}{2}+\frac{1}{\poly(\secp)}$.
\end{itemize}
This means that we can decide $x[i]=0$ or $x[i]=1$ by estimating the success probability of a good decryptor that receives $\FE.\Enc(\pk,\tlC[b,\msg_0,\msg_1,i])$. Thus, we can extract $x$ from good decryptors.
This extraction technique is based on the post-quantum watermarking extraction technique by Kitagawa and Nishimaki~\cite{EC:KitNis22}.
Hence, the extraction succeeds without collapsing good quantum decryptors $\qD_1$ and $\qD_2$. See~\cref{sec:ot_SDE_HE} for the detail.

There is one subtle issue in the argument above. Since $\pk$ depends on $x$, we need leakage information about $x$ to simulate $\pk$ in the reduction. More specifically, let $\Leak(\cdot)$ be a leakage function and the reduction needs $\Leak(x)$ to simulate $\pk$ of $\CPFE$. We can consider such a leakage variant of the monogamy of entanglement game, where $\qA_0$ receives $\ket{x^\theta}$ and $\Leak(x)$ in the first stage. The variant holds if $\abs{\Leak(x)}=\secp$ that is short enough compared to $n = \abs{x}$ since we can simply guess $\Leak(x)$ with probability $\frac{1}{2^\secp}$. Although the bound is degraded to $\frac{2^\secp}{\exp(n)}$, it is still negligible by setting $n$ appropriately. Hence, we use single-key CPFE with \emph{succinct key} to ensure that $\Leak(x)$ does not have much information about $x$.

A single-key CPFE scheme has a succinct key if it satisfies the following properties.
$\FE.\Setup$ consists of two algorithms $(\HKG,\Hash)$, computes a hash key $\hk \gets \HKG(1^{\abs{x}})$ and a hash value $h \gets \Hash(\hk,x)$ from $x$, and outputs a public key $\pk \seteq (\hk,h)$ and a decryption key $\sk_x \seteq x$. The length of $h$ should be the same as the security parameter (no matter how large $x$ is).
These properties are crucial for our construction since we consider $\Leak(x)\seteq \Hash(\hk,\cdot)$ and a hash value $h$ does not have much information about $x$.

We can achieve single-key CPFE with succinct key from hash encryption (HE)~\cite{PKC:DGHM18}, which can be achieved from the LWE. We use HE instead of plain PKE in the well-known single-key FE scheme based on PKE~\cite{CCS:SahSey10}. Thanks to the compression property of hash encryption, we can achieve the succinct key property. A decryption key of HE is a pre-image of a hash. Hence, we can use $x$ as a decryption key $\sk_x$.

\paragraph{One-out-of-many unclonable encryption.}
Georgiou and Zhandry~\cite{EPRINT:GeoZha20} showed that under appropriate definitions, it is possible to transform one-time SDE to unclonable encryption.
We show that by using the same transformation, we can transform any one-out-of-many secure SDE to one-out-of-many secure unclonable encryption.
By combining this transformation with the above one-out-of-many SDE, we can obtain one-out-of-many secure unclonable encryption based on the LWE assumption.

\paragraph{Copy protection for single-bit output point functions.}
A point function $f_{y,\msg}$ is a function that outputs $\msg$ on input $y$ and outputs $0^\abs{\msg}$ otherwise.
When we say single-bit output, we set $\msg=1$.
When we say multi-bit output, we set $\msg$ as a multi-bit string sampled from some high min-entropy distribution.
We denote the family of single-bit output point functions and multi-bit output point functions as $\PFs$ and $\PFm$, respectively.

We introduce a simplified security game for copy protection for point functions.
It follows the template given in \cref{sec:background}.
The first stage adversary is given a copy protected program $\rho$ of a randomly generated point function $f_{y,\msg}$.
In the second stage, $\qA_1$ and $\qA_2$ are given a challenge input $x$ sampled from some distribution and try to output $f_{y,\msg}(x)$ simultaneously.
The copy protection security guarantees that the success probability of the adversary is bounded by the trivial winning probability.

Coladangelo et al.~\cite{ARXIV:ColMajPor20} proposed a generic construction of copy protection for $\PFm$ using unclonable encryption.\footnote{
 Ananth and Kaleoglu~\cite{TCC:AnaKal21} also proposed a similar construction.
}
The construction is as follows.
To copy protect a multi-bit output point function $f_{y,\msg}$, it generates a quantum ciphertext of $\msg$ of an unclonable encryption scheme $\UE$ under the key $y$, that is $\qct\la\UE.\qEnc(y,\msg)$.
For simplicity, we assume that given a key $y^\prime$ and a ciphertext $\qct^\prime$ of $\UE$, we can efficiently check whether $\qct^\prime$ is generated under the key $y^\prime$ or not.
Then, to evaluate this copy protected program with input $x$, it first checks if $x$ and $\qct$ match or not, and if so, just output the decryption result of $\qct$ under the key $x$.
We see that the construction satisfies the correctness.
Coladangelo et al. also show that if $\UE$ satisfies one-wayness-based unclonability, the construction satisfies copy protection security.

In this work, we propose a generic construction of copy protection for $\PFs$ using unclonable-indistinguishable secure unclonable encryption.
The above simple construction by Coladangelo et al. does not work if our goal is copy protection for $\PFs$, even if the underlying unclonable encryption is unclonable-indistinguishable secure.
The above construction crucially relies on the fact that $\msg$ is sampled from high min-entropy distribution in $\PFm$.
In fact, if $\msg$ is fixed as the case of $\PFs$, the construction is completely insecure under the above condition that we can efficiently check the correspondence between a key and a ciphertext of $\UE$, which is required to achieve correctness.

To fix this issue, our construction uses quantum FHE~\cite{FOCS:Mahadev18b} and obfuscation for $\CandC$ programs~\cite{FOCS:WicZir17,FOCS:GoyKopWat17}, both of which can be realized from the LWE assumption.
Roughly speaking, in our construction, the above UE-based copy protected program is encrypted by QFHE.
The evaluation of the new copy protected program is done by the homomorphic evaluation of QFHE, and we obtain the evaluation result from the QFHE ciphertext by using decryption circuit of QFHE obfuscated by obfuscation for $\CandC$ programs.
Our construction works in both standard three adversary style setting explained above and one-out-of-many setting.

\paragraph{Unclonable PE.}
We also define and construct unclonable PE.
Our security definition of unclonable PE is simulation-based.
It also can be seen as an extension of simulation-based security notion for (not unclonable) PE defined by Gorbunov et al.~\cite{C:GorVaiWee15}.

Our construction of unclonable PE is an extension of ABE-to-PE transformation based on obfuscation for $\CandC$ programs proposed in classical cryptography~\cite{FOCS:WicZir17,FOCS:GoyKopWat17}.
The above construction of copy protection for $\PFs$ can be extended to copy protection for $\CandC$ programs by encrypting a $\CandC$ program together with the ciphertext of unclonable encryption into the QFHE ciphertext.
At a high level, we show that by replacing obfuscation for $\CandC$ programs with this copy protection for $\CandC$ programs in the ABE-to-PE transformation, we can obtain unclonable PE.
To achieve hiding of policies in PE, the construction crucially uses the security notions of the underlying QFHE and obfuscation for $\CandC$ programs.
Thus, we do not use the abstraction of copy protection for $\CandC$ programs, and present our construction directly using ABE and the building blocks of our copy protection construction.
Our construction works in both standard three adversary style setting and one-out-of-many setting.

In the above transformation, we use simulation-based secure ABE instead on indistinguishability-based one.
As far as we know, simulation-based secure ABE was not studied before and there is no existing construction.
Thus, we construct simulation-based secure ABE by ourselves.
The construction is based on indistinguishability-based secure ABE and obfuscation for $\CandC$ programs, both of which can be based on the LWE assumption.
Interestingly, we can also use our simulation-based secure ABE to convert unclonable encryption into the first unclonable ABE via the standard KEM-DEM framework.

%\paragraph{Secure software leasing.}
%Our simulation-based security for secure software leasing is a simple adaptation of our simulation-based security for copy-protection.
%The difference is that there is an explicit returning process of leased programs.
%Our construction of simulation-based secure software leasing is almost the same as our construction of simulation-based secure copy-protection.
%The only difference is that we use encryption with certified deletion~\cite{TCC:BroIsl20} instead of unclonable encryption since it is a secure leasing version of unclonable encryption. 
%Moreover, we can achieve PE with certified deletion by replacing copy-protection for $\CandC$ with secure software leasing for $\CandC$ in our unclonable PE construction.

\ifnum\submission=1

\subsection{Organization}
In \cref{sec:one-out-of-many}, we introduce the notion of one-out-of-many unclonable security and prove some implications.
In \cref{sec:ot_SDE_HE}, we show how to construct secret key one-time SDE based on the LWE assumption.
In \cref{sec:CP}, we show how to realize one-out-of-many copy protection for single-bit output point functions. 
In \cref{sec:UPE}, we show how to realize one-out-of-many unclonable PE.

Due to the space limitation, we present more related works in~\cref{sec:more_related_work} and provide preliminaries, including notations in \cref{sec:prelim}.
Also, some security proofs are presented in the supplemental materials.

\else

% !TEX root = main.tex

\ifnum\submission=0
\subsection{More on Related Work}\label{sec:more_related_work}
\else
\section{More on Related Work}\label{sec:more_related_work}
\fi

\paragraph{Copy-protection for $\CandC$ programs.}
Aaronson proposed candidate constructions of copy-protection for point functions~\cite{CCC:Aaronson09}. However, he did not provide reduction-based proofs.
Coladangelo, Majenz, and Poremba proposed copy-protection for $\CandC$ programs in the QROM and copy-protection for multi-bit output point functions based on one-way-secure unclonable encryption~\cite{ARXIV:ColMajPor20}. They also show that we can convert copy-protection for point functions into copy-protection for $\CandC$ programs. Ananth and Kaleoglu proposed copy-protected point functions based on indistinguishable-secure unclonable encryption~\cite{TCC:AnaKal21}. Ananth et al.~\cite{C:AKLLZ22} proposed indistinguishable-secure unclonable encryption and copy-protection for single-bit output point functions in the QROM.
Ananth and La Placa~\cite{EC:AnaLaP21} show that there exists a class of functions that we cannot achieve copy-protection in the plain model.
Ananth and Kaleoglu~\cite{EPRINT:AnaKal22} extend the impossibility result by Ananth and La Placa~\cite{EC:AnaLaP21} and show that there exists a class of functions that we cannot achieve copy-protection in the classical-accessible random oracle model (CAROM). CAROM is a model where both constructions and adversaries can only classically access the random oracle.

\paragraph{Unclonable encryption.}

Broadbent and Lord~\cite{TQC:BroLor20} proposed the notion of unclonable encryption based on the idea by Gottesman~\cite{QIC:Gottesman03}.\footnote{The notion of unclonable encryption by Gottesman is slightly diffrent from the one in this paper. His definition focuses on tamper detection.} They considered two security definitions for unclonable encryption. One is one-wayness against cloning attacks (one-way-secure unclonable encryption) and they achieve information-theoretic one-wayness by using BB84 states. The other is indistinguishability against cloning attacks (indistinguishable-secure unclonable encryption). However, they did not achieve it. They constructed indistinguishable-secure unclonable encryption only in a very restricted model by using PRFs. Ananth and Kaleoglu~\cite{TCC:AnaKal21} proposed a transformation from unclonable encryption to public key unclonable encryption. Ananth et al.~\cite{C:AKLLZ22} proposed the first indistinguishable-secure unclonable encryption in the QROM.

\paragraph{Unclonable decryption.}

Georgiou and Zhandry~\cite{EPRINT:GeoZha20} proposed the notion of SDE and show the equivalence between indistinguishable-secure unclonable encryption and their SDE.\footnote{Selectively secure secret key SDE in the setting of honestly generated keys. See~\cite{EPRINT:GeoZha20} for the detail.}, Coladangelo et al.~\cite{C:CLLZ21} proposed new definitions of SDE and constructed a \emph{public key} SDE scheme that satisfies their definitions from IO and the LWE assumption. Although they needed the strong monogamy of entanglement property conjecture for their constructions, the conjecture was proved without any assumptions by Culf and Vidick~\cite{Quantum:CulVid22}. It is unclear whether SDE under the definitions by Coladangelo et al.~\cite{C:CLLZ21} is equivalent to unclonable encryption. Liu, Liu, Qian, and Zhandry~\cite{myTCC:LLQZ22} achieved bounded collusion-resistant public key SDE, where adversaries can receive many copy-protected decryption keys, from IO and the LWE assumption. They also consider bounded collusion-resistant copy-protection for PRFs and signatures.
Sattath and Wyborski~\cite{ARXIV:SatWyb22} also extend SDE to unclonable decryptors, where we can generate multiple copy-protected decryption keys from a classical decryption key. They constructed a secret key unclonable decryptors scheme from copy-protection for balanced binary functions. However, they need IO or a quantum oracle to instantiate copy-protection for balanced binary functions.

\fi
\ifnum\submission=0
% !TEX root = main.tex

\section{Preliminaries}\label{sec:prelim}

% !TEX root = main.tex

\paragraph{Notations and conventions.}
% $\qA^{\ket{F(\cdot)}}$.
In this paper, standard math or sans serif font stands for classical algorithms (e.g., $C$ or $\algo{Gen}$) and classical variables (e.g., $x$ or $\keys{pk}$).
Calligraphic font stands for quantum algorithms (e.g., $\qalgo{Gen}$) and calligraphic font and/or the bracket notation for (mixed) quantum states (e.g., $\qstateq$ or $\ket{\psi}$).
For strings $x$ and $y$, $x \concat y$ denotes the concatenation of $x$ and $y$.
%%Let $\mv{0}$ denote a string consisting of an appropriate number of $0$.
Let $[\ell]$ denote the set of integers $\{1, \cdots, \ell \}$, $\secp$ denote a security parameter, and $y \seteq z$ denote that $y$ is set, defined, or substituted by $z$.

In this paper, for a finite set $X$ and a distribution $D$, $x \chosen X$ denotes selecting an element from $X$ uniformly at random, $x \chosen D$ denotes sampling an element $x$ according to $D$. Let $y \gets \algo{A}(x)$ and $y \gets \qalgo{A}(\qstate{x})$ denote assigning to $y$ the output of a probabilistic or deterministic algorithm $\algo{A}$ and a quantum algorithm $\qalgo{A}$ on an input $x$ and $\qstate{x}$, respectively. When we explicitly show that $\algo{A}$ uses randomness $r$, we write $y \gets \algo{A}(x;r)$.
PPT and QPT algorithms stand for probabilistic polynomial-time algorithms and polynomial-time quantum algorithms, respectively.
Let $\negl$ denote a negligible function. Let $\cind$ denote computational indistinguishability.

% !TEX root = main.tex
\subsection{Quantum information}

We review several quantum information concepts.

\paragraph{Basics.}
Let $\cH$ be a finite-dimensional complex Hilbert space. A (pure) quantum state is a vector $\ket{\psi}\in \cH$.
Let $\cS(\cH)$ be the space of Hermitian operators on $\cH$. A density matrix is a Hermitian operator $\qstate{X} \in \cS(\cH)$ with $\Trace(\qstate{X})=1$, which is a probabilistic mixture of pure states.
A quantum state over $\cH=\bbC^2$ is called qubit, which can be represented by the linear combination of the standard basis $\setbk{\ket{0},\ket{1}}$. More generally, a quantum system over $(\bbC^2)^{\tensor n}$ is called an $n$-qubit quantum system for $n \in \bbN \setminus \setbk{0}$.

A Hilbert space is divided into registers $\cH= \cH^{\qreg{R}_1} \tensor \cH^{\qreg{R}_2} \tensor \cdots \tensor \cH^{\qreg{R}_n}$.
We sometimes write $\qstate{X}^{\qreg{R}_i}$ to emphasize that the operator $\qstate{X}$ acts on register $\cH^{\qreg{R}_i}$.\footnote{The superscript parts are gray colored.}
When we apply $\qstate{X}^{\qreg{R}_1}$ to registers $\cH^{\qreg{R}_1}$ and $\cH^{\qreg{R}_2}$, $\qstate{X}^{\qreg{R}_1}$ is identified with $\qstate{X}^{\qreg{R}_1} \tensor \mat{I}^{\qreg{R}_2}$.

A unitary operation is represented by a complex matrix $\mat{U}$ such that $\mat{U}\mat{U}^\dagger = \mat{I}$. The operation $\mat{U}$ transforms $\ket{\psi}$ and $\qstate{X}$ into $\mat{U}\ket{\psi}$ and $\mat{U}\qstate{X}\mat{U}^\dagger$, respectively.
A projector $\mat{P}$ is a Hermitian operator ($\mat{P}^\dagger =\mat{P}$) such that $\mat{P}^2 = \mat{P}$.

For a quantum state $\qstate{X}$ over two registers $\cH^{\qreg{R}_1}$ and $\cH^{\qreg{R}_2}$, we denote the state in $\cH^{\qreg{R}_1}$ as $\qstate{X}[\qreg{R}_1]$, where $\qstate{X}[\qreg{R}_1]= \Trace_2[\qstate{X}]$ is a partial trace of $\qstate{X}$ (trace out $\qreg{R}_2$).

% Given a function $F: X\ra Y$, a quantum-accessible oracle $O$ of $F$ is modeled by a unitary transformation $\mat{U}_F$ operating on two registers $\cH^{\qreg{in}}$ and $\cH^{\qreg{out}}$, in which $\ket{x}\ket{y}$ is mapped to $\ket{x}\ket{y\oplus F(x)}$, where $\oplus$ denotes XOR group operation on $Y$.
% We write $\qA^{\ket{O}}$ to denote that the algorithm $\qA$'s oracle $O$ is a quantum-accessible oracle.
% \ryo{We use standard PRFs in the security proofs. Should we add the definition of quantum PRFs? If we do not use the definition of quantum PRFs, I will remove this paragraph.}

\paragraph{Measurement Implementation}\label{sec:measurement_implementation}

We review some concepts on quantum measurements.

\begin{definition}[Projective Implementation~\cite{TCC:Zhandry20}]\label{def:projective_implementation}
Let:
\begin{itemize}
 \item $\cD$ be a finite set of distributions over an index set $\cI$.
 \item $\cP=\setbk{\mat{P}_i}_{i\in \cI}$ be a POVM
 \item $\cE = \setbk{\mat{E}_D}_{D\in\cD}$ be a projective measurement with index set $\cD$.
 \end{itemize}
 We consider the following measurement procedure.
 \begin{enumerate}
 \item Measure under the projective measurement $\cE$ and obtain a distribution $D$.
 \item Output a random sample from the distribution $D$.
 \end{enumerate}
 We say $\cE$ is the projective implementation of $\cP$, denoted by $\projimp(\cP)$, if the measurement process above is equivalent to $\cP$.
\end{definition}

\begin{theorem}[{\cite[Lemma 1]{TCC:Zhandry20}}]\label{lem:commutative_projective_implementation}
Any binary outcome POVM $\cP=(\mat{P},\mat{I}-\mat{P})$ has a unique projective implementation $\projimp(\cP)$.
\end{theorem}

%\begin{definition}[Threshold Implementation~\cite{EPRINT:ALLZZ20}]\label{def:threshold_implementation}
%A threshold implementation with parameter $\gamma$ of a binary POVM $\cP=(P,Q)$, denoted by $(\TI_{\gamma}(\cP),\mat{I}-\TI_{\gamma}(\cP))$, is a variant of projective implementation $\projimp(\cP)$. This executes the following measurement.
%\begin{enumerate}
%\item $\TI_{\gamma}(\cP)$ measures whether the corresponding distribution $D=(d_0,d_1)$ has $d_0 \ge \gamma$.
%\item Output $0$ with probability $\Tr[\TI_{\gamma}(\cP)\qstateq]$ and $1$ with probability $1-\Tr[\TI_{\gamma}(\cP)\qstateq]$ for any quantum state $\qstateq$.
%\end{enumerate}
%\end{definition}

\begin{definition}[Mixture of Projetive Measurement~\cite{TCC:Zhandry20}]\label{def:mixture_projective_measurement}
Let $D: \cR \ra \cI$ where $\cR$ and $\cI$ are some sets.
Let $\setbk{(\mat{P}_i,\mat{Q}_i)}_{\in \cI}$ be a collection of binary projective measurement.
The mixture of projective measurements associated to $\cR$, $\cI$, $D$, and $\setbk{(\mat{P}_i,\mat{Q}_i)}_{\in \cI}$ is the binary POVM $\cP_D =(\mat{P}_D,\mat{Q}_D)$ defined as follows
\begin{align}
& \mat{P}_D = \sum_{i\in\cI}\Pr[i \chosen D(R)]\mat{P}_i && \mat{Q}_D = \sum_{i\in\cI}\Pr[i \chosen D(R)]\mat{Q}_i,
\end{align}
where $R$ is uniformly distributed in $\cR$.
\end{definition}

\begin{definition}[Shift Distance]\label{def:shift_distance}
For two distributions $D_0,D_1$, the shift distance with parameter $\epsilon$, denoted by $\shiftdis{\epsilon}(D_0,D_1)$, is the smallest quantity $\delta$ such that for all $x \in \R$:
\begin{align}
\Pr[D_0\le x] & \le \Pr[D_1\le x + \epsilon] + \delta,&& \Pr[D_0\ge x]  \le \Pr[D_1\ge x - \epsilon] + \delta,\\
\Pr[D_1\le x] & \le \Pr[D_0\le x + \epsilon] + \delta,&& \Pr[D_1\ge x]  \le \Pr[D_0\ge x - \epsilon] + \delta.
\end{align}
For two real-valued measurements $\cM$ and $\cN$ over the same quantum system, the shift distance between $\cM$ and $\cN$ with parameter $\epsilon$ is
\[
\shiftdis{\epsilon}(\cM,\cN)\seteq \sup_{\ket{\psi}}\shiftdis{\epsilon}(\cM(\ket{\psi}),\cN(\ket{\psi})).
\]
\end{definition}

\begin{theorem}[\cite{TCC:Zhandry20,EC:KitNis22}]\label{thm:api_property}
Let $D$ be any probability distribution and $\cP=\setbk{(\Pi_i,\mat{I} -\Pi_i)}_i$ be a collection of binary outcome projective measurements. For any $0<\epsilon,\delta<1$, there exists an algorithm of measurement $\API_{\cP,\cD}^{\epsilon,\delta}$ that satisfies the following.
\begin{itemize}
\item $\shiftdis{\epsilon}(\API_{\cP,D}^{\epsilon,\delta},\projimp(\cP_D))\le \delta$.
\item $\API_{\cP,D}^{\epsilon,\delta}$ is $(\epsilon,\delta)$-almost projective in the following sense. For any quantum state $\ket{\psi}$, we apply $ \API_{\cP,D}^{\epsilon ,\delta}$ twice in a row to $\ket{\psi}$ and obtain measurement outcomes $x$ and $y$, respectively. Then, $\Pr[\abs{x-y}\le \epsilon]\ge 1-\delta$.
\item $\API_{\cP,D}^{\epsilon,\delta}$ is $(\epsilon,\delta)$-reverse almost projective in the following sense. For any quantum state $\ket{\psi}$, we apply $ \API_{\cP,D}^{\epsilon ,\delta}$ and $\API_{\cPrev,D}^{\epsilon,\delta}$ in a row to $\ket{\psi}$ and obtain measurement outcomes $x$ and $y$, respectively, where $\cPrev=\setbk{(\mat{I} -\Pi_i,\Pi_i)}_i$. Then, $\Pr[\abs{(1-x)-y}\le \epsilon]\ge 1-\delta$.
\item The expected running time of $\API_{\cP,D}^{\epsilon,\delta}$ is $T_{\cP,D}\cdot \poly(1/\epsilon,\log(1/\delta))$ where $T_{\cP,D}$ is the combined running time of $D$, the procedure mapping $i \ra (\mat{P}_i,\mat{I}- \mat{P}_i)$, and the running time of measurement $(\mat{P}_i,\mat{I}-\mat{P}_i)$.
\end{itemize}
\end{theorem}

%\begin{theorem}[{\cite[Theorem 3]{TCC:Zhandry20}}]\label{thm:cind_sample_projective_implementation}
%Let $\qstateq$ be an efficiently constructible mixed state, and $D_0,D_1$ efficiently sampleable distributions.
%If $D_0$ and $D_1$ are computationally indistinguishable, for any inverse polynomial $\epsilon$, there exists a negligible $\delta$ such that $\shiftdis{\epsilon}(\cM_0(\qstateq),\cM_1(\qstateq)) \le \delta$.
%\end{theorem}

\begin{theorem}[{\cite[Corollary 1]{TCC:Zhandry20}}]\label{cor:cind_sample_api}
Let $\qstateq$ be an efficiently constructible, potentially mixed state, and $D_0,D_1$ efficiently sampleable distributions.
If $D_0$ and $D_1$ are computationally indistinguishable, for any inverse polynomial $\epsilon$ and any function $\delta$, we have $\shiftdis{3\epsilon}(\API_{\cP,D_0}^{\epsilon,\delta},\API_{\cP,D_1}^{\epsilon,\delta}) \le 2\delta + \negl(\secp)$.
\end{theorem}

\paragraph{Quantum program with classical inputs and outputs}
We formalize quantum programs whose inputs and outputs are always classical strings.

\begin{definition}[Quantum Program with Classical Inputs and Outputs~\cite{C:ALLZZ21}]\label{def:Q_program_C_IO}
A quantum program with classical inputs is a pair of quantum state $\qstateq$ and unitaries $\setbk{\mat{U}_x}_{x\in[N]}$ where $[N]$ is the domain, such that the state of the program evaluated on input $x$ is equal to $\mat{U}_x \qstateq \mat{U}_x^\dagger$. We measure the first register of $\mat{U}_x \qstateq \mat{U}_x^\dagger$ to obtain an output. We say that $\setbk{\mat{U}_x}_{x\in[N]}$ has a compact classical description $\mat{U}$ when applying $\mat{U}_x$ can be efficiently computed given $\mat{U}$ and $x$.
\end{definition}

\subsection{Standard Cryptographic Tools}\label{sec:crypto_tools}

\paragraph{Commitment.}
We introduce the notion of injective commitment with equivocal mode.
This is an extension of injective commitment introduced by Cohen et al.~\cite{SIAMCOMP:CHNVW18}.
\begin{definition}[Injective Commitment with Equivocal Mode]\label{def:inj_com}
An injective commitment scheme $\Com$ with equivocal mode for the message space $\cM$ and random coin space $\cR$ is a tuple of four algorithms $(\Setup,\Commit,\EqSetup,\Open)$.
\begin{itemize}
\item The setup algorithm $\Setup$ takes as input a security parameter $1^\lambda$, and outputs a commitment key $\ck$.

\item The commitment algorithm $\Commit$ takes as input the commitment key $\ck$, a message $m\in\cM$, and a random coin $r\in\cR$, and outputs a commitment $\com$.

\item The equivocation setup algorithms $\EqSetup$ takes as input a security parameter $1^\secp$, and outputs a commitment key $\ck^*$, a commitment $\com^*$, and a trapdoor $\td$.

\item The open algorithm $\Open$ takes as input the trapdoor $\td$, a message $m\in\cM$, and a commitment $\com^*$, and outputs a random coin $r^*\in\cR$.

\end{itemize}
We say that injective commitment with equivocal mode is secure if it satisfies the following two properties.
\begin{description}
\item[Injectivity:] We require that
\ifnum\llncs=0
\begin{align}
\Pr[\exists (m_1,r_1)\ne(m_2,r_2) \textrm{~~s.t.~~} \Commit(\ck,m_1;r_1)=\Commit(\ck,m_2,r_2)\mid \ck\la\Setup(1^\secp)]=\negl(\secp).
\end{align}
\else
\begin{align}
\Pr\left[
\begin{array}{ll}
\exists (m_1,r_1)\ne(m_2,r_2)\\
\textrm{s.t.~} \Commit(\ck,m_1;r_1)=\Commit(\ck,m_2,r_2)
\end{array}
\ \middle |
\ck\la\Setup(1^\secp)
\right]
=\negl(\secp).
\end{align}
\fi

\item[Trapdoor Equivocality:] For any message $m\in\Ms$, we have
\begin{align}
(\ck,\com,r)\cind(\ck^*,\com^*,r^*),
\end{align}
where $\ck\la\Setup(1^\secp)$, $r\la\cR$, $\com\la\Commit(\ck,m;r)$, $(\ck^*,\com^*,\td)\la\EqSetup(1^\secp)$, and $r^*\la\Open(\td,m,\com^*)$.
\end{description}
We do not explicitly require a hiding property since we do not need it in this work.
\end{definition}

\begin{theorem}\label{thm:injective_Com_equivocation_Naor}
If the LWE assumption holds, there exists a secure injective commitment with equivocal mode.
\end{theorem}
We can obtain~\cref{thm:injective_Com_equivocation_Naor} from the construction of injective commitment by Kitagawa and Nishimaki~\cite{EC:KitNis22} since it is Naor's commitment scheme~\cite{JC:Naor91} and it is well known that Naor's commitment has a trapdoor equivocal mode.
See~\cref{sec:Naor_com} for the construction.

\paragraph{Ciphertext-policy attribute-based encryption.} We define ciphertext-policy attribute-base encryption (CP-ABE) and adaptive-indistinguishability (AD-IND security) and adaptive-simulation security (AD-SIM security) for it.
\begin{definition}[CP-ABE (Syntax)]\label{def:cpabe_syntax}
A CP-ABE scheme is a tuple of PPT algorithms $(\Setup,\KG,\Enc,\Dec)$ with plaintext space $\Ms$.
\begin{description}
\item [$\Setup(1^\secp)\ra (\mpk,\msk)$:] The setup algorithm takes as input the security 
    parameter $1^\secp$ and outputs a key pair $(\mpk,\msk)$.
    \item [$\KG(\msk,x)\ra \sk_x$:] The key generation algorithm takes as input the master secret key $\msk$ and an attribute $x$, and outputs a decryption key $\sk_x$.
    \item [$\Enc(\mpk,C,\msg)\ra \ct$:] The encryption algorithm takes as input $\mpk$, policy $C$, and a plaintext $\msg\in\cM$, and outputs a ciphertext $\ct$.
    \item [$\Dec(\sk_x,\ct)\ra \msg^\prime \mbox{ or }\bot$:] The decryption algorithm takes as input $\sk_x$ and $\ct$ and outputs a plaintext $\msg^\prime$ or $\bot$.

\item[Decryption Correctness:] There exists a negligible function $\negl$ such that for any $\msg\in\Ms$, $C$, and $x$ such that $C(x)=1$, we have
\begin{align}
\Pr\left[
\Dec(\sk_x,\ct)= \msg
\ \middle |
\begin{array}{ll}
(\mpk,\msk)\la\Setup(1^\secp)\\
\sk_x\gets \KG(\msk,x)\\
\ct \gets \Enc(\mpk,C,\msg)
\end{array}
\right] 
=1-\negl(\secp).
\end{align}

\end{description}  
\end{definition}

\begin{definition}[AD-IND Security for CP-ABE]\label{def:adind_cpabe_security}
Let $\CPABE=(\Setup,\KG,\allowbreak \Enc,\allowbreak \Dec)$ be a CP-ABE scheme.
We consider the following security experiment $\expb{\CPABE,\qA}{ad}{ind}(\secp,\coin)$ for a QPT adversary $\qA$.

\begin{enumerate}
    \item The challenger computes $(\mpk,\msk) \gets \KG(1^\secp)$ and sends $\mpk$ to $\qA$.
    \item $\qA$ can get access to the following oracle.
            \begin{description}
            \item[$\Oracle{\KG,1}(x)$:] Given $x$, it returns $\sk_x\la\KG(\msk,x)$. 
            \end{description}
    \item $\qA$ sends $C$ and $(\msg_0,\msg_1) \in \Ms^2$ to the challenger, where $C$ satisfies $C(x)=0$ for all $x$ queried by $\qA$ in the previous step. The challenger returns $\ct \gets \Enc(\mpk,C,m_\coin)$ to $\qA$.
    \item$\qA$ can get access to the following oracle.
    \begin{description}
            \item[$\Oracle{\KG,2}(x)$:] Given $x$, if $C(x)=0$, it returns $\sk_x \gets \KG(\msk,x)$. If $C(x)=1$, it returns $\bot$.
            \end{description}
    \item $\qA$ outputs $\coin^\prime\in \bit$.
\end{enumerate}
We say that $\CPABE$ is AD-IND secure if for any QPT adversary $\qA$, it holds that
\begin{align}
\advb{\CPABE,\qA}{ad}{ind}(\secp)\seteq \abs{\Pr[ \expb{\CPABE,\qA}{ad}{ind}(\secp, 0)=1] - \Pr[ \expb{\CPABE,\qA}{ad}{ind}(\secp, 1)=1] }\leq \negl(\secp).
\end{align}
\end{definition}

\begin{remark}
We can consider the selective variant of~\cref{def:adind_cpabe_security}, where $\qA$ declares $C$ at the beginning of the games.
We can also consider the selective variant of~\cref{def:adsim_cpabe_security} introduced below. We denote these notions by SEL-IND and SEL-SIM, respectively.
\end{remark}

\begin{theorem}[\cite{JACM:GorVaiWee15,EC:BGGHNS14}]\label{thm:ABE_from_LWE}
If the LWE assumption holds, there exists SEL-IND secure CP-ABE for all boolean circuits. In addition, if the LWE assumption holds against sub-exponential time algorithms, there exists ADA-IND secure CP-ABE for all boolean circuits.
\end{theorem}

We define AD-SIM security for CP-ABE, which is a natural adaptation of AD-SIM security for functional encryption~\cite{C:GorVaiWee12} to CP-ABE.

\begin{definition}[AD-SIM Security for CP-ABE]\label{def:adsim_cpabe_security}
Let $\CPABE=(\Setup,\KG,\allowbreak \Enc,\allowbreak \Dec)$ be a CP-ABE scheme.
We consider the following security experiment $\expb{\CPABE,\qSim,\qA}{ad}{sim}(\secp,\coin)$ for a QPT simulator $\qSim=(\qSimEnc,\qSimKG)$ and a QPT adversary $\qA$.

\begin{enumerate}
    \item The challenger computes $(\mpk,\msk) \gets \KG(1^\secp)$ and sends $\mpk$ to $\qA$.
    \item $\qA$ can get access to the following oracle.
            \begin{description}
            \item[$\Oracle{\KG,1}(x)$:] Given $x$, it returns $\sk_x\la\KG(\msk,x)$. 
            \end{description}
    \item $\qA$ sends $C$ and $m \in \Ms$ to the challenger, where $C$ satisfies $C(x)=0$ for all $x$ queried by $\qA$ in the previous step. The challenger does the following.
    \begin{itemize}
    \item If $\coin =0$, the challenger generates $\ct \gets \Enc(\mpk,C,m)$ and returns $\ct$ to $\qA$.
    \item If $\coin=1$, the challenger generates $(\ct,\state) \gets \qSimEnc(\mpk,C)$ and returns $\ct$ to $\qA$.
    \end{itemize}
    \item$\qA$ can get access to the following oracle.
    \begin{description}
            \item[$\Oracle{\KG,2}(x)$:] Given $x$, if $C(x)=0$, it returns $\sk_x \gets \KG(\msk,x)$. If $C(x)=1$, it does the following.
                \begin{itemize}
    \item If $\coin=0$, the challenger returns $\sk_x \gets \KG(\msk,x)$ to $\qA$.
    \item If $\coin=1$, the challenger returns $\sk_x \gets \qSimKG(\msk,\state,x,m)$ to $\qA$.
    \end{itemize}
            \end{description}
    \item $\qA$ outputs $\coin^\prime\in \bit$.
\end{enumerate}
We say that $\CPABE$ is AD-SIM secure if there exists $\qSim$ such that for any QPT adversary $\qA$, it holds that
\begin{align}
\advb{\CPABE,\qSim,\qA}{ad}{sim}(\secp)\seteq \abs{\Pr[ \expb{\CPABE,\qSim,\qA}{ad}{sim}(\secp, 0)=1] - \Pr[ \expb{\CPABE,\qSim,\qA}{ad}{sim}(\secp, 1)=1] }\leq \negl(\secp).
\end{align}
\end{definition}

AD-SIM security has been widely studied for functional encryption, but as far as we know, there was no previous work studied it for CP-ABE.
In \cref{sec:ADSIM_CPABE}, we show how to transform any AD-IND secure CP-ABE scheme into AD-SIM secure one by using tools implied by the LWE assumption.

\paragraph{Compute-and-compare obfuscation.}
We define a class of circuits called compute-and-compare circuits for which we study copy protection and secure software leasing in this work.

\begin{definition}[Compute-and-Compare Circuits]\label{def:cc_circuits_searchability}
A compute-and-compare circuit $\cnc{P}{\lock,\msg}$ is of the form
\[
\cnc{P}{\lock,\msg}(x)\left\{
\begin{array}{ll}
\msg&(P(x)=\lock)\\
0&(\text{otherwise})~,
\end{array}
\right.
\]
where $P$ is a circuit, $\lock$ is a string called lock value, and $\msg$ is a message.
\end{definition}

We introduce the definition of compute-and-compare obfuscation.
We assume that a program $P$ has an associated set of parameters $\pp_P$ (input size, output size, circuit size) which we do not need to hide.
\begin{definition}[Compute-and-Compare Obfuscation]\label{def:CCObf}
A PPT algorithm $\CCObf$ is an obfuscator for the family of distributions $D=\{D_\secp\}$ if the following holds:
\begin{description}
\item[Functionality Preserving:] There exists a negligible function $\negl$ such that for all program $P$, all lock value $\lock$, and all message $\msg$, it holds that
\begin{align}
\Pr[\forall x, \tlP(x)=\cnc{P}{\lock,\msg}(x) \mid \tlP\la\CCObf(1^\secp,P,\lock,\msg)]=1-\negl(\secp).
\end{align}
\item[Distributional Indistinguishability:] There exists an efficient simulator $\Sim$ such that for all message $\msg$, we have
\begin{align}
(\CCObf(1^\secp,P,\lock,\msg),\qaux)\cind(\Sim(1^\secp,\pp_P,\abs{\msg}),\qaux),
\end{align}
where $(P,\lock,\qaux)\la D_\secp$.
\end{description}
\end{definition}

\begin{theorem}[\cite{FOCS:GoyKopWat17,FOCS:WicZir17}]
If the LWE assumption holds, there exists compute-and-compare obfuscation for all families of distributions $D=\{D_\secp\}$, where each $D_\secp$ outputs uniformly random lock value $\lock$ independent of $P$ and $\qaux$.
\end{theorem}

\subsection{Quantum Cryptographic Tools}

\paragraph{Unclonable encryption.}
We introduce the definition of secret key unclonable encryption (SKUE)~\cite{TQC:BroLor20} and one-time indistinguishability for it.

\begin{definition}[SKUE (Syntax)]\label{def:unclonable_ske}
A SKUE scheme with the message space $\cM$ is a tuple of quantum algorithms $(\KG,\qencrypt,\qdecrypt)$.
\begin{description}
\item[$\KG(1^\secp)\ra\uk$:] The key generation algorithm takes as input the security parameter $1^\secp$ and outputs a key $\uk$.
    \item[$\qencrypt(\uk,m) \ra \qct$:] The encryption algorithm takes as input $\uk$ and a plaintext $\msg\in\Ms$ and outputs a ciphertext $\qct$.
    \item[$\qdecrypt(\uk,\qct) \ra \msg^\prime$:] The decryption algorithm takes as input $\uk$ and $\qct$ and outputs a plaintext $\msg^\prime \in \Ms$ or $\bot$.

\item[Decryption correctness:] For any $\msg\in\Ms$, it holds that
\begin{align}
\Pr\left[
\qDec(\uk,\qct)= \msg
\ \middle |
\begin{array}{ll}
\uk\gets \KG(1^\secp)\\
\qct \gets \qencrypt(\uk,\msg)
\end{array}
\right] 
=1-\negl(\secp).
\end{align}
\end{description}
\end{definition}

\begin{definition}[One-Time Unclonable-Indistinguishable Security for SKUE]\label{def:ot_unclonable_ind_ske}
Let $\UE=(\qencrypt, \qdecrypt)$ be an SKUE scheme with the key space $\cK$ and the message space $\cM$.
We consider the following security experiment $\expc{\UE,\qA}{ot}{ind}{clone}(\secp)$, where $\qA=(\qA_0,\qA_1,\qA_2)$.

\begin{enumerate}
\item $\qA_0$ sends $(m_0,m_1)$ to the challenger.
    \item The challenger generates $\coin\la\bit$, $\uk \la \KG(1^\secp)$, and $\qct \la \qencrypt(\uk,m_\coin)$, and sends $\qct$ to $\qA_0$.
    \item $\qA_0$ creates a bipartite state $\qstateq$ over registers $\qreg{R}_1$ and $\qreg{R}_2$. $\qA$ sends $\qstateq[\qreg{R}_1]$ and $\qstateq[\qreg{R}_2]$ to $\qA_1$ and $\qA_2$, respectively.
\item The challenger sends $\uk$ to $\qA_1$ and $\qA_2$.
$\qA_1$ and $\qA_2$ respectively output $\coin_1^\prime$ and $\coin_2^\prime$.
If $\coin_i^\prime = \coin$ for $i\in \setbk{1,2}$, the challenger outputs $1$, otherwise outputs $0$.

\end{enumerate}
We say that $\UE$ is one-time unclonable-indistinguishable secure SKUE scheme if for any QPT $\qA$, it holds that
\begin{align}
\advc{\UE,\qA}{ot}{ind}{clone}(\secp)\seteq \Pr[ \expc{\UE,\qA}{ot}{ind}{clone}(\secp)=1]\le \frac{1}{2} + \negl(\secp).
\end{align}
\end{definition}

\paragraph{Single-decryptor encryption.}
We review the definition of single-decryptor encryption (SDE).
We consider a one-time secret key variant of SDE by Coladangelo et al.~\cite{C:CLLZ21} in this work.

\begin{definition}[Secret Key SDE (Syntax)]\label{def:sd_ske}
A secret key SDE scheme $\SDE$ is a tuple of quantum algorithms $(\qKG,\Enc,\qDec)$ with plaintext space $\Ms$.
\begin{description}
    \item[$\qKG (1^\secp) \ra (\ek,\qdk)$:] The key generation algorithm takes as input the security parameter $1^\secp$ and outputs an encryption key $\ek$ and a quantum decryption key $\qdk$.
    \item[$\Enc(\ek,m) \ra \ct$:] The encryption algorithm takes as input $\ek$ and a plaintext $m\in\Ms$ and outputs a ciphertext $\ct$.
    \item[$\qDec(\qdk,\ct) \ra m^\prime$:] The decryption algorithm takes as input $\qdk$ and $\ct$ and outputs a plaintext $m^\prime \in \Ms$ or $\bot$.

\item[Decryption correctness:] There exists a negligible function $\negl$ such that for any $m\in\Ms$, 
\begin{align}
\Pr\left[
\qDec(\qdk,\ct)= m
\ \middle |
\begin{array}{ll}
(\ek,\qdk)\gets \qKG(1^\secp)\\
\ct \gets \Enc(\ek,m)
\end{array}
\right] 
=1-\negl(\secp).
\end{align}

\end{description}
\end{definition}

\begin{definition}[One-Time Strong Anti-Piracy Security for Secret Key SDE]\label{def:otsanti_piracy}
Let $\gamma\in[0,1]$. Let $\SDE=(\qKG,\Enc,\qDec)$ be a secret key SDE scheme. We consider the one-time strong anti-piracy game $\expc{\SDE,\qA}{ot}{santi}{piracy}(\secp,\gamma)$ between the challenger and an adversary $\qA=(\qA_0,\qA_1,\qA_2)$ below.
\begin{enumerate}
\item The challenger generates $(\ek,\qdk) \gets \qKG(1^\secp)$ and sends $\qdk$ to $\qA_0$.
\item $\qA_0$ creates a bipartite state $\qstateq$ over registers $\qreg{R}_1$ and $\qreg{R}_2$. $\qA$ sends $(m_0,m_1)$, $\qstateq[\qreg{R}_1]$, and $\qstateq[\qreg{R}_2]$ to the challenger, $\qA_1$, and $\qA_2$, respectively.

\item $\qA_1$ and $\qA_2$ respectively output $\pirateD_1=(\rho[\qreg{R}_1],\mat{U}_1)$ and $\pirateD_2=(\rho[\qreg{R}_2],\mat{U}_2)$.

\item For $\alpha\in[2]$, let $\mat{P}_{\alpha}=(\mat{P}_{\alpha,b,\ct},\mat{I}-\mat{P}_{\alpha,b,\ct})_{b,\ct}$ be a collection of binary projective measurements, where
\begin{align}
\mat{P}_{\alpha,b,\ct}=\mat{U}_{\alpha,\ct}^\dagger\ket{b}\bra{b}\mat{U}_{\alpha,\ct}.
\end{align}
We also define $D$ as the distribution that generates $b\la\bit$ and $\ct\la\Enc(\ek,m_b)$ and outputs $(b,\ct)$.
Also, for $\alpha\in[2]$, we denote the mixture of $\mat{P}_\alpha$ with respect to $D$ as $\mat{P}_{\alpha,D}$.
Then, for every $\alpha\in[2]$, the challenger applies $\projimp(\mat{P}_{\alpha,D})$ to $\rho[\qreg{R}_\alpha]$ and obtains $p_\alpha$.
If $p_\alpha > \frac{1}{2}+\gamma$ for every $\alpha\in[2]$, the challenger outputs $1$.
Otherwise, the challenger outputs $0$.
\end{enumerate}

We say that $\SDE$ is one-time strong anti-piracy secure if for any $\gamma\in[0,1]$ and QPT adversary $\qA$, it satisfies that
\begin{align}
\advc{\SDE,\qA}{ot}{santi}{piracy}(\secp,\gamma)\seteq\Pr[ \expc{\SDE,\qA}{ot}{santi}{piracy}(\secp, \gamma)=1]= \negl(\secp).
\end{align}
\end{definition}

\begin{remark}
Readers might think the meaning of ``one-time'' is unclear in~\cref{def:otsanti_piracy}.
Here, ``one-time'' means that $\qA_0$ cannot access an encryption oracle that returns a ciphertext for a query $\msg$. This naming might sound strange since $\qA_0$ does not receive any ciphertext. However, we stick to this naming for correspondence with one-time unclonable-indistinguishable security of unclonable encryption in~\cref{def:ot_unclonable_ind_ske}.
\end{remark}

\begin{remark}[On the issue in indistinguishability-based definitions]\label{rem:SDE_def}
There are two indistinguishability-based security notions for SDE.
The first one is defined by Georgiou and Zhandry (denoted by GZ)~\cite{EPRINT:GeoZha20} and the second one is defined by Coladangelo et al. (denoted by CLLZ)\footnote{They call ``CPA-style anti-piracy security''.}~\cite{C:CLLZ21}.
Both of them are defined using a security game similar to that in \cref{def:otsanti_piracy}, except that $\qA_1$ and $\qA_2$ are given the challenge ciphertexts and required to guess the challenge bits.
In the GZ definition, $\qA_1$ and $\qA_2$ receive the same ciphertext $\Enc(\sk,m_\coin)$ for the single challenge bit. However, in the CLLZ definition, $\qA_1$ and $\qA_2$ receive different ciphertexts $\Enc(\sk,m_{\coin_1})$ and $\Enc(\sk,m_{\coin_2})$, respectively, where $\coin_1$ and $\coin_2$ are independent challenge bits.
Currently, the relationship between these two security notions for SDE remains elusive.\footnote{Ananth et al.~\cite{ARXIV:AnaKalLiu23} show the relationship between \emph{one-wayness-based} security with the same ciphertext and one with the different ciphertexts.}
The GZ definition is known to imply unclonable-indistinguishable secure unclonable encryption, but the CLLZ definition is not.
Also, strong anti-piracy security defined in \cref{def:otsanti_piracy} implies the CLLZ definition but not the GZ definition.
\end{remark}

\paragraph{Quantum fully homomorphic encryption.}
We introduce quantum fully homomorphic encryption (QFHE) with classical ciphretexts.

\begin{definition}[QFHE with Classical Ciphertexts~\cite{FOCS:Mahadev18b,C:Brakerski18}]\label{def:QFHE}
A QFHE scheme with classical ciphertext is a tuple of algorithms $(\KG,\Enc,\qEval,\Dec)$.
\begin{description}
\item[$\KG(1^\secp)\ra (\pk,\sk)$:] The key generation algorithm takes as input the security parameter $1^\secp$ and outputs a key pair $(\pk,\sk)$.
\item[$\Enc(\pk,m)\ra \ct$:] The encryption algorithm takes as input a public key $\pk$ and a plaintext $m$, and outputs a ciphertext $\ct$.
Without loss of generality, we can assume that a ciphertext includes the public key $\pk$.
\item[$\qEval(\qC,\rho,\ct_1,\cdots,\ct_\inplen)\ra\evalct$:] The evaluation algorithm takes as input a quantum circuit $\qC$ with classical outputs, quantum state $\rho$, and ciphertexts $\ct_1,\cdots,\ct_\inplen$, and outputs a classical ciphertext $\evalct$.
\item[$\Dec(\sk,\ct)\ra m^\prime$:] The decryption algorithm takes as input a secret key $\sk$ and a ciphertext $\ct$, and outputs a plaintext $m$.
\item[Decryption Correctness:] Let $(\pk,\sk)\gets\KG(1^\secp)$.
Let $(m_1,\cdots,m_\inplen)$ be any $\inplen$ messages.
For any $i\in[\inplen]$, let $\ct_i\gets\Enc(\pk,m_i)$ for every $i\in[\inplen]$.
Let $\qC$ be a quantum circuit that takes a quantum state and $\inplen$ classical input, $\rho$ a quantum state, and $\evalct\gets\qEval(\qC,\rho,\ct_1,\cdots,\ct_\inplen)$.
Then, we have $\Dec(\sk,\evalct)=\qC(\rho,m_1,\cdots,m_\inplen)$.

\item[Semantic Security:] For any two messages of equal length $m_0,m_1$, we have
\begin{align}
(\pk,\Enc(\pk,m_0))\cind(\pk,\Enc(\pk,m_1)),
\end{align}
where $(\pk,\sk)\gets\KG(1^\secp)$.
\end{description}
\end{definition}

The existing QFHE schemes~\cite{FOCS:Mahadev18b,C:Brakerski18} can be seen as QFHE with classical ciphertexts, since they have a property that if the encrypted plaintext is classical, we can make the ciphertext classical. Thus, the following theorem holds.

\begin{theorem}[\cite{FOCS:Mahadev18b,C:Brakerski18}]
If the LWE assumption holds, there exists QFHE with classical ciphertexts.
\end{theorem}

QFHE with classical ciphertexts was previously used in the context of impossibility on copy protection~\cite{EC:AnaLaP21} and quantum obfuscation~\cite{C:ABDS21}.
For the detailed explanation for how to use the existing QFHE schemes as QFHE with classical ciphertexts, please refer to~\cite{C:ABDS21}.

\else\fi
%\input{MICCObf}
% !TEX root = main.tex

\section{One-out-of-Many Unclonable Security}\label{sec:one-out-of-many}
This section introduces new unclonable security notions that we call one-out-of-many unclonable security.
The definition is roughly as follows.
The one-out-of-many unclonable security game is an indistinguishability-style game played by a tuple of $\numa+1$ adversaries $(\qA_0,\qA_1,\cdots,\qA_\numa)$, where $ 2\le \numa$ is arbitrary.
At the first stage of the game, $\qA_0$ is given a single quantum object (such as ciphertext in unclonable encryption), generates possibly entangled $\numa$-partite states $\rho_1,\ldots,\rho_\numa$, and sends $\rho_k$ to $\qA_k$ for $k\in \setbk{1,\ldots,\numa}$. At the second stage, the challenger selects one of $(\qA_1,\ldots,\qA_\numa)$ by a random $\alpha\chosen \setbk{1,\ldots, \numa}$ and sends additional information (such as a secret key in unclonable encryption) \emph{only to $\qA_\alpha$}, and only $\qA_\alpha$ tries to guess the challenge bit $\coin\in\bit$.
% Here, we focus on distinguishing games such as unclonable encryption and simulation-based secure copy-protection.
The one-out-of-many unclonable security guarantees that the adversary cannot win this game with a probability significantly better than the trivial winning probability $\frac{1}{2}+\frac{1}{2n}$.\footnote{Suppose $\qA_0$ forwards the given quantum state to $\qA_1$ and nothing to $(\qA_2,\ldots,\qA_\numa)$. If $\alpha=1$ is chosen, the adversaries win with probability $1$ because the additional information, together with the original quantum object, can be used to compute the challenge bit $\coin$ correctly. If one of $(\qA_2,\ldots,\qA_\numa)$ is chosen, the adversaries win with probability $\frac{1}{2}$ by random guess. Hence, the advantage is $\frac{1}{n}\cdot 1+\frac{n-1}{n}\cdot\frac{1}{2}=\frac{1}{2}+\frac{1}{2n}$, which we consider as the trivial advantage.}
.

 The one-out-of-many unclonable security notion guarantees that no adversary can generate $\numa$ copies with a probability significantly better than $\frac{1}{\numa}$ for any $\numa$.
This is because an adversary who can generate $\numa$ copies with probability $\frac{1}{\numa}+\delta$ can win the one-out-of-many game with probability at least $(\frac{1}{\numa}+\delta)\cdot1+(1-\frac{1}{\numa}-\delta)\cdot\frac{1}{2}=\frac{1}{2}+\frac{1}{2\numa}+\frac{\delta}{2}$, which violates to the one-out-of-many unclonable security.
The one-out-of-many unclonable security notion does not rule out a copying process that can generate $\numa$ copies with probability $\frac{1}{\numa}$. However, it guarantees that such a process must completely destroy the original object with probability $1-\frac{1}{\numa}$.
In fact, it guarantees that the expected number of successful copies generated by any copying adversary is at most $1$.

Although one-out-of-many unclonable security looks weaker than existing unclonable security, it still seems useful in some applications. For example, suppose we publish many quantum objects, say $\ell$ objects. Then, the one-out-of-many security guarantees that no matter what copying attacks are applied to those objects, there are expected to be only $\ell$ objects on average in this world.

Below, we define one-out-of-many one-time unclonable-indistinguishable security for SKUE and one-out-of-many one-time anti-piracy for secret key SDE.
We prove that one-out-of-many one-time unclonable-indistinguishable security for SKUE implies one-time IND-CPA security.
We also prove that one-time strong anti-piracy for secret key SDE implies one-out-of-many one-time anti-piracy for secret key SDE.
Then, we show that we can transform secret key SDE with one-out-of-many one-time anti-piracy into one-out-of-many one-time unclonable-indistinguishable secure SKUE.

In \cref{sec:CP}, we define one-out-of-many copy protection security for single-bit output point functions.
In \cref{sec:UPE}, we define simulation-based security for unclonable PE and introduce its one-out-of-many variant.
    
\subsection{One-out-of-Many Security Notions for SKUE and Secret Key SDE}

\begin{definition}[One-out-of-Many One-Time Unclonable-Indistinguishability for SKUE]\label{def:ek_os_ot_unclone_ind}
Let $\UE=(\KG,\qencrypt, \qdecrypt)$ be an SKUE scheme with the message space $\cM$.
We consider one-out-of-many one-time unclonable-indistinguishability game $\expd{\UE,\qA}{om}{ot}{clone}{ind}(\secp,\numa)$ between the challenger and an adversary $\qA=(\qA_0,\qA_1,\allowbreak\cdots,\qA_\numa)$ below.
\begin{enumerate}
\item $\qA_0$ sends $(\msg_0,\msg_1)$ to the challenger.
\item The challenger generates $\coin\la\bit$, $\uk \la \KG(1^\secp)$, and $\qct \la \qencrypt(\uk,\msg_\coin)$, and sends $\qct$ to $\qA_0$.
\item $\qA_0$ creates a quantum state $\qstateq$ over $\numa$ registers $\qreg{R}_1,\cdots,\qreg{R}_\numa$. $\qA_0$ sends $\qstateq[\qreg{R}_i]$ to $\qA_i$ for every $i\in[\numa]$.

\item The challenger generates $\alpha\la[\numa]$, and gives $\uk$ to $\qA_\alpha$. $\qA_\alpha$ outputs $\coin^\prime$.
The challenger outputs $1$ if $\coin^\prime=\coin$ and outputs $0$ otherwise.
\end{enumerate}

We say that $\UE$ is one-out-of-many one-time unclonable-indistinguishable if for any polynomial $\numa=\numa(\secp)$ and QPT adversary $\qA=(\qA_0,\qA_1,\cdots,\qA_\numa)$, it satisfies that
\begin{align}
\advd{\UE,\qA}{om}{ot}{clone}{ind}(\secp,\numa)\seteq\Pr[ \expd{\UE,\qA}{om}{ot}{unclone}{ind}(\secp,\numa)=1]
\le \frac{1}{2}+\frac{1}{2\numa}+\negl(\secp).
\end{align}
\end{definition}

\begin{theorem}\label{thm:omind_to_ind}
Let $\UE=(\KG,\qEnc,\qDec)$ be a one-out-of-many one-time unclonable-indistinguishable secure SKUE scheme with the message space $\cM$.
Then, $\UE$ satisfies one-time IND-CPA security, that is,
$$\qEnc(\uk,\msg_0)\cind\qEnc(\uk,\msg_1)$$
for any $(\msg_0,\msg_1)\in\cM$, where $\uk\la\KG(1^\secp)$.
\end{theorem}
\begin{proof}
Suppose there exists $\qB$ who can distinguish $\qEnc(\uk,\msg_0)$ from $\qEnc(\uk,\msg_1)$ with probability $\frac{1}{2}+p$ for some $(\msg_0,\msg_1)\in\cM$ and inverse polynomial $p$, where $\uk\la\KG(1^\secp)$.
Consider the following $\numa=\frac{1}{p}$ tuple of adversaries $\qA=(\qA_0,\qA_1,\cdots,\qA_\numa)$ for the one-out-of-many one-time unclonable-indistinguishable security.
On input $\qct$, $\qA_0$ gives them to $\qB$ and obtains $\qB$'s guess $\coin^\prime$, and sends it to $\qA_i$ for every $i\in[\numa]$.
$\qA_i$ just outputs $\coin^\prime$ if $\alpha=i$ is chosen by the challenger.
We have $\advd{\UE,\qA}{om}{ot}{unclone}{ind}(\secp,\numa)=\frac{1}{2}+p$, which contradicts to the one-out-of-many unclonable-indistinguishable security since $p>\frac{1}{2\numa}=\frac{p}{2}$.
\end{proof}

\begin{definition}[One-out-of-Many One-Time Anti-Piracy Security for Secret Key SDE]\label{def:om_anti_piracy}
Let $\SDE=(\qKG,\Enc,\qDec)$ be a secret key SDE scheme. We consider one-out-of-many one-time anti-piracy game $\expc{\SDE,\qA}{om}{otanti}{piracy}(\secp,\numa)$ between the challenger and an adversary $\qA=(\qA_0,\qA_1,\cdots,\qA_\numa)$ below.
\begin{enumerate}
\item The challenger generates $(\ek,\qdk) \gets \qKG(1^\secp)$ and sends $\qdk$ to $\qA_0$.
\item $\qA_0$ creates a quantum state $\qstateq$ over $\numa$ registers $\qreg{R}_1,\cdots,\qreg{R}_\numa$. $\qA_0$ sends $(\msg_0,\msg_1)$ to the challenger. $\qA_0$ also sends $\qstateq[\qreg{R}_i]$ to $\qA_i$ for every $i\in[\numa]$.

\item $\qA_i$ outputs $\pirateD_i=(\rho[\qreg{R}_i],\mat{U}_i)$ for every $i\in[\numa]$.

\item The challenger generates $\alpha\la[\numa]$ and $\coin\la\bit$, and generates $\ct\la\Enc(\ek,\msg_\coin)$. The challenge runs $\pirateD_\alpha$ on input $\ct$ and obtains $\coin^\prime$.
The challenger outputs $1$ if $\coin^\prime=\coin$ and outputs $0$ otherwise.
\end{enumerate}

We say that $\SDE$ is one-out-of-many one-time anti-piracy secure if for any polynomial $\numa=\numa(\secp)$ and QPT adversary $\qA=(\qA_0,\qA_1,\cdots,\qA_\numa)$, it satisfies that
\begin{align}
\advc{\SDE,\qA}{om}{otanti}{piracy}(\secp,\numa)\seteq\Pr[ \expc{\SDE,\qA}{om}{otanti}{piracy}(\secp,\numa)=1]
\le \frac{1}{2}+\frac{1}{2\numa}+\negl(\secp).
\end{align}
\end{definition}

We show that one-time strong anti-piracy security for secret key SDE implies one-ouf-of-many one-time anti-piracy for secret key SDE.
\begin{theorem}\label{thm:santi-piracy_implies_om_anti-piracy}
Let $\SDE$ be a secret key SDE scheme.
If $\SDE$ is one-time strong anti-piracy secure, then $\SDE$ is also one-out-of-many one-time anti-piracy secure.
\end{theorem}
\begin{proof}
$\expc{\SDE,\qA}{om}{anti}{piracy}(\secp,\numa)$ is equivalent to the security game where item 4 is replaced with the following.
\begin{itemize}
\item The challenger generates $\alpha\la[\numa]$. For every ${\alpha^\prime}\in[\numa]$, the challenger applies $\projimp(\mat{P}_{{\alpha^\prime},D})$ to $\rho[\qreg{R}_{\alpha^\prime}]$ and obtains $p_{\alpha^\prime}$.
The challenger outputs $1$ with probability $p_\alpha$.
\end{itemize}
This equivalence follows from the definition of $\projimp$.
Using this version of $\expc{\SDE,\qA}{om}{anti}{piracy}(\secp,\numa)$, we prove that $\advc{\SDE,\qA}{om}{anti}{piracy}(\secp,\numa)\le\frac{1}{2}+\frac{1}{2\numa}+\gamma$ for any inverse polynomial $\gamma$.
Since $\SDE$ is strong anti-piracy secure, except for some single index $i^*$, $p_i$ computed by the challenger is smaller than $\frac{1}{2}+\gamma$, with overwhelming probability.
Thus, we have
\begin{align}
\advc{\SDE,\qA}{om}{anti}{piracy}(\secp,\numa)&\le\frac{1}{\numa}\cdot 1+\frac{\numa-1}{\numa}\cdot(\frac{1}{2}+\gamma) + \negl(\secp)\\
&\le\frac{1}{2}+\frac{1}{2\numa}+\gamma.
\end{align}
This completes the proof.
\end{proof}

\subsection{From Secret-Key SDE to SKUE: One-out-of-Many Setting}\label{sec:SDSKE_to_OTindSKUE}

We present a transformation from a secret key SDE scheme that satisfies~\cref{def:om_anti_piracy} into a SKUE scheme that satisfies~\cref{def:ek_os_ot_unclone_ind}. Georgiou and Zhandry developed this transformation~\cite{EPRINT:GeoZha20}.
Note that they do not consider one-out-of-many security for secret key SDE and SKUE.
We show that their transformation works even in the one-out-of-many setting.

Let $\SDE=(\SDE.\qKG,\SDE.\Enc,\SDE.\qDec)$ be a secret key SDE scheme. We also let $\ell$ be the length of ciphertexts of $\SDE$.
We construct a SKUE scheme $\UE=(\UE.\KG,\UE.\qEnc,\UE.\qDec)$ as follows.

\begin{description}
 \item[$\UE.\KG(1^\secp)$:] $ $
 \begin{itemize}
\item Output $\uk \chosen \zo{\ell}$.
 \end{itemize}
 \item[$\UE.\qEnc(\uk,\msg)$:] $ $
 \begin{itemize}
 \item Generate $(\sde.\ek,\sde.\qdk)\gets \SDE.\qKG(1^\secp)$.
 \item Compute $\sde.\ct \gets \SDE.\Enc(\sde.\ek,\msg)$.
 \item Output $\ue.\qct \seteq (\sde.\ct \xor \uk, \sde.\qdk)$.
 \end{itemize}
\item[$\UE.\qDec(\uk,\ue.\qct)$:] $ $
\begin{itemize}
\item Parse $(\ct_1^\prime, \sde.\qdk) =\ue.\qct$.
\item Compute $\sde.\ct^\prime \seteq \ct_1^\prime \xor \uk$.
\item Output $\msg^\prime \gets \SDE.\qDec(\sde.\qdk,\sde.\ct^\prime)$.
\end{itemize}
\end{description}

\begin{theorem}\label{thm:om_santi_piracy_implies_om_ind-ue}
If $\SDE$ is one-out-of-many one-time anti-piracy secure, $\UE$ is one-out-of-many one-time unclonable-indistinguishable secure.
\end{theorem}
\begin{proof}
Let $\numa$ be any polynomial of $\secp$.
We construct an adversary $\qB =(\qB_0,\qB_1,\cdots,\qB_\numa)$ for $\SDE$ by using the adversary $\qA=(\qA_0,\qA_1,\cdots,\qA_\numa)$ for $\UE$.
\begin{enumerate}
\item $\qB_0$ is given $\sde.\qdk$ from its challenger.
\item $\qB_0$ runs $\qA_0$ and receives $(\msg_0,\msg_1)$, and passes $(\msg_0,\msg_1)$ to its challenger.
\item $\qB_0$ generates $\uk \chosen \zo{\ell}$, sets $\ue.\qct \seteq (\uk,\sde.\qdk)$, and passes $\ue.\qct$ to $\qA_0$. Then, $\qA_0$ create a quantum state $\qstateq_\qA$ over $\numa$ registers $\qreg{C}_1,\cdots,\qreg{C}_\numa$. 
$\qB_0$ receives them.
\item $\qB_0$ creates a quantum state $\qstateq_\qB$ over $\numa$ registers $\qreg{R}_1,\cdots,\qreg{R}_\numa$ such that $\qstateq_\qB[\qreg{R}_i] \seteq (\uk,\qstateq_\qA[\qreg{C}_i])$ for every $i\in[\numa]$, then $\qB_0$ sends $\qstateq_\qB[\qreg{R}_i]$ to $\qB_i$ for every $i\in[\numa]$.
\item For every $i\in[\numa]$, $\qB_i$ outputs $\pirateD_i=(\qstateq_\qB[\qreg{R}_1],\mat{U}_i)$, where $\mat{U}_i$ is a unitary that takes $\sde.\ct$ and $\qstateq_\qB[\qreg{R}_1]$ as inputs, and outputs $\coin^\prime \gets \qA_i (\qstateq_\qA[\qreg{C}_i],\uk \xor \sde.\ct)$.
\end{enumerate}
If the challenger of one-out-of-many security of SDE chooses $\alpha \chosen [\numa]$ and $\coin \chosen \zo{}$, and generates $\sde.\ct \gets \SDE.\Enc(\sde.\ek,\msg_\coin)$, then the challenger runs $\pirateD_\alpha$ on input $\sde.\ct$ and obtains the output $\coin^\prime \gets \qA_\alpha (\qstateq_\qA[\qreg{C}_\alpha],\uk \xor \sde.\ct)$.
If $\sde.\ct = \SDE.\Enc(\sde.\ek,\msg_\coin)$, then $\ue.\qct = (\uk,\sde.\qdk)= \UE.\qEnc(\sde.\ct\xor \uk,\msg_\coin)$. $\qB$ correctly simulates the one-out-of-many security game of SKUE for $\qA$ since $\uk$ is a uniformly random string.
Therefore, the probability that $\qA$ succeeds in breaking $\UE$ is bounded by the probability that $\qB$ succeeds in breaking $\SDE$.
This completes the proof.
\end{proof}

% !TEX root = main.tex

\section{One-Time Secret Key SDE from LWE}\label{sec:ot_SDE_HE}

We construct secret key SDE satisfying one-time strong anti-piracy based on the LWE assumption in this section.

\subsection{Tools}
\paragraph{Ciphertext-policy functional encryption.} We introduce ciphertext-policy functional encryption (CPFE).
Since we consider single-key setting by default, we use a simplified syntax where the setup algorithm takes as input $x$ and outputs a master public key and a functional decryption key for $x$.
There is no key generation algorithm.

\begin{definition}[Single-Key Ciphertext-Policy Functional Encryption]\label{def:cpfe}
A single-key CPFE scheme for the circuit space $\cC$ and the input space $\cX$ is a tuple of algorithms $(\Setup, \Enc, \Dec)$.
\begin{itemize}
\item The setup algorithm $\Setup$ takes as input a security parameter $1^\lambda$ and an input $x \in \cX$, and outputs a master public key $\MPK$ and functional decryption key $\sk_x$.

\item The encryption algorithm $\Enc$ takes as input the master public key $\MPK$ and $C\in\cC$, and outputs a ciphertext $\ct$.

\item The decryption algorithm $\Dec$ takes as input a functional decryption key $\sk_x$ and a ciphertext $\ct$, and outputs $y$.
\end{itemize}

\begin{description}
\item[Decryption Correctness:] We require $\Dec(\sk_x, \Enc(\MPK, C)) \allowbreak = C(x)$ for every $C\in\cC$, $x\in\cX$, and $\left(\MPK,\sk_x \right) \la \Setup(1^\lambda,x)$.
\end{description}

\end{definition}

\begin{definition}[$1$-Bounded Security]\label{def:CPFE_security}
Let $\CPFE$ be a single-key CPFE scheme.  %whose message space and function space are $\calM$ and $\calF$, respectively.
We define the game $\expt{\CPFE,\qA}{1\textrm{-}bounded}(\secp,\coin)$ as follows.

%\begin{description}
%\item[Initialization:] $ $
\begin{enumerate}
\item $\qA$ sends $x\in\cX$ to the challenger.
\item The challenger generates $(\MPK,\sk_x) \la \Setup(1^\lambda,x)$ and sends $(\MPK,\sk_x)$ to $\qA$.

\item $\qA$ outputs $(C_0,C_1)$ such that $C_0(x)=C_1(x)$. The challenger generates $\ct\la\Enc(\MPK,C_\coin)$, and sends $\ct$ to $\qA$.

\item $\qA$ outputs $\coin' \in \bin$.
%\end{description}
\end{enumerate}
We say that $\CPFE$ is $1$-bounded secure if for every QPT $\qA$, we have
\begin{align}
\adva{\CPFE,\qA}{1\textrm{-}bounded}(\secp)=
\abs{\Pr[
\expt{\CPFE,\qA}{1\textrm{-}bounded}(\secp,0)=1
] -\Pr[
\expt{\CPFE,\qA}{1\textrm{-}bounded}(\secp,1)=1
]} =\negl(\secp).
\end{align}

\end{definition}

\begin{definition}[Succinct Key]\label{def:succinct_key}
We say that a single-key CPFE scheme satisfies succinct key property if there exist two algorithms $\HKG$ and $\Hash$ satisfying the following conditions.
\begin{itemize}
\item $\Setup(1^\secp,x)$ runs $\hk\la\HKG(1^\secp,1^\abs{x})$, compute $h\la\Hash(\hk,x)$, and outputs $\MPK \seteq (\hk,h)$ and $\sk_x\seteq x$. For the setup of a CPFE scheme with succinct key property, we omit to write $\sk_x \seteq x$ from the output of $\Setup$ and we simply write $\MPK\la\Setup(1^\secp,x)$.
\item The length of $h$ output by $\Hash$ is $\secp$ regardless of the length of the input $x$.
\end{itemize}
\end{definition}

\begin{theorem}\label{thm:single_key_CPFE_succinct_key}
If the LWE or exponentially-hard LPN assumption holds, there exists single-key CPFE with succinct key for $\Ppoly$.
\end{theorem}
See~\cref{sec:succinct_CPFE} for the proof of~\cref{thm:single_key_CPFE_succinct_key}.

\paragraph{Monogamy of entanglement.}
We review the monogamy of entanglement property of of BB84 states~\cite{NJP:TFKW13} and its variant.

\begin{theorem}[Monogamy Property of BB84 States~\cite{NJP:TFKW13}]\label{thm:MoE_BB84}
Consider the following game between a challenger and an adversary $\qA=(\qA_0,\qA_1,\qA_2)$.
\begin{enumerate}
\item The challenger picks a uniformly random strings $x\in\bit^n$ and $\theta\in\bit^n$. It sends $\ket{x^\theta} \seteq H^{\theta[1]}\ket{x[1]}\tensor\cdots \tensor H^{\theta[n]}\ket{x[n]}$ to $\qA_0$.
\item $\qA_0$ creates a bipartite state $\qstateq$ over registers $\qreg{R}_1$ and $\qreg{R}_2$. Then, $\qA_0$ sends register $\qreg{R}_1$ to $\qA_1$ and register $\qreg{R}_2$ to $\qA_2$.
\item $\theta$ is then sent to both $\qA_1$ and $\qA_2$.
$\qA_1$ and $\qA_2$ return respectively $x_1^\prime$ and $x_2^\prime$.
\end{enumerate}
Let $\MoEBB(\qA,\secp)$ be a random variable which takes the value $1$ if $x_1^\prime=x_2^\prime=x$, and takes the value $0$ otherwise. Then, there exists an exponential function $\expo$ such that for any adversary $\qA=(\qA_0,\qA_1,\qA_2)$, it holds that
\begin{align}
\Pr[\MoEBB(\qA,\secp)=1]\leq 1/\exp(n).
\end{align}
\end{theorem}

We introduce a variant of the monogamy property above where the adversary can select a leakage function $\Leak$ and obtain $\Leak(x)$.
\begin{theorem}[Monogamy Property of BB84 States with Leakage]\label{thm:MoE_BB84_leakage}
Consider the following game between a challenger and an adversary $\qA=(\qA_0,\qA_1,\qA_2)$.
\begin{enumerate}
\item $\qA$ sends a function $\Leak$ whose output length is $\ell$ to the challenger.
\item The challenger picks a uniformly random strings $x\in\bit^n$ and $\theta\in\bit^n$. It sends $\ket{x^\theta} \seteq H^{\theta[1]}\ket{x[1]}\tensor\cdots \tensor H^{\theta[n]}\ket{x[n]}$ and $\Leak(x)$ to $\qA_0$.
\item $\qA_0$ creates a bipartite state $\qstateq$ over registers $\qreg{R}_1$ and $\qreg{R}_2$. Then, $\qA_0$ sends register $\qreg{R}_1$ to $\qA_1$ and register $\qreg{R}_2$ to $\qA_2$.
\item $\theta$ is then sent to both $\qA_1$ and $\qA_2$.
$\qA_1$ and $\qA_2$ return respectively $x_1^\prime$ and $x_2^\prime$.
\end{enumerate}
Let $\MoEBBLeak(\qA,\secp)$ be a random variable which takes the value $1$ if $x_1^\prime=x_2^\prime=x$, and takes the value $0$ otherwise. Then, there exists an exponential function $\expo$ such that for any adversary $\qA=(\qA_0,\qA_1,\qA_2)$, it holds that
\begin{align}
\Pr[\MoEBBLeak(\qA,\secp)=1]\leq 2^\ell/\exp(n).
\end{align}
Especially, if $\ell$ is independent of $n$, the right hand side is negligible in $\secp$ by setting $n$ appropriately.
\end{theorem}

We can reduce \cref{thm:MoE_BB84_leakage} to \cref{thm:MoE_BB84} by guessing $\Leak(x)$ with probability $1/2^\ell$.

\subsection{Construction}\label{sec:construction_sk_SDE}

We use a CPFE scheme with succinct key property $\CPFE=(\Setup,\Enc,\Dec)$ to construct a secret key SDE scheme $\SDE=(\qSDKG,\SDEnc,\qSDDec)$.
The description of $\SDE$ is as follows. The plaintext space of $\SDE$ is $\bit^\msglen$.
\begin{description}

 \item[$\qSDKG(1^\secp)$:] $ $
 \begin{itemize}
\item Generate $x,\theta\la\bit^n$.
\item Generate $\ket{x^\theta}=H^{\theta[1]}\ket{x[1]}\tensor\ldots H^{\theta[n]}\ket{x[n]}$.
\item Generate $\MPK\la\Setup(1^\secp,x)$.
\item Output $\ek\seteq(\theta,\MPK)$ and $\qdk\seteq\ket{x^\theta}$.
 \end{itemize}
 \item[$\SDEnc(\ek,m)$:] $ $
 \begin{itemize}
 \item Parse $\ek = (\theta,\MPK)$.
\item Let $C[m]$ be a constant circuit that outputs $m$ on any input. $C$ is padded so that it has the same size as $C^*$ appeared in the security proof.
\item Compute $\ct\la\CPFE.\Enc(\MPK,C[m])$.
 \item Output $\sdct \seteq(\theta,\ct)$.
 \end{itemize}
\item[$\qSDDec(\qdk,\sdct)$:] $ $
\begin{itemize}
\item Parse $\qdk =\ket{x^\theta}$ and $\sdct = (\theta,\ct)$.
\item Compute $x$ from $\ket{x^\theta}$ and $\theta$.
\item Output $m\la\Dec(x,\ct)$.
\end{itemize}
\end{description}

\begin{theorem}\label{thm:SDE_from_MOE}
If $\CPFE$ is a CPFE scheme that satisfies succinct key property and $1$-bounded security,  $\SDE$ is a one-time strong anti-piracy secure single-decryptor SKE.
\end{theorem}

From~\cref{thm:SDE_from_MOE,thm:single_key_CPFE_succinct_key,thm:om_santi_piracy_implies_om_ind-ue,thm:santi-piracy_implies_om_anti-piracy}, we immediately obtain the following corollary.
\begin{corollary}
If the LWE assumption holds, there exists one-out-of-many unclonable-indistinguishable secure unclonable encryption.
\end{corollary}
\begin{proof}[Proof of~\cref{thm:SDE_from_MOE}]
Let $\gamma\in[0,1]$ and $\qA=(\qA_0,\qA_1,\qA_2)$ be any QPT adversary.
$\qA_1$ and $\qA_2$ respectively output $\pirateD_1=(\qreg{R}_1,\mat{U}_1)$ and $\pirateD_2=(\qreg{R}_2,\mat{U}_2)$, where $\qreg{R}_1$ and $\qreg{R}_2$ have a possibly entangled quantum state $\rho$, and $\mat{U}_\alpha=(\mat{U}_{\alpha,\ct})_\ct$ for $\alpha\in[2]$.
We define the collection of binary projective measurements $\mat{P}_{\alpha}=(\mat{P}_{\alpha,b,\ct},\mat{I}-\mat{P}_{\alpha,b,\ct})_{b,\ct}$ for every $\alpha\in[2]$, the distribution $D$, and the mixture of $\mat{P}_{\alpha}$ with respect to $D$ $\mat{P}_{\alpha,D}$ for every $\alpha\in[2]$ in the same way as \cref{def:otsanti_piracy}.
Then, we can write
\begin{align}
\advc{\SDE,\qA}{ot}{santi}{piracy}(\secp,\gamma)=\Pr[p_1>\frac{1}{2}+\gamma\land p_2>\frac{1}{2}+\gamma],
\end{align}
where $p_\alpha$ is the result of applying $\projimp(\mat{P}_{\alpha,D})$ to $\qreg{R}_\alpha$.
 
We assume that $\advc{\SDE,\qA}{ot}{santi}{piracy}(\secp,\gamma)=\eta$ for some inverse polynomial $\eta$.
Using $\qA$, we construct $\qB=(\qB_0,\qB_1,\qB_2)$ that attacks the monogamy property of BB84 states with leakage.
Recall that since $\CPFE$ satisfies succinct key property, there exist two algorithms $\HKG$ and $\Hash$ such that $\Setup(1^\secp,x)$ runs $\hk\la\HKG(1^\secp,1^\abs{x})$, compute $h\la\Hash(\hk,x)$, and outputs $\MPK:=(\hk,h)$, where $h$ is a $\secp$-bit string. (See~\cref{def:succinct_key}.)
\begin{enumerate}
\item $\qB_0$ generates $\hk\la\HKG(1^\secp,1^n)$ and sends a function $\Leak(\cdot):=\Hash(\hk,\cdot)$ to the challenger.
\item $\qB_0$ is given $\ket{x^\theta}$ and $\Leak(x)=h$. $\qB_0$ sets $\MPK:=(\hk,h)$.
\item By setting $\qdk\seteq\ket{x^\theta}$, $\qB_0$ simulates $\expc{\SDE,\qA}{ot}{santi}{piracy}(\secp,\gamma)$ for $\qA_0$ and obtains $(m_0,m_1)$ and a quantum state $\qstateq$ over registers $\qreg{R}_1$ and $\qreg{R}_2$. Then, $\qB_0$ sends $(\gamma,\MPK,m_0,m_1,\qreg{R}_1)$ to $\qB_1$ and $(\gamma,\MPK,m_0,m_1,\qreg{R}_2)$ to $\qB_2$.
\item $\theta$ is then sent to both $\qB_1$ and $\qB_2$. $\qB_1$ and $\qB_2$ behave as follows.
\begin{itemize}
\item $\qB_\alpha$ first sets $\ek=(\theta,\MPK)$.
$\qB_\alpha$ simulates $\expc{\SDE,\qA}{ot}{santi}{piracy}(\secp,\gamma)$ for $\qA_\alpha$ and obtains $\qD_\alpha=(\qstateq_\alpha,\mat{U}_\alpha)$.
$\qB_\alpha$ computes $x_\alpha^\prime\la \qExtract(\MPK,m_0^*,m_1^*,\qD_\alpha,\gamma)$ and outputs $x_\alpha^\prime$, where the algorithm $\qExtract$ is described below.
\end{itemize}
\end{enumerate}

\begin{description}
\item[$\qExtract(\MPK,\msg_0^*,\msg_1^*,\qD_\alpha,\epsilon)$:]$ $
\begin{itemize}
\item Let $\epsilon'=\epsilon/4(n+1)$ and $\delta'=2^{-\lambda}$.
\item Parse $(\qstateq_\alpha,\mat{U}_\alpha)\gets\qD_\alpha$.
Let $D_{i}$ be the following distribution for every $i\in[\secp]$.
\begin{description}
\item[$D_{i}$:] Generate $b\la\bit$. Generate $\ct\la\CPFE.\Enc(\MPK,C^*[b,m_0,m_1,i])$, where $C^*[b,m_0,m_1,i]$ is a circuit that takes $x$ as input and outputs $m_{ b\oplus x[i]}$. Output $(b,\ct)$.
\end{description}
\item Compute $\tlp_{\alpha,0} \gets \API_{\cP,D}^{\epsilon' ,\delta'}(\qstateq_\alpha)$. If $\tlp_{\alpha,0}<\frac{1}{2}+\epsilon-4\epsilon'$, return $\bot$. Otherwise, letting $\qstateq_{\alpha,0}$ be the post-measurement state, go to the next step.
\item For all $i \in [\secp]$, do the following.
\begin{enumerate}
\item Compute $\tlp_{\alpha,i} \gets \API_{\cP,D_{i}}^{\epsilon' ,\delta'}(\qstateq_{\alpha,i-1})$. Let $\qstateq_{\alpha,i}$ be the post-measurement state.
\item If $\tlp_{\alpha,i}>\frac{1}{2}+\epsilon-4(i+1)\epsilon'$, set $x_\alpha^\prime[i]=0$. If $\tlp_{\alpha,i}<\frac{1}{2}-\epsilon+4(i+1)\epsilon'$, set $x_\alpha^\prime[i]=1$. Otherwise, exit the loop and output $\bot$.
\end{enumerate}
\item Output $x_\alpha^\prime=x_\alpha^\prime[1] \concat \cdots \concat x_\alpha^\prime[\secp]$.
\end{itemize}
\end{description}

We will estimate $\Pr[\MoEBBLeak(\qA,\secp)=1]$.
We define the events $\BadDec_\alpha$, and $\BadExt_{\alpha,i}$ for every $i\in[\secp]$.
\begin{description}
\item[$\BadDec_\alpha$:]When $\qB_\alpha$ runs $\qExtract(\MPK,m_0^*,m_1^*,\qD_\alpha,\epsilon)$, $\tlp_{\alpha,0}<\frac{1}{2}+\epsilon-4\epsilon'$ holds.
\item[$\BadExt_{\alpha,i}$:]When $\qB_\alpha$ runs $\qExtract(\MPK,m_0^*,m_1^*,\qD_\alpha,\epsilon)$, the following conditions hold.
\begin{itemize}
\item $\tlp_{\alpha,0}\geq\frac{1}{2}+\epsilon-4\epsilon'$ holds.
\item $x_\alpha^\prime[j]=x[j]$ holds for every $j\in[i-1]$.
\item $x_\alpha^\prime[i]\neq x[i]$ holds.
\end{itemize}
\end{description}

From the assumption that $\advc{\SDE,\qA}{ek}{santi}{piracy}(\secp,\gamma)=\eta$, for $\tlp_{\alpha,0}$ computed in $\qExtract$, $\tlp_{\alpha,0}\ge\frac{1}{2}+\epsilon-\epsilon^\prime$ holds with probability $\eta-\negl(\secp)$ for $\alpha\in[2]$, due to the first item of \cref{thm:api_property}.
This means that $\Pr[\BadDec_{1}\lor\BadDec_{2}]=1-\eta+\negl(\secp)$.
Then, we have
\ifnum\llncs=0
\begin{align}
\Pr[\MoEBBLeak(\qA,\secp)=1]&\ge1-\left(\Pr[\BadDec_1\lor\BadDec_2]+\sum_{i\in[\secp]}\Pr[\BadExt_{1,i}]+\sum_{i\in[\secp]}\Pr[\BadExt_{2,i}]\right)\\
&=\eta - \negl(\secp) -\left(\sum_{i\in[\secp]}\Pr[\BadExt_{1,i}]+\sum_{i\in[\secp]}\Pr[\BadExt_{2,i}]\right)
\end{align}
\else
\begin{align}
&\Pr[\MoEBBLeak(\qA,\secp)=1]\\
&\ge1-\left(\Pr[\BadDec_1\lor\BadDec_2]+\sum_{i\in[\secp]}\Pr[\BadExt_{1,i}]+\sum_{i\in[\secp]}\Pr[\BadExt_{2,i}]\right)\\
&=\eta - \negl(\secp) -\left(\sum_{i\in[\secp]}\Pr[\BadExt_{1,i}]+\sum_{i\in[\secp]}\Pr[\BadExt_{2,i}]\right)
\end{align}
\fi

\paragraph{Estimation of $\Pr[\BadExt_{\alpha,i}]$ for every $\alpha\in[2]$ and $i\in[\secp]$.}
We first estimate $\Pr[\BadExt_{\alpha,1}]$.
We first consider the case of $x[1]=0$.
From the first item of the event, we have $\tlp_{\alpha,0}>\frac{1}{2}+\epsilon-4\epsilon'$.
Let $\tlp'_{\alpha,0} \gets \API_{\cP,D}^{\epsilon' ,\delta'}(\qstateq_{\alpha,0})$.
From, the almost-projective property of $\API$, we have
\begin{align}
\Pr[\tlp'_{\alpha,0}>\frac{1}{2}+\epsilon-4\epsilon'-\epsilon']\geq1-\delta'.
\end{align}
\begin{lemma}\label{lem:ind_cpfe_0}
When $x[1]=0$, $D_{1}$ is computationally indistinguishable from $D$.
\end{lemma}
\begin{proof}
The difference between $D_{1}$ and $D$ is that $\ct$ is generated as  $\ct\la\CPFE.\Enc\allowbreak(\MPK,C^*[b,m_0,m_1,1])$ in $D_1$ and it is generated as  $\ct\la\CPFE.\Enc(\MPK,C[\msg_{b}])$ in $D$.
From the condition that $x[1]=0$, we have $C^*[b,m_0,m_1,1](x)=C[m_{b}](x)=m_{b}$.
Thus, from the $1$-bounded security of $\CPFE$, $D_{1}$ and $D$ are computationally indistinguishable when $x[1]=0$.
\end{proof}
Thus, from \cref{cor:cind_sample_api} and \cref{lem:ind_cpfe_0}, we have
\begin{align}
1-\delta'\leq
\Pr[\tlp'_{\alpha,0}>\frac{1}{2}+\epsilon-5\epsilon']
\leq\Pr[\tlp_{\alpha,1}>\frac{1}{2}+\epsilon-8\epsilon']+\negl(\secp).
\end{align}
This means that $\Pr[\BadExt_{\alpha,1}]=\negl(\secp)$ when $x[1]=0$.
We next consider the case of $x[1]=1$.
We define the following distribution $\Drev$.
\begin{description}
\item[$\Drev$:] Generate $(b,\ct)\gets D$. Output $(1\oplus b,\ct)$.
\end{description}
That is, the first bit of the output is flipped from $D$.
Then, for any random coin $r$, we have $(\mat{P}_{\Drev(r)},\mat{Q}_{\Drev(r)})=(\mat{Q}_{D(r)},\mat{P}_{D(r)})$.
This is because we have $\mat{Q}_{b,\ct}=\mat{I}-\mat{P}_{b,\ct}=\mat{P}_{1\oplus b,\ct}$ for any tuple $(b,\ct)$.
Therefore, $ \API_{\cP,\Drev}^{\epsilon' ,\delta'}$ is exactly the same process as  $\API_{\cPrev,D}^{\epsilon' ,\delta'}$, where $\cPrev=(\mat{Q}_{b,\ct},\mat{P}_{b,\ct})_{b,\ct}$.
Let $\tlp'_{\alpha,0} \gets \API_{\cP,\Drev}^{\epsilon' ,\delta'}(\qstateq_{\alpha,0})$.
From, the reverse-almost-projective property of $\API$, we have
\begin{align}
\Pr[\tlp'_{\alpha,0}<\frac{1}{2}-\epsilon+4\epsilon'+\epsilon']\geq1-\delta'.
\end{align}
\begin{lemma}\label{lem:ind_cpfe_1}
When $x[1]=1$, $D_{1}$ is computationally indistinguishable from $\Drev$.
\end{lemma}
\begin{proof}
We see that $\Drev$ is identical to the following distribution.
\begin{itemize}
\item Generate $b\la\bit$ and $\ct\la\Enc(\ek,\msg_{1\oplus b})$. Output $(b,\ct)$.
\end{itemize}
Then, the difference between $D_{1}$ and $\Drev$ is that $\ct$ is generated as  $\ct\la\CPFE.\Enc(\MPK,C^*[b,m_0,m_1,1])$ in $D$ and it is generated as $\ct\la\CPFE.\Enc(\MPK,C[\msg_{1\oplus b}])$ in $\Drev$.
From the condition that $x[1]=1$, we have $C^*[b,m_0,m_1,1](x)=C[m_{1\oplus b}](x)=\msg_{1\oplus b}$.
Thus, from the $1$-bounded security of $\CPFE$, $D_{1}$ and $\Drev$ are computationally indistinguishable when $x[1]=1$.
\end{proof}
Thus, from \cref{cor:cind_sample_api} and \cref{lem:ind_cpfe_1}, we have
\begin{align}
1-\delta'\leq
\Pr[\tlp'_{\alpha,0}<\frac{1}{2}-\epsilon+5\epsilon']
\leq\Pr[\tlp_{\alpha,1}<\frac{1}{2}-\epsilon+8\epsilon']+\negl(\secp).
\end{align}
This means that $\Pr[\BadExt_{\alpha,1}]=\negl(\secp)$ when $x[1]=1$.

Overall, $\Pr[\BadExt_{\alpha,1}]=\negl(\secp)$ regardless of the value of $x[1]$.
We can similarly show that $\Pr[\BadExt_{\alpha,i}]=\negl(\secp)$ for $i\in\{2,\cdots,\secp\}$ using the fact that $D_{i}$ is computationally indistinguishable from $D$ if $x[i]=0$ and it is computationally indistinguishable form $\Drev$ if $x[i]=1$.
We omit the details.

\medskip

From the above discussion, we have $\Pr[\MoEBBLeak(\qA,\secp)=1]=\eta-\negl(\secp)$ for some inverse polynomial $\eta$, which contradicts to the monogamy property of BB84 states with leakage.
This completes the proof of \cref{thm:SDE_from_MOE}.
\end{proof}

% !TEX root = main.tex

\section{Quantum Copy-Protection from Unclonable Encryption}\label{sec:CP}

We introduce one-out-of-many copy protection security for single-bit output point functions and present how to achieve it using one-out-of-many secure unclonable encyption in this section.

\subsection{Definition}\label{sec:def_copy_protection}

\begin{definition}[Copy-Protection (Syntax)]
A copy-protection scheme $\CP$ for a family of circuits $\calC$ consists of two algorithms $(\qCopyProtect,\qEval)$.
\begin{description}
\item[$\qCopyProtect(1^\secp,C)\ra\rho$:]The copy-protection algorithm takes as input the security parameter $1^\secp$, a circuit $C\in\calC$, and outputs a quantum state $\rho$.
\item[$\qEval(\rho,x)$:]The evaluation algorithm takes as input a quantum state $\rho$ and an input $x$, and outputs $y$.
\item[Evaluation Correctness:]For every circuit $C$ and input $x$, we have
\begin{align}
\Pr[
\qEval(\rho,x)=C(x) \mid \rho\la\qCopyProtect(1^\secp,C)
]=1-\negl(\secp).
\end{align}
\end{description}
\end{definition}

\begin{remark}\label{rem:reusability}
We can assume without loss of generality that a copy protected program $\rho$ output by $\qCopyProtect$ is reusable, that is, it can be reused polynomially many times.
This is because the output of $\qEval$ on input $\rho$ and any input $x$ is almost deterministic by correctness, and thus such an operation can be done without almost disturbing $\rho$ by the gentle measurement lemma~\cite{TransIT:Winter99}.
\end{remark}

We focus on copy protection scheme for a family of single-bit output point functions that we denote $\PFs$.
We also define a family of single-bit output point functions $\PFone$ as $\PFone=\{f_{y}\}_{y\in\bit^{\pfinplen}}$, where $f_{y}$ outputs $1$ if the input is $y\in\bit^{\pfinplen}$ and $0$ otherwise.

We review the widely used copy-protection security for $\PFs$ originally introduced by Coladangelo et al.~\cite{ARXIV:ColMajPor20}.

\begin{definition}[Copy-Protection Security for $\PFs$]\label{def:cp_pf1}
Let $\CP$ be a copy protection scheme for $\PFone$.
Let $D_Y$ be a distribution over $\bit^{\pfinplen}$.
Let $D_X(\cdot)$ be a distribution over $\bit^{\pfinplen}\times\bit^{\pfinplen}$, where $D_X(\cdot)$ takes as input $y^\prime\in\bit^{\pfinplen}$.
We consider the following security experiment $\expb{\CP,D_Y,D_X,\qA}{cp}{pf1}(\secp)$ for $\qA=(\qA_0,\qA_1,\qA_2)$.

\begin{enumerate}
\item The challenger generates $y\la D_Y$. The challenger generates $\rho\la\qCopyProtect(1^\secp,f_{y})$ and sends $\rho$ to $\qA_0$.
    \item $\qA_0$ creates a bipartite state $\qstateq$ over registers $\qreg{R}_1$ and $\qreg{R}_2$. $\qA_0$ sends $\qstateq[\qreg{R}_1]$ and $\qstateq[\qreg{R}_2]$ to $\qA_1$ and $\qA_2$, respectively.
\item The challenger generates $(x_1,x_2)\la D_X(y)$, and sends $x_1$ and $x_2$ to $\qA_1$ and $\qA_2$, respectively.
\item $\qA_1$ and $\qA_2$ respectively output $b_1$ and $b_2$.
If $b_i = f_{y}(x_i)$ for $i\in \setbk{1,2}$, the challenger outputs $1$, otherwise outputs $0$.
\end{enumerate}

We define $\ptriv=\max_{i\in\setbk{1,2},\qS}p_{i,\qS}$, where
\begin{align}
p_{i,\qS}=\Pr\left[
b_i=f_y(x_i)
\ \middle |
\begin{array}{ll}
y\la D_Y,(x_1,x_2)\la D_X(y)\\
b_i\la\qS(x_i)
\end{array}
\right] 
\end{align}
and the maximization is done by all possibly computationally unbounded algorithm $\qS$.

We say that $\CP$ satisfies copy-protection security for $\PFs$ with respect to $D_Y$ and $D_X$ if for any QPT $\qA$, it holds that
\begin{align}
\advb{\CP,D_Y,D_X,\qA}{cp}{pf1}(\secp)\seteq \Pr[ \expb{\CP,D_Y,D_X,\qA}{cp}{pf1}(\secp)=1] 
\le \ptriv+\negl(\secp).
\end{align}

\end{definition}

The following definition is a natural adaptation of \cref{def:cp_pf1} into one-out-of-many setting.

\begin{definition}[One-out-of-Many Copy-Protection Security for $\PFs$]\label{def:cp_pf1_om}
Let $\CP$ be a copy protection scheme for $\PFone$.
Let $D_Y$ and $D_X(\cdot)$ be distributions over $\bit^{\pfinplen}$, where $D_X(\cdot)$ takes as input $y^\prime\in\bit^{\pfinplen}$.
We consider the following security experiment $\expc{\CP,D_Y,D_X,\qA}{cp}{pf1}{om}(\secp,\numa)$ for $\qA=(\qA_0,\qA_1,\cdots,\qA_\numa)$.

\begin{enumerate}
    \item The challenger generates $y\la D_Y$. The challenger generates $\rho\la\qCopyProtect(1^\secp,f_{y})$ and sends $\rho$ to $\qA_0$.
    \item $\qA_0$ creates a quantum state $\qstateq$ over $\numa$ registers $\qreg{R}_1,\cdots,\qreg{R}_\numa$. $\qA_0$ sends $\qstateq[\qreg{R}_i]$ to $\qA_i$ for every $i\in[\numa]$.
    \item The challenger generates $\alpha\la[\numa]$. The challenger generates $x\la D_X(y)$ and sends $x$ to $\qA_\alpha$.
$\qA_\alpha$ outputs $b_\alpha$.
If $b_\alpha=f_y(x)$, the challenger outputs $1$, otherwise outputs $0$.

\end{enumerate}

We define $\ptriv=\max_{\qS}p_{\qS}$, where
\begin{align}
p_{\qS}=\Pr\left[
b=f_y(x)
\ \middle |
\begin{array}{ll}
y\la D_Y,x\la D_X(y)\\
b\la\qS(x)
\end{array}
\right] 
\end{align}
and the maximization is done by all possibly computationally unbounded algorithm $\qS$.

We say that $\CP$ satisfies one-out-of-many copy-protection security for $\PFs$ with respect to $D_Y$ and $D_X$ if for any polynomial $\numa=\numa(\secp)$ and QPT $\qA$, it holds that
\begin{align}
\advc{\CP,D_Y,D_X,\qA}{cp}{pf1}{om}(\secp,\numa)\seteq \Pr[ \expc{\CP,D_Y,D_X,\qA}{cp}{pf1}{om}(\secp,\numa)=1] 
\le \frac{1}{\numa}\cdot 1+\frac{\numa-1}{\numa}\cdot\ptriv+\negl(\secp).
\end{align}
\end{definition}

\subsection{Construction}\label{sec:construction_copy_protection_from_ue}

We construct a copy-protection scheme for single-bit output point functions $\PFone$, where $\pfinplen$ is specified later.
We use the following tools:
\begin{itemize}
\item SKUE $\UE=(\UE.\KG,\UE.\qEnc,\UE.\qDec)$. Suppose the plaintext space of $\UE$ is $\bit^\secp$ and the key length is $\uekeylen$.
\item Injective commitment scheme with equivocal mode $\Com=(\Setup,\Commit,\EqSetup,\Open)$. Suppose the message space of $\Com$ is $\bit^{\uekeylen}$ and the random coin space is $\bit^{\comrandlen}$.
\item Compute-and-compare obfuscation $\CCObf$ with the simulator $\CC.\Sim$. In this section, the message feed to $\CCObf$ is fixed to $1$. Thus, we omit to write it from the input.
\item QFHE with classical ciphertexts $\QFHE=(\QFHE.\KG,\QFHE.\Enc,\QFHE.\qEval,\QFHE.\Dec)$.
\end{itemize}
We set $\pfinplen=\uekeylen+\comrandlen$.
The description of $\CP$ is as follows.
\begin{description}

 \item[$\qCopyProtect(1^\secp,f_y)$:] $ $
 \begin{itemize}
 \item Generate $(\pk,\sk)\la\QFHE.\KG(1^\secp)$, $\uk\la\UE.\KG(1^\secp)$, and $\ck\la\Setup(1^\secp)$.
 \item Generate $\lock\la\bit^\secp$.
 \item Parse $\maskval\|\comrand\la y$ and generate $\com\la\Commit(\ck,\maskval;\comrand)$.
\item Generate $\maskeduk\la\maskval\oplus\uk$.
\item Generate $\qfhe.\ct\la\QFHE.\Enc(\pk,(\com,\maskeduk))$.
\item Generate $\ue.\qct\la\UE.\qEnc(\uk,\lock)$.
\item Generate $\tlD\la\CCObf(1^\secp,\QFHE.\Dec(\sk,\cdot),\lock)$.
\item Output $\rho=(\ck,\qfhe.\ct,\ue.\qct,\tlD)$.
 \end{itemize}
 \item[$\qEval(\rho,x)$:] $ $
 \begin{itemize}
 \item Parse $\rho = (\ck,\qfhe.\ct,\ue.\qct,\tlD)$.
 \item Compute $\evalct\la\QFHE.\qEval(C[\ck,x],\ue.\qct,\qfhe.\ct)$, where the circuit $C[\ck,x]$ is described in \cref{fig:c_cp}.
 \item Output $y\la\tlD(\evalct)$.
 \end{itemize}
\end{description}

\protocol
{Quantum circuit $C[\ck,x]$}
{The description of $C[\ck,x]$}
{fig:c_cp}
{
\begin{description}
\setlength{\parskip}{0.3mm} % between paragraph
\setlength{\itemsep}{0.3mm} % between items
\item[Constants:]  Strings $\ck$ and $x$.
\item[Input:] A quantum state $\ue.\qct$ and strings $\com$ and $\maskeduk$.
\end{description}
\begin{enumerate}
\item Parse $x_{\mathsf{mask}}\| x_{\mathsf{comr}}\la x$. 
\item If $\com\ne\Commit(\ck,x_{\mathsf{mask}};x_{\mathsf{comr}})$, output $0^\secp$. Otherwise, go to the next step.
\item Compute $\uk^\prime \la x_{\mathsf{mask}}\oplus \maskeduk$.
\item Output $\lock^\prime \la\UE.\qDec(\uk^\prime,\ue.\qct)$.
\end{enumerate}
}

\paragraph{Evaluation Correctness.}
It is easy to see that $\CP$ satisfies evaluation correctness from the correctness of $\CCObf$, $\QFHE$, and $\UE$, and injectivity of $\Com$.

\paragraph{Security.}
For security, we have the following theorems.

\begin{theorem}\label{thm:CP_from_lattice_security_onesided}
Let $0\le w\le 1$.
We define the distributions $U_{\pfinplen}$ and $\Dresamp{w}(\cdot)$ as follows.
\begin{itemize}
\item $U_{\pfinplen}$ is the uniform distribution over $\bit^{\pfinplen}$.
\item $\Dresamp{w}(\cdot)$ is a distribution such that $\Dresamp{w}(y)$ outputs $y$ with probability $1-w$ and outputs a resampled value $z\la U_{\pfinplen}$ with probability $w$.
\end{itemize}
Let $D=\{D_\secp\}$ be a family of distributions where each $D_\secp$ outputs $(\QFHE.\Dec(\sk,\cdot),\lock,\qaux:=\pk)$ generated as those in $\qCopyProtect$.
If $\CCObf$ is secure with respect to $D$, $\QFHE$ satisfies semantic security, $\UE$ satisfies one-out-of-many one-time unclonable-indistinguishable security, and $\Com$ satisfies trapdoor equivocality, then $\CP$ satisfies one-out-of-many copy protection security for $\PFs$ with respect to the distributions $D_Y=U_{\pfinplen}$ and $D_X(\cdot)=\Dresamp{w}(\cdot)$.
\end{theorem}

\begin{theorem}\label{thm:CP_from_lattice_security}
We define $\Dresamp{w}^2(\cdot)$ as follows.
\begin{itemize}
\item $\Dresamp{w}^2(\cdot)$ is a distribution such that $\Dresamp{w}^2(y)$ outputs $(y,y)$ with probability $1-w$ and outputs $(z,z)$ for a resampled value $z\la U_{\pfinplen}$ with probability $w$.
\end{itemize}
In~\cref{thm:CP_from_lattice_security_onesided}, if we use one-time unclonable-indistinguishable secure $\UE$, $\CP$ satisfies copy protection security for $\PFs$ with respect to the distributions $D_Y=U_{\pfinplen}$ and $D_X(\cdot)=\Dresamp{w}^2(\cdot)$.
\end{theorem}

\begin{remark}[On instantiations]\label{rem:instantiation_QROM_UE}
When we instantiate $\CP$ based on \cref{thm:CP_from_lattice_security}, we need to be careful about the fact that the existing one-time unclonable-indistinguishable secure SKUE scheme uses QROM.
The construction of $\CP$ evaluates the decryption circuit of $\UE$ by $\QFHE$.
Thus, to use the QROM based SKUE scheme as the building block of $\CP$, we have to assume that it is secure when we replace QRO with real hash functions so that the decryption algorithm has a concrete description that $\QFHE$ can evaluate. Note that we always need this assumption to use QROM-based SKUE schemes in the real world.

When we instantiate $\CP$ based on \cref{thm:CP_from_lattice_security_onesided}, there is no such issue and we can obtain a construction secure in the standard model, since we have one-out-of-many one-time unclonable-indistinguishable secure SKUE based on the LWE assumption in the standard model.
\end{remark}

\ifnum\llncs=0

We below prove \cref{thm:CP_from_lattice_security_onesided} and omit the proof of~\cref{thm:CP_from_lattice_security}. It is easy to see that we can similarly prove~\cref{thm:CP_from_lattice_security} by using one-time unclonable-indistinguishable security at the transition from $\hybi{3}$ to $\hybi{4}$ in the proof of~\cref{thm:CP_from_lattice_security_onesided}.

% !TEX root = main.tex

%\begin{proof}

\ifnum\llncs=0
\begin{proof}[Proof of~\cref{thm:CP_from_lattice_security_onesided}]
\else
\section{Proof of~\cref{thm:CP_from_lattice_security_onesided}}\label{sec:proof_CP}
\fi
Let $\numa$ be any polynomial of $\secp$ and $\qA=(\qA_0,\qA_1,\cdots,\qA_\numa)$ any efficient adversary.
We consider the case where $1-w\le w$.
The proof when $w\le 1-w$ is similar.
In this case, $\ptriv=w$ and our goal is to show that
\begin{align}
\advc{\CP,D_Y,D_X,\qA}{cp}{pf1}{om}(\secp)
&=\frac{1}{\numa}\cdot 1+\frac{\numa-1}{\numa}\cdot w +\negl(\secp).
\end{align}
We prove it by using the following sequence of experiments.

\begin{description}
\item[$\hybi{1}$:] This is $\expc{\CP,D_Y,D_X,\qA}{cp}{pf1}{om}(\secp,\numa)$ where $D_X(y)=\Dresamp{w}(y)$ outputs $x=y$ without doing re-sampling, except that the output of the experiment is set to the adversary's output.
\begin{enumerate}
\item The challenger generates $y\la U_{\pfinplen}$. The challenger sends $\rho$ generated as follows to $\qA_0$.
 \begin{itemize}
 \item Generate $(\pk,\sk)\la\QFHE.\KG(1^\secp)$, $\uk\la\UE.\KG(1^\secp)$, and $\ck\la\Setup(1^\secp)$.
 \item Generate $\lock\la\bit^\secp$.
 \item Parse $\maskval\|\comrand\la y$ and generate $\com\la\Commit(\ck,\maskval;\comrand)$.
\item Generate $\maskeduk\la\maskval\oplus\uk$.
\item Generate $\qfhe.\ct\la\QFHE.\Enc(\pk,(\com,\maskeduk))$.
\item Generate $\ue.\qct\la\UE.\qEnc(\uk,\lock)$.
\item Generate $\tlD\la\CCObf(1^\secp,\QFHE.\Dec(\sk,\cdot),\lock)$.
\item Set $\rho=(\ck,\qfhe.\ct,\ue.\qct,\tlD)$.
 \end{itemize}
     \item $\qA_0$ creates a quantum state $\qstateq$ over $\numa$ registers $\qreg{R}_1,\cdots,\qreg{R}_\numa$. $\qA_0$ sends $\qstateq[\qreg{R}_i]$ to $\qA_i$ for every $i\in[\numa]$.
\item The challenger generates $\alpha\la[\numa]$. The challenger sends $x=y$ to $\qA_\alpha$.
$\qA_\alpha$ outputs $b_\alpha$.
The output of the experiment is $b_\alpha$.
\end{enumerate}
\end{description}

\begin{description}
\item[$\hybi{2}$:]This is the same as $\hybi{1}$ except that the tuple $(\ck,\com,\comrand)$ is generated as $(\ck,\com,\td)\la\EqSetup(1^\secp)$ and $\comrand\la\Open(\td,\maskval,\com)$.
\end{description}

From the trapdoor equivocation property of $\Com$, we have $\abs{\Pr[\hybi{1}=1]-\Pr[\hybi{2}=1]}=\negl(\secp)$.

\begin{description}
\item[$\hybi{3}$:]This is the same as $\hybi{2}$ except that $\maskval$ is replaced with $\maskval\oplus\uk$. By this change, $\qfhe.\ct$ and $\comrand$ are generated as $\qfhe.\ct\la\QFHE.\Enc(\pk,(\com,\maskval))$ and $\comrand\la\Open(\td,\maskval\oplus\uk,\com)$. Moreover, $\qA_\alpha$ is given $x=(\maskval\oplus\uk)\|\comrand$.
\end{description}

We have $\Pr[\hybi{2}=1]=\Pr[\hybi{3}=1]$.

\begin{description}
\item[$\hybi{4}$:]This is the same as $\hybi{3}$ except that $\ue.\qct$ is generated as $\ue.\qct\la\UE.\qEnc(\uk,0^\secp)$.
\end{description}

We consider the following adversary $\qB=(\qB_0,\qB_1,\cdots,\qB_\numa)$ that attacks the one-out-of-many one-time unclonable-indistinguishable security of $\UE$. $\qB_0$ behaves as follows.
\begin{enumerate}
\item $\qB_0$ sends $(\lock,0^\secp)$ to the challenger, where $\lock\la\bit^\secp$ and obtains $\ue.\qct$ from the challenger.
Then $\qB_0$ sends $\rho$ generated as follows to $\qA_0$.
    \begin{itemize}
    \item Generate $(\pk,\sk)\la\QFHE.\KG(1^\secp)$ and $(\ck,\com,\td)\la\EqSetup(1^\secp)$.
    \item $\maskval\la\bit^{\uekeylen}$.
\item Generate $\qfhe.\ct\la\QFHE.\Enc(\pk,(\com,\maskval))$.
\item Generate $\tlD\la\CCObf(1^\secp,\QFHE.\Dec(\sk,\cdot),\lock)$.
\item Set $\rho=(\ck,\qfhe.\ct,\ue.\qct,\tlD)$.
    \end{itemize}
\item When $\qA_0$ creates a quantum state $\qstateq$ over $\numa$ registers $\qreg{R}_1,\cdots,\qreg{R}_\numa$, $\qB_0$ sends $(\com,\td,\maskval,\qstateq[\qreg{R}_i])$ to $\qB_i$ for every $i\in[\numa]$.
\end{enumerate}
$\qB_\alpha$ behaves as follows, where $\alpha\la[\numa]$ is chosen by the challenger.
\begin{enumerate}
\item $\qB_\alpha$ send $\qstateq[\qreg{R}_\alpha]$ to $\qA_\alpha$.
\item When $\qB_\alpha$ obtains $\uk$ from the challenger, it computes $\comrand\la\Open(\td,\maskval\oplus\uk,\com)$ and sends $x=\maskval\oplus\uk\|\comrand$ to $\qA_\alpha$.
\item When $\qA_\alpha$ outputs $b_\alpha$, it outputs $\coin_\alpha=b_\alpha$.
\end{enumerate}

Let the challenge bit in the security game played by $\qB$ be $\coin$.
If $\coin=0$, $\qB$ perfectly simulates $\hybi{3}$ to $\qA$.
If $\coin=1$, $\qB$ perfectly simulates $\hybi{4}$ to $\qA$.
Also, $\qB$ outputs $\qA$'s output.
Then, we have 
\begin{align}
\Pr[\coin_\alpha=\coin]-\frac{1}{2}
&=\frac{1}{2}(\Pr[\coin_\alpha=1 \mid \coin=0]-\Pr[\coin_\alpha=1 \mid \coin=1])\\
&=\frac{1}{2}(\Pr[\hybi{3}=1]-\Pr[\hybi{4}=1]).
\end{align}
Thus, from the one-out-of-many one-time unclonable-indistinguishable security of $\UE$, we have $\frac{1}{2}(\Pr[\hybi{3}=1]-\Pr[\hybi{4}=1])\le\frac{1}{2\numa}+\negl(\secp)$.

\begin{description}
\item[$\hybi{5}$:]This is the same as $\hybi{4}$ except that $\tlD$ is generated as $\tlD\la\CC.\Sim(1^\secp,\pp_{\QFHE.\Dec})$, where $\pp_{\QFHE.\Dec}$ is the parameters of $\QFHE.\Dec$.
\end{description}

From the security of $\CCObf$, we have $\abs{\Pr[\hybi{4}=1]-\Pr[\hybi{5}=1]}=\negl(\secp)$.

\begin{description}
\item[$\hybi{6}$:]This is the same as $\hybi{5}$ except that $\qfhe.\ct$ is generated as $\qfhe.\ct\la\QFHE.\Enc(\pk,0^L)$, where $L=\abs{\com}+\uekeylen$.
\end{description}

From the security of $\QFHE$, we have $\abs{\Pr[\hybi{5}=1]-\Pr[\hybi{6}=1]}=\negl(\secp)$.

\begin{description}
\item[$\hybi{7}$:]This is the same as $\hybi{6}$ except that $\maskval$ is replaced with $\maskval\oplus\uk$. By this change, $\comrand$ is generated as $\comrand\la\Open(\ck,\maskval,\com)$. Moreover, $\qA_\alpha$ is given $y:=\maskval\|\comrand$.
\end{description}

We have $\Pr[\hybi{6}=1]=\Pr[\hybi{7}=1]$.

\begin{description}
\item[$\hybi{8}$:]This is the same as $\hybi{7}$ except that $\ck$ is generated as $\ck\la\Setup(1^\secp)$ and $\comrand$ is generated uniformly at random.
\end{description}

From the trapdoor equivocation property of $\Com$, we have $\abs{\Pr[\hybi{7}=1]-\Pr[\hybi{8}=1]}=\negl(\secp)$.

\begin{description}
\item[$\hybi{9}$:] This is the same as $\hybi{8}$ except that a re-sampled value $x\la U_{\pfinplen}$ is given to $\qA_\alpha$ instead of $y=\maskval\|\comrand$.
\end{description}

In $\hybi{8}$, $\rho$ given to $\qA_0$ is independent of $y=\maskval\|\comrand$ and $y$ is uniformly at random, and thus, we have $\Pr[\hybi{8}=1]=\Pr[\hybi{9}=1]$.

\begin{description}
\item[$\hybi{10}$:]This is the same as $\hybi{9}$ except that we generate $\rho$ in the same way as $\hybi{1}$.
\end{description}

We can prove $\abs{\Pr[\hybi{9}=1]-\Pr[\hybi{10}=1]}=\negl(\secp)$ by using the security of $\CCObf$, $\QFHE$, $\UE$, and $\Com$, essentially undoing the changes between $\hybi{1}$ and $\hybi{9}$.
To make this change, we can rely on one-time IND-CPA security of $\UE$, not one-out-of-many one-time unclonable-indistinguishable security for the following reason.
In $\hybi{9}$ and $\hybi{10}$, $\qA_\alpha$ is given a re-sample value $x$ and not $y=\maskval\|\comrand$.
Then, we can ensure that the information of $\uk$ is not given to $\qA$ except $\ue.\qct$ in this transition, which allows us to use one-time IND-CPA security of $\UE$.
Note that the one-time IND-CPA security of $\UE$ is implied by the one-out-of-many one-time unclonable-indistinguishable security of $\UE$ as proven in \cref{thm:omind_to_ind}.

$\hybi{10}$ is $\expc{\CP,D_Y,D_X,\qA}{cp}{pf1}{om}(\secp,\numa)$ where $D_X(y)=\Dresamp{w}(y)$ outputs a resampled $x\la U_{\pfinplen}$ and the output is set to $\qA_\alpha$'s output.
Let $\Resamp$ be the event that $D_X(y)=\Dresamp{w}(y)$ does re-sampling.
Then, we have 
\begin{align}
&\advc{\CP,\qA}{cp}{pf1}{om}(\secp)-w\\
&\le (1-w) \cdot \Pr[b_\alpha=1\mid \lnot\Resamp]+w\cdot\Pr[b_\alpha=0\mid \Resamp]+\negl(\secp)-w\\
&=(1-w)\cdot\Pr[b_\alpha=1\mid \lnot\Resamp]-w\cdot(1-\Pr[b_\alpha=0\mid \Resamp])+\negl(\secp)\\
&\le(1-w)\cdot\Pr[b_\alpha=1\mid \lnot\Resamp]-(1-w)\cdot\Pr[b_\alpha=1\mid \Resamp]+\negl(\secp)\\
&\le (1-w)\cdot (\Pr[\hybi{1}=1]-\Pr[\hybi{10}=1])+\negl(\secp).
\end{align}
The third inequality uses $1-w \le w$.
From the above discussions, we have $\Pr[\hybi{1}]-\Pr[\hybi{10}]\le\frac{1}{\numa}+\negl(\secp)$.
Therefore, we have
\begin{align}
\advc{\CP,\qA}{cp}{pf1}{om}(\secp)
&\le w+(1-w)\cdot\frac{1}{\numa}+\negl(\secp)\\
&=\frac{1}{\numa}\cdot 1+\frac{\numa-1}{\numa}\cdot w +\negl(\secp).
\end{align}
This completes the proof.
\ifnum\llncs=0
\end{proof}
\else\fi

\else

We provide the proof of \cref{thm:CP_from_lattice_security_onesided} in \cref{sec:proof_CP}, and omit the proof of~\cref{thm:CP_from_lattice_security}.

\fi

% !TEX root = main.tex

\newcommand{\tlDqfhe}{\tlD_{\mathsf{qfhe}}}

\section{Unclonable Predicate Encryption}\label{sec:UPE}
We introduce unclonable predicate encryption (PE) and present how to achieve it in this section.

\subsection{Definition}
The definition of unclonable PE is a natural extension of PE to an unclonable variant.
\begin{definition}[Unclonable Predicate Encryption (Syntax)]\label{def:unclonable_pe}
An unclonable predicate encryption scheme is a tuple of quantum algorithms $(\Setup,\KG,\allowbreak\qencrypt,\qdecrypt)$ with plaintext space $\Ms$.
\begin{description}
\item [$\Setup(1^\secp)\ra (\mpk,\msk)$:] The setup algorithm takes as input the security 
    parameter $1^\secp$ and outputs a key pair $(\mpk,\msk)$.
    \item [$\KG(\msk,x)\ra \sk_x$:] The key generation algorithm takes as input the master secret key $\msk$ and an attribute $x$, and outputs a decryption key $\sk_x$.
    \item [$\qEnc(\mpk,C,\msg)\ra \qct$:] The encryption algorithm takes as input $\mpk$, predicate $C$, and a plaintext $\msg\in\cM$, and outputs a ciphertext $\qct$.
    \item [$\qDec(\sk_x,\qct)\ra \msg^\prime \mbox{ or }\bot$:] The decryption algorithm takes as input $\sk_x$ and $\qct$ and outputs a plaintext $\msg^\prime$ or $\bot$.

\item[Decryption Correctness:] There exists a negligible function $\negl$ such that for any $\msg\in\Ms$, $C$, and $x$ such that $C(x)=1$, we have
\begin{align}
\Pr\left[
\qDec(\sk_x,\qct)= \msg
\ \middle |
\begin{array}{ll}
(\mpk,\msk)\la\Setup(1^\secp)\\
\sk_x\gets \KG(\msk,x)\\
\qct \gets \qEnc(\mpk,C,\msg)
\end{array}
\right] 
=1-\negl(\secp).
\end{align}
\end{description}
\end{definition}

We define simulation-based security for unclonable PE, and then discuss its validity.

\begin{definition}[Adaptive Unclonable-Simulation Security for PE]\label{def:adp_unclonable_sim_pe}
Let $\UPE=(\Setup,\KG, \qEnc, \qDec)$ be an unclonable predicate encryption scheme.
We consider the following security experiment $\expc{\UPE,\qSim,\qA}{ada}{sim}{clone}(\secp)$, where $\qSim$ is a QPT simulation algorithm and $\qA=(\qA_0,\qA_1,\qA_2)$.

\begin{enumerate}
    \item The challenger generates $(\mpk,\msk)\gets \Setup(1^\secp)$ and sends $\mpk$ to $\qA_0$.
    \item $\qA_0$ can get access to the following oracle.
            \begin{description}
            \item[$\Oracle{\KG,1}(x)$:] Given $x$, it returns $\sk_x\la\KG(\msk,x)$. 
            \end{description}
    \item $\qA_0$ sends $C$ and $\msg \in \Ms$ to thfe challenger, where $C$ satisfies $C(x)=0$ for all $x$ queried by $\qA_0$ in the previous step. The challenger picks $\coin\la\bit$ and does the following.
    \begin{itemize}
    \item If $\coin =0$, the challenger generates $\qct \gets \qEnc(\mpk,C,\msg)$ and returns $\qct$ to $\qA_0$.
    \item If $\coin=1$, the challenger generates $\qct \gets \qSim(1^\secp,\abs{C},\abs{\msg})$ and returns $\qct$ to $\qA_0$.
    \end{itemize}
    Hereafter, $\qA_0$ is not allowed to query $x$ such that $C(x)=1$ to $\Oracle{\KG,1}$.
    \item $\qA_0$ creates a bipartite state $\qstateq$ over registers $\qreg{R}_1$ and $\qreg{R}_2$. $\qA_0$ sends $\qstateq[\qreg{R}_1]$ and $\qstateq[\qreg{R}_2]$ to $\qA_1$ and $\qA_2$, respectively.
\item $\qA_1$ and $\qA_2$ can get access to the following oracle.
    \begin{description}
            \item[$\Oracle{\KG,2}(x)$:] Given $x$, it returns $\sk_x\la\KG(\msk,x)$. Note that $\Oracle{\KG,2}$ accepts a query $x$ such that $C(x)=1$.
            \end{description}
            \item $\qA_1$ and $\qA_2$ respectively output $\coin_1^\prime$ and $\coin_2^\prime$.
If $\coin_i^\prime = \coin$ for $i\in \setbk{1,2}$, the challenger outputs $1$, otherwise outputs $0$.    
\end{enumerate}
We say that $\UPE$ is unclonable-simulation secure if there exists QPT $\qSim$ such that for any QPT $\qA$, it holds that
\begin{align}
\advc{\UPE,\qSim,\qA}{ada}{sim}{clone}(\secp)\seteq \Pr[ \expc{\UPE,\qA}{ada}{sim}{clone}(\secp)=1] \le \frac{1}{2} + \negl(\secp).
\end{align}
\end{definition}
\begin{remark}
We select a simulation-based security definition since it intuitively captures security of encryption and stronger than indistinguishability-based one.
There are impossibility results of simulation-based secure FE~\cite{TCC:BonSahWat11,C:AGVW13}.
However, those are not applied to our setting since we consider the single challenge ciphertext setting of PE, where the message (a.k.a payload) part is recovered by a secret key.
\end{remark}

\begin{remark}[On the validity of unclonable-simulation security]
We claim that our unclonable-simulation security for PE captures unclonability of both the payload part and the predicate part.
To this end, we first argue that by using both the payload part and the predicate part, we can realize SKUE satisfying unclonable-simulation security where $\qA_0$ is given a real or simulated challenge ciphertext in the security game.
We then discuss about the validity of unclonable-simulation security for SKUE.

It is rather clear that we can achieve unclonable-simulation secure SKUE using the payload part. 
Then, we show how to construct an unclonable-simulation secure SKUE scheme for the message space $\setbk{1,...N}$ using the predicate part, where $N$ is a polynomial in $\secp$.
Let $C[i]$ be a predicate that takes as input $j$ and output $1$ if and only if $i=j$.
In this construction, when we encrypt a message $i^\ast$, we generate a ciphertext of the payload $\msg=1$  with the predicate $C[i^\ast]$ by the unclonable PE scheme.
The decryption key of this construction is $(\sk_1,\ldots,\sk_N)$, where $\sk_j$ is the decryption key for the attribute $j$.
Decryption is done by testing if a ciphertext can be decrypted by $\sk_j$ for every $j$.
The unclonable-simulation security of this SKUE scheme is reduced to that of the unclonable PE scheme.

We now discuss the validity of unclonable-simulation security of SKUE.
Although we do not see the formal implication from unclonable-simulation security to unclonable-indistinguishable security, we can use any SKUE scheme with the former security notion as that with the latter security notion whose message space is $\bit$.
This is done by encoding $1$-bit messages using real ciphertext and simulated ciphertext.
Also, we can formally prove that unclonable-simulation security is strictly stronger than one-wayness-based unclonability, similarly to unclonable-indistinguishability.
In the proof of the implication from unclonable-indistinguishability to one-wayness-based unclonability, the argument goes though by using the former to switch a real ciphertext into a junk ciphertext (such as a ciphertext of $0$).
A similar argument can be done by using the power of simulation-based security that is used to switch a real ciphertext into a simulated one.
In general, we can use unclonable-simulation security as a drop-in replacement of unclonable-indistinguishable security, if the latter is used to switch a real ciphertext into a junk one.
For example, this is the case in our construction of copy-protection for single-bit output point functions presented in \cref{sec:CP}, which proves the usefulness of unclonable-simulation security.
\end{remark}

We propose one-out-of-many variant of unclonable-simulation security for PE.

\begin{definition}[One-out-of-Many Adaptive Unclonable-Simulation Security for PE]\label{def:om_adp_unclonable_sim_pe}
Let $\UPE=(\Setup,\KG, \qEnc, \qDec)$ be an unclonable PE scheme.
We consider the following security experiment $\expd{\UPE,\qSim,\qA}{om}{ada}{sim}{clone}(\secp,\numa)$, where $\qSim$ is a QPT simulation algorithm and $\qA=(\qA_0,\qA_1,\cdots,\qA_\numa)$.

\begin{enumerate}
    \item The challenger generates $(\mpk,\msk)\gets \Setup(1^\secp)$ and sends $\mpk$ to $\qA_0$.
    \item $\qA_0$ can get access to the following oracle.
            \begin{description}
            \item[$\Oracle{\KG,1}(x)$:] Given $x$, it returns $\sk_x\la\KG(\msk,x)$. 
            \end{description}
    \item $\qA_0$ sends $C$ and $\msg \in \Ms$ to thfe challenger, where $C$ satisfies $C(x)=0$ for all $x$ queried by $\qA_0$ in the previous step. The challenger picks $\coin\la\bit$ and does the following.
    \begin{itemize}
    \item If $\coin =0$, the challenger generates $\qct \gets \qEnc(\mpk,C,\msg)$ and returns $\qct$ to $\qA_0$.
    \item If $\coin=1$, the challenger generates $\qct \gets \qSim(1^\secp,\abs{C},\abs{\msg})$ and returns $\qct$ to $\qA_0$.
    \end{itemize}
    Hereafter, $\qA_0$ is not allowed to query $x$ such that $C(x)=1$ to $\Oracle{\KG,1}$.
    \item $\qA_0$ creates a quantum state $\qstateq$ over $\numa$ registers $\qreg{R}_1,\cdots,\qreg{R}_\numa$. $\qA_0$ sends $\qstateq[\qreg{R}_i]$ to $\qA_i$ for every $i\in[\numa]$.
\item The challenger generates $\alpha\chosen [\numa]$. $\qA_\alpha$ can access the following oracle.
    \begin{description}
            \item[$\Oracle{\KG,2}(x)$:] Given $x$, it returns $\sk_x\la\KG(\msk,x)$. Note that $\Oracle{\KG,2}$ accepts a query $x$ such that $C(x)=1$.
            \end{description}
            \item $\qA_\alpha$ outputs $\coin^\prime$.
If $\coin^\prime = \coin$, the challenger outputs $1$, otherwise outputs $0$.
\end{enumerate}
We say that $\UPE$ is one-out-of-many unclonable-simulation secure if there exists QPT $\qSim$ such that for any polynomial $\numa=\numa(\secp)$ and QPT $\qA=(\qA_0,\qA_1,\cdots,\qA_\numa)$, it holds that
\begin{align}
\advd{\UPE,\qSim,\qA}{om}{ada}{sim}{clone}(\secp,\numa)\seteq \Pr[ \expd{\UPE,\qA}{om}{ada}{sim}{clone}(\secp,\numa)=1]
\le \frac{1}{2}+\frac{1}{2\numa}+\negl(\secp).
\end{align}
\end{definition}

\begin{remark}
We can consider selective variants of~\cref{def:om_adp_unclonable_sim_pe,def:adp_unclonable_sim_pe}, where $\qA_0$ declares $C$ at the beginning of the games. We call selective unclonable-simulation security and one-out-of-many selective unclonable-simulation security, respectively.
\end{remark}

\subsection{Construction}\label{sec:construction_unclonable-PE_from_copy-protection}

We construct an unclonable PE scheme $\UPE$ using the following tools.
\begin{itemize}
\item SKUE $\UE=(\UE.\KG,\UE.\qEnc,\UE.\qDec)$. Suppose the plaintext space of $\UE$ is $\bit^\secp$.
\item Compute-and-compare obfuscation $\CCObf$ with the simulator $\CC.\Sim$.
\item QFHE with classical ciphertexts $\QFHE=(\QFHE.\KG,\QFHE.\Enc,\QFHE.\qEval,\QFHE.\Dec)$.
\item AD-SIM secure CP-ABE $\CPABE=(\Setup,\KG,\Enc,\Dec)$ with a QPT simulator $\qABESim=(\qSimEnc,\qSimKG)$.
\end{itemize}
The description of $\UPE$ is as follows.

\begin{description}
 \item[$\UPE.\Setup(1^\secp)$:] $ $
 \begin{itemize}
\item Output $(\mpk,\msk)\la\Setup(1^\secp)$.
 \end{itemize}
 \item[$\UPE.\KG(\msk,x)$:] $ $
\begin{itemize}
\item Output $\sk_x\la\KG(\msk,x)$.
\end{itemize}
 \item[$\UPE.\Enc(\mpk,C,\msg)$:] $ $
 \begin{itemize}
\item Generate $(\pk,\sk)\la\QFHE.\KG(1^\secp)$ and $\uk\la\UE.\KG(1^\secp)$.
\item Generate $\lock\la\bit^\secp$.
\item Generate $\ct\la\Enc(\mpk,C,\uk)$.
\item Let $D(\cdot)$ be the decryption circuit $\Dec(\cdot,\ct)$ of $\CPABE$ that has $\ct$ hardwired.
\item Generate $\qfhe.\ct\la\QFHE.\Enc(\pk,D)$.
\item Generate $\ue.\qct\la\UE.\qEnc(\uk,\lock)$.
\item Generate $\tlDqfhe\la\CCObf(1^\secp,\QFHE.\Dec(\sk,\cdot),\lock,\msg)$.
\item Output $\upe.\qct\la(\qfhe.\ct,\ue.\qct,\tlDqfhe)$.
\end{itemize}
\item[$\UPE.\Dec(\sk_x,\upe.\qct)$:] $ $
\begin{itemize}
 \item Parse $\rho = (\qfhe.\ct,\ue.\qct,\tlDqfhe)$.
 \item Compute $\evalct\la\QFHE.\qEval(C[\sk_x],\ue.\qct,\qfhe.\ct)$, where the circuit $C[\sk_x]$ is described in \cref{fig:c_pe}.
 \item Output $\msg\la\tlDqfhe(\evalct)$.
\end{itemize}
\end{description}

\protocol
{Quantum circuit $C[\sk_x]$}
{The description of $C[\sk_x]$}
{fig:c_pe}
{
\begin{description}
\setlength{\parskip}{0.3mm} % between paragraph
\setlength{\itemsep}{0.3mm} % between items
\item[Constants:]  A string $\sk_x$.
\item[Input:] A quantum state $\ue.\qct$ and a circuit description $D$.
\end{description}
\begin{enumerate}
\item Compute $\uk^\prime \la D(\sk_x)$. 
\item Output $\lock^\prime\la\UE.\qDec(\uk^\prime,\ue.\qct)$.
\end{enumerate}
}

The correctness of $\UPE$ immediately follows from the correctness of the building blocks.

\paragraph{Security.}
We have the following theorems.

\begin{theorem}\label{thm:UPE_from_copy-protection_onesided}
Let $D=\{D_\secp\}$ be a family of distributions where each $D_\secp$ outputs $(\QFHE.\Dec(\sk,\cdot),\lock,\qaux:=\pk)$ generated as those in $\UPE.\Enc$.
If $\CCObf$ is secure with respect to $D$, $\QFHE$ satisfies semantic security, $\UE$ satisfies one-out-of-many one-time unclonable-indistinguishable security, and $\CPABE$ is AD-SIM secure, then $\UPE$ satisfies one-out-of-many adaptive unclonable-simulation security.
\end{theorem}

\begin{theorem}\label{thm:UPE_from_copy-protection_ABE}
In \cref{thm:UPE_from_copy-protection_onesided}, if we use a one-time unclonable-indistinguishable secure SKFE scheme $\UE$, then $\UPE$ satisfies adaptive unclonable-simulation security.
\end{theorem}

Similarly to our copy-protection scheme presented in \cref{sec:CP}, when we instantiate $\UPE$ based on \cref{thm:UPE_from_copy-protection_ABE}, we need to be careful about the use of QRO by the existing SKUE schemes.
See \cref{rem:instantiation_QROM_UE}.

If our goal is selective unclonable-simulation security or one-out-of-many selective unclonable-simulation security, we can use SEL-SIM secure CP-ABE instead of AD-SIM secure CP-ABE in~\cref{thm:UPE_from_copy-protection_ABE,thm:UPE_from_copy-protection_onesided}.

\ifnum\llncs=0
We below prove \cref{thm:UPE_from_copy-protection_onesided} and omit the proof of~\cref{thm:UPE_from_copy-protection_ABE}. It is easy to see that we can similarly prove~\cref{thm:UPE_from_copy-protection_ABE} by using one-time unclonable-indistinguishable security at the transition from $\hybi{2}$ to $\hybi{3}$ in the proof of~\cref{thm:UPE_from_copy-protection_onesided}.

% !TEX root = main.tex

%\begin{proof}

\ifnum\llncs=0
\begin{proof}[Proof of~\cref{thm:UPE_from_copy-protection_onesided}]
\else
\section{Proof of~\cref{thm:UPE_from_copy-protection_onesided}}\label{sec:proof_UPE}
\fi
Let $\numa$ be any polynomial of $\secp$ and $\qA=(\qA_0,\ldots,\qA_\numa)$ any efficient adversary. Also, let $\UPE.\qSim$ be the following algorithm.
\begin{description}
 \item[$\UPE.\qSim(1^\secp,\abs{C},\abs{\msg})$:] $ $
 \begin{itemize}
\item Generate $(\pk,\sk)\la\QFHE.\KG(1^\secp)$ and $\uk\la\UE.\KG(1^\secp)$.
\item Generate $\lock\la\bit^\secp$.
%\item Let $D(\cdot)$ be the decryption circuit $\Dec(\cdot,\ct)$ of $\CPABE$ that has $\ct$ hardwired.
\item Generate $\qfhe.\ct\la\QFHE.\Enc(\pk,0^L)$, where $L$ is the size of $\CPABE$'s decryption circuit $\Dec(\cdot,\ct)$ that has hardwired ciphertext of a $\secp$-bit message with a $\abs{C}$-bit predicate.
\item Generate $\ue.\qct\la\qEnc(\uk,0^\secp)$.
\item Generate $\tlDqfhe\la\CC.\Sim(1^\secp,\pp_{\QFHE.\Dec},\abs{\msg})$, where $\pp_{\QFHE.\Dec}$ is the parameters of $\QFHE.\Dec$.
\item Output $\upe.\qct\la(\qfhe.\ct,\ue.\qct,\tlDqfhe)$.
 \end{itemize}
% Note that we need only $\secp$ and $\abs{C}$ to compute $\Dec(\cdot,\ct).\pp$ since $\secp$ and $\abs{C}$ determine the plaintex and ciphertext length of $\CPABE$ and $\abs{\Dec(\cdot,\ct)}$.
\end{description}
We prove $\advd{\UPE,\UPE.\qSim,\qA}{om}{ada}{sim}{clone}(\secp,\numa)\le \frac{1}{2}+\frac{1}{2\numa}+\negl(\secp)$ for any polynomial $\numa$ and QPT $\qA=(\qA_0,\qA_1,\cdots,\qA_\numa)$ using the following sequence of experiments.
\begin{description}
\item[$\hybi{1}$:]This is $\expd{\UPE,\UPE.\qSim,\qA}{om}{ada}{sim}{clone}(\secp,\numa)$ where $\coin=0$ and the output of the experiment is set to the adversary's output. The detailed description is as follows.
\begin{enumerate}
    \item The challenger generates $(\mpk,\msk)\gets \Setup(1^\secp)$ and sends $\mpk$ to $\qA_0$.
    \item $\qA_0$ can get access to the following oracle.
            \begin{description}
            \item[$\Oracle{\KG,1}(x)$:] Given $x$, it returns $\sk_x\la\KG(\msk,x)$. 
            \end{description}
    \item $\qA_0$ sends $C$ and $\msg \in \Ms$ to the challenger, where $C$ satisfies $C(x)=0$ for all $x$ queried by $\qA_0$ in the previous step. The challenger returns $\upe.\qct$ generated as follows.
     \begin{itemize}
\item Generate $(\pk,\sk)\la\QFHE.\KG(1^\secp)$ and $\uk\la\UE.\KG(1^\secp)$.
\item Generate $\lock\la\bit^\secp$.
\item Generate $\ct\la\Enc(\mpk,C,\uk)$.
\item Let $D(\cdot)$ be the decryption circuit $\Dec(\cdot,\ct)$ of $\CPABE$ that has $\ct$ hardwired.
\item Generate $\qfhe.\ct\la\QFHE.\Enc(\pk,D)$.
\item Generate $\ue.\qct\la\qEnc(\uk,\lock)$.
\item Generate $\tlDqfhe\la\CCObf(1^\secp,\QFHE.\Dec(\sk,\cdot),\lock,\msg)$.
\item set $\upe.\qct\la(\qfhe.\ct,\ue.\qct,\tlDqfhe)$.
\end{itemize}
    Hereafter, $\qA_0$ is not allowed to query $x$ such that $C(x)=1$ to $\Oracle{\KG,1}$.
    \item $\qA_0$ creates a quantum state $\qstateq$ over $n$ registers $\qreg{R}_1,\ldots,\qreg{R}_\numa$. $\qA_0$ sends $\qstateq[\qreg{R}_i]$ to $\qA_i$ for every $i\in [\numa]$.
\item The challenger generates $\alpha\chosen [\numa]$. $\qA_\alpha$ can get access to the following oracle.
    \begin{description}
            \item[$\Oracle{\KG,2}(x)$:] Given $x$, it returns $\sk_x\la\KG(\msk,x)$. Note that $\Oracle{\KG,2}$ accepts a query $x$ such that $C(x)=1$.
            \end{description}
            \item $\qA_\alpha$ outputs $\coin_\alpha^\prime$. The final output of the experiment is $\coin_\alpha^\prime$.    
\end{enumerate}
\end{description}

\begin{description}
\item[$\hybi{2}$:]This is the same as $\hybi{1}$ except the following changes.
\begin{itemize}
\item $\ct$ is generated as $(\ct,\state)\la\qSimEnc(\mpk,C)$ instead of $\ct \la \Enc(\mpk,C,\uk)$.
\item $\Oracle{\KG,2}$ returns $\sk_x\la\qSimKG(\msk,\state,x,\uk)$ given $x$ if $C(x)=1$.
\end{itemize}
\end{description}
From the AD-SIM security of $\CPABE$, we have $\abs{\Pr[\hybi{1}=1]-\Pr[\hybi{2}=1]}=\negl(\secp)$.

\begin{description}
\item[$\hybi{3}$:]This is the same as $\hybi{2}$ except that $\ue.\qct$ is generated as $\ue.\qct\la\UE.\qEnc(\uk,0^\secp)$.
\end{description}

We consider the following adversary $\qB=(\qB_0,\qB_1,\cdots,\qB_\numa)$ that attacks the one-out-of-many one-time unclonable-indistinguishable security of $\UE$. $\qB_0$ behaves as follows.
\begin{enumerate}
    \item $\qB_0$ generates $(\mpk,\msk)\gets \Setup(1^\secp)$ and sends $\mpk$ to $\qA_0$.
    \item $\qB_0$ simulates $\Oracle{\KG,1}(x)$ for $\qA_0$ using $\msk$.
    \item When $\qA_0$ outputs $C$ and $\msg \in \Ms$, where $C$ satisfies $C(x)=0$ for all $x$ queried by $\qA_0$ in the previous step, $\qB_0$ outputs $(\lock,0^\secp)$, obtains $\ue.\qct$, and returns $\upe.\qct$ generated as follows, where $\lock\la\bit^\secp$.
     \begin{itemize}
\item Generate $(\pk,\sk)\la\QFHE.\KG(1^\secp)$ and $\uk\la\UE.\KG(1^\secp)$.
\item Generate $\lock\la\bit^\secp$.
\item Generate $(\ct,\state)\la\qSimEnc(\mpk,C)$.
\item Let $D(\cdot)$ be the decryption circuit $\Dec(\cdot,\ct)$ of $\CPABE$ that has $\ct$ hardwired.
\item Generate $\qfhe.\ct\la\QFHE.\Enc(\pk,D)$.
\item Generate $\tlDqfhe\la\CCObf(1^\secp,\QFHE.\Dec(\sk,\cdot),\lock,\msg)$.
\item set $\upe.\qct\la(\qfhe.\ct,\ue.\qct,\tlDqfhe)$.
\end{itemize}
    Hereafter, $\qA_0$ is not allowed to query $x$ such that $C(x)=1$ to $\Oracle{\KG,1}$.
    \item When $\qA_0$ outputs a quantum state $\qstateq$ over $n$ registers $\qreg{R}_1,\ldots,\qreg{R}_\numa$, $\qB_0$ sends $(\msk,\state,\qstateq[\qreg{R}_i])$ to $\qB_i$ for every $i\in [\numa]$.
    \end{enumerate}
    
$\qB_\alpha$ behaves as follows, where $\alpha\la[\numa]$ is chosen by the challenger.
\begin{enumerate}
\item Given $\uk$ as an input, $\qB_\alpha$ first send $\qstateq[\qreg{R}_\alpha]$ to $\qA_\alpha$.
\item $\qB_\alpha$ simulates $\Oracle{\KG,2}(x)$ for $\qA_\alpha$ as follows.
\begin{description}
\item[$\Oracle{\KG,2}(x)$:] Given $x$, $\qB_\alpha$ returns $\sk_x\la\qSimKG(\msk,\state,x,\uk)$.
\end{description}
\item When $\qA_\alpha$ outputs $\coin_\alpha^\prime$, $\qB_\alpha$ outputs $b^\prime=\coin_\alpha^\prime$.
\end{enumerate}

Let the challenge bit in the security game played by $\qB$ be $b$.
If $b=0$, $\qB$ perfectly simulates $\hybi{2}$ to $\qA$.
If $b=1$, $\qB$ perfectly simulates $\hybi{3}$ to $\qA$.
Also, $\qB$ outputs $\qA$'s output.
Then, we have 
\begin{align}
\Pr[b^\prime=b]-\frac{1}{2}
&=\frac{1}{2}(\Pr[b^\prime=1|b=0]-\Pr[b^\prime=1|b=1])\\
&=\frac{1}{2}(\Pr[\hybi{2}=1]-\Pr[\hybi{3}=1]).
\end{align}
Thus, from the one-out-of-many one-time unclonable-indistinguishable security of $\UE$, we have $\frac{1}{2}(\Pr[\hybi{2}=1]-\Pr[\hybi{3}=1])\le\frac{1}{2\numa}+\negl(\secp)$.

\begin{description}
\item[$\hybi{4}$:]This is the same as $\hybi{3}$ except that $\tlDqfhe$ is generated as $\tlDqfhe\la\CC.\Sim(1^\secp,\pp_{\QFHE.\Dec},\abs{\msg})$, where $\pp_{\QFHE.\Dec}$ is the parameters of $\QFHE.\Dec$.
\end{description}

From the security of $\CCObf$, we have $\abs{\Pr[\hybi{3}=1]-\Pr[\hybi{4}=1]}=\negl(\secp)$.

\begin{description}
\item[$\hybi{5}$:]This is the same as $\hybi{4}$ except that $\qfhe.\ct$ is generated as $\qfhe.\ct\la\QFHE.\Enc(\pk,0^L)$, where $L$ is the size of $\CPABE$'s decryption circuit $\Dec(\cdot,\ct)$ that has hardwired ciphertext of a $\secp$-bit message with a $\abs{C}$-bit predicate.
\end{description}

From the security of $\QFHE$, we have $\abs{\Pr[\hybi{4}=1]-\Pr[\hybi{5}=1]}=\negl(\secp)$.

$\hybi{5}$ is $\expd{\UPE,\UPE.\qSim,\qA}{om}{ada}{sim}{clone}(\secp,\numa)$ where $\coin=1$ and the output of the experiment is set to $\qA_\alpha$'s output.
Also, we have 
\begin{align}
&\advd{\UPE,\UPE.\qSim,\qA}{om}{ada}{sim}{clone}(\secp,\numa)-\frac{1}{2}\\
&=\frac{1}{2}(\Pr[\coin^\prime=1|\coin=0]-\Pr[\coin^\prime=1|\coin=1])\\
&=\frac{1}{2}(\Pr[\hybi{1}=1]-\Pr[\hybi{5}=1]).
\end{align}
From the above discussions, we have $\frac{1}{2}(\Pr[\hybi{1}=1]-\Pr[\hybi{5}=1])\le \frac{1}{2\numa}+\negl(\secp)$, which means that $\advd{\UPE,\UPE.\qSim,\qA}{om}{ada}{sim}{clone}(\secp,\numa)\le\frac{1}{2}+\frac{1}{2\numa}+\negl(\secp)$.
This completes the proof.
\ifnum\llncs=0
\end{proof}
\else\fi

\begin{remark}
We can also consider a setting where $\qA_\alpha$ receives the master secret key $\msk$ of unclonable PE. If we use IO, we can achieve the stronger definition. The non-committing ABE scheme based on IO by Hiroka et al.~\cite{AC:HMNY21} achieves stronger security where the adversary is given a master secret key after a challenge ciphertext is given. If we use their scheme instead of our simulation-based secure ABE scheme, we can achieve the stronger security for unclonable PE.
\end{remark}

\else

We provide the proof of \cref{thm:UPE_from_copy-protection_onesided} in \cref{sec:proof_UPE} and omit the proof of~\cref{thm:UPE_from_copy-protection_ABE}.

\fi

\ifnum\anonymous=1
\else
% \section*{Acknowledgement}
\fi

	\ifnum\llncs=1
\bibliographystyle{alpha}
\bibliography{../abbrev3,../crypto,../siamcomp_jacm,../other-bib}
	\else
\bibliographystyle{alpha} 
\bibliography{abbrev3,crypto,siamcomp_jacm,other-bib}
	\fi

%%% Appendix part %%%

\ifnum\cameraready=0
	\ifnum\llncs=0
	%%%%%% Full version region %%%%%%
	\appendix
%%% appendix files here
	% !TEX root = main.tex

\newcommand{\INDCPABE}{\mathsf{IND}\textrm{-}\mathsf{CP}\textrm{-}\mathsf{ABE}}
\renewcommand{\IND}{\mathsf{IND}}

\section{AD-SIM secure CP-ABE}\label{sec:ADSIM_CPABE}
We show how to transform any AD-IND secure CP-ABE scheme into AD-SIM secure one in this section.
\subsection{Additional Tool}
We introduce pseudorandom ciphertext secret key encryption (SKE).
\begin{definition}[Pseudorandom Ciphertext SKE]\label{def:ske}
A SKE scheme $\SKE$ is a two tuple $(\E, \D)$ of PPT algorithms. 
\begin{itemize}
\item The encryption algorithm $\E$, given a key $k \in \bin^\lambda$ and a plaintext $m \in \cM$, outputs a ciphertext $c$,
where $\cM$ is the plaintext space of $\SKE$.
\item The decryption algorithm $\D$, given a key $k$ and a ciphertext $c$, outputs a plaintext $\tilde{m} \in \{ \bot \} \cup \cM$.
This algorithm is deterministic.
\end{itemize}
\begin{description}
\item[Correctness:] We require $\D(k, \E(k, m)) = m$ for every $m \in \cM$ and key $k \in \bin^\lambda$.
\item[Pseudorandom Ciphertext Property:]
Let $\zo{\ell}$ be the ciphertext space of $\SKE$.
We define the following experiment $\expb{\SKE,\qA}{pr}{ct}(1^\secp,\coin)$ between the challenger and an adversary $\qA$.
\begin{enumerate}
\item The challenger generates $k \chosen \bin^\lambda$.
Then, the challenger sends $1^\lambda$ to $\qA$.
\item $\qA$ may make polynomially many encryption queries adaptively.
$\qA$ sends $m \in \cM $ to the challenger.
Then, the challenger returns $c \la \E(k, m)$ if $\coin=0$, otherwise $c \chosen \zo{\ell}$.
\item $\qA$ outputs $\coin' \in \bin$.
%\end{description}
\end{enumerate}
We require that for any QPT adversary $\qA$, we have 
\[
\advb{\SKE, \qA}{pr}{ct}(\lambda) = \abs{\Pr[\expb{\SKE,\qA}{pr}{ct}(1^\secp,0)=1]-\Pr[\expb{\SKE,\qA}{pr}{ct}(1^\secp,1)=1]} \le \negl(\secp).
\]

\end{description}
\end{definition}

\begin{theorem}\label{thm:pseudorandom_ske}
If OWFs exist, there exists a pseudorandom-secure SKE scheme.
\end{theorem}

\subsection{Construction}
Since the plaintext space of AD-SIM secure ABE can be expanded by parallel repetition using different instances, we focus on constructing a scheme with the plaintext space $\bit$.

We construct $\CPABE=(\Setup,\KG,\Enc,\Dec)$ using the following tools.
\begin{itemize}
\item Compute-and-compare obfuscation $\CCObf$ with the simulator $\CC.\Sim$.
\item Ciphertext-policy ABE $\INDCPABE=(\IND.\Setup,\IND.\KG,\IND.\Enc,\IND.\Dec)$.
\item Pseudorandom ciphertext SKE $\SKE=(\SKE.\E,\SKE.\D)$.
\end{itemize}
The description is as follows.

\begin{description}
 \item[$\Setup(1^\secp)$:] $ $
 \begin{itemize}
\item Generate $(\ind.\mpk,\ind.\msk)\la\IND.\Setup(1^\secp)$.
\item Generate $R\la\bit^\secp$.
\item Output $\mpk\seteq(R,\ind.\mpk)$ and $\msk\seteq(R,\ind.\msk)$.
 \end{itemize}
 \item[$\KG(\msk,x)$:] $ $
\begin{itemize}
\item Parse $\msk\seteq(R,\ind.\msk)$.
\item Generate $c\la\bit^\ctlen$.
\item Generate $\ind.\sk_{(x,c)}\la\IND.\KG(\ind.\msk,x\|c)$.
\item Output $\sk_x\seteq\ind.\sk_{(x,c)}$.
\end{itemize}
 \item[$\Enc(\mpk,C,x,\msg\in\bit)$:] $ $
 \begin{itemize}
 \item Parse $\mpk\seteq(R,\ind.\mpk)$.
\item If $\msg=1$, generate $\ct$ as follows.
 \begin{itemize}
 \item Generate $k\la\bit^\secp$ and $\lock\la\bit^\secp$.
 \item Generate $\ind.\ct\la\IND.\Enc(\ind.\mpk,G[C,k,R],\lock)$, where $G[C,k,R]$ takes as input $(x,c)$ and output $1$ if and only if $C(x)=1$ and $\SKE.\D(k,c)\ne R$.
 \item Generate $\ct\la\CCObf(1^\secp,\IND.\Dec(\cdot,\ind.\ct),\lock,R)$.
 \end{itemize}
\item If $\msg=0$, generate $\ct\la\CC.\Sim(1^\secp,\pp_{\IND.\Dec},\abs{R})$.
\item Output $\ct$.
\end{itemize}
\item[$\Dec(\sk_x,\ct)$:] $ $
Output $1$ if $\ct(\sk_x)=R$ and $0$ otherwise.
\end{description}

\begin{theorem}
Let $D=\{D_{\mpk,C}\}$ be a family of distributions where each $D_{\mpk,C}$ outputs $(\IND.\Dec(\cdot,\ind.\ct),\lock,\allowbreak\qaux \seteq k)$ generated as follows.
\begin{itemize}
\item Parse $\mpk \seteq (\ind.\mpk,R)$
 \item Generate $k\la\bit^\secp$ and $\lock\la\bit^\secp$.
 \item Generate $\ct\la\IND.\Enc(\ind.\mpk,G[C,k,R],0^\secp)$.
 \item Output $(\IND.\Dec(\cdot,\ind.\ct),\lock,\qaux \seteq k)$.
\end{itemize}
If $\CCObf$ is secure with respect to $D$, $\INDCPABE$ is AD-IND secure, and $\SKE$ is a pseudorandom ciphertext SKE scheme, then $\CPABE$ satisfies AD-SIM security.
\end{theorem}
If our goal is SEL-SIM secure CP-ABE, we can use SEL-IND secure CP-ABE as a building block.
\begin{proof}
We first provide the construction of the simulator $\qSim=(\qSimEnc,\qSimKG)$.
We see that $\qSimEnc(\mpk,C)$ is exactly the same as $\Enc(\mpk,C,\msg=1)$.
\begin{description}
\item[$\qSimEnc(\mpk,C)$:] $ $
 \begin{itemize}
 \item Parse $\mpk\seteq (R,\ind.\mpk)$.
 \item Generate $k\la\bit^\secp$ and $\lock\la\bit^\secp$.
 \item Generate $\ct\la\IND.\Enc(\ind.\mpk,G[C,k,R],\lock)$, where $G[C,k,R]$ takes as input $(x,c)$ and output $1$ if and only if $C(x)=1$ and $\D(k,c)\ne R$.
 \item Generate $\ct\la\CCObf(1^\secp,\Dec(\cdot,\ct),\lock,R)$.
 \item Output $\ct$ and $\state\seteq k$.
 \end{itemize}
\item[$\qSimKG(\msk,\state, x,\msg\in\bit)$:] $ $
 \begin{itemize}
 \item Parse $\msk\seteq(R,\ind.\msk)$ and $\state\seteq k$.
 \item Generate $c\la\bit^\ctlen$ if $\msg=1$ and $c\la\SKE.\E(k,R)$ otherwise.
\item Generate $\sk_{(x,c)}\la\IND.\KG(\ind.\msk,x\|c)$.
\item Output $\sk_x\seteq\sk_{(x,c)}$.
 \end{itemize}
\end{description}

Let $\qA$ be any QPT adversary.
Conditioned that an adversary $\qA$ outputs $\msg=1$, $\expb{\CPABE,\qSim,\qA}{ad}{sim}(\secp, 0)$ and $\expb{\CPABE,\qSim,\qA}{ad}{sim}(\secp,1)$ are exactly the same experiment.
Thus, it is sufficient to prove $\Pr[1\la \expb{\CPABE,\qSim,\qA}{ad}{sim}(\secp, 0)]$ and $\Pr[1\la \expb{\CPABE,\qSim,\qA}{ad}{sim}(\secp, 1)]$ are negligibly close conditioned that $\msg=0$.
We prove this by using the following sequence of experiments.
\begin{description}
\item{$\hybi{1}$:}This is $\expb{\CPABE,\qSim,\qA}{ad}{sim}(\secp, 1)$ where $\qA$ outputs $\msg=0$. The detailed description is as follows.
\begin{enumerate}
    \item The challenger computes $(\ind.\mpk,\ind.\msk) \gets \IND.\Setup(1^\secp)$ and $R\la\bit^\secp$, and sends $\mpk\seteq(R,\ind.\mpk)$ to $\qA$.
    \item $\qA$ can get access to the following oracle.
            \begin{description}
            \item[$\Oracle{\KG,1}(x)$:] Given $x$, it generates $c\la\bit^\ctlen$ and returns $\sk_x\la\IND.\KG(\ind.\msk,x\|c)$. 
            \end{description}
    \item $\qA$ sends $C$ and $\msg=0$ to the challenger, where $C$ satisfies $C(x)=0$ for all $x$ queried by $\qA$ in the previous step. The challenger does the following.
 \begin{itemize}
 \item Generate $k\la\bit^\secp$ and $\lock\la\bit^\secp$.
 \item Generate $\ind.\ct\la\IND.\Enc(\ind.\mpk,G[C,k,R],\lock)$, where $G[C,k,R]$ takes as input $(x,c)$ and output $1$ if and only if $C(x)=1$ and $\SKE.\D(k,c)\ne R$.
 \item Return $\ct\la\CCObf(1^\secp,\IND.\Dec(\cdot,\ind.\ct),\lock,R)$.
 \end{itemize}
    \item$\qA$ can get access to the following oracle.
    \begin{description}
            \item[$\Oracle{\KG,2}(x)$:] Given $x$, if $C(x)=0$, it returns $\sk_x\la\IND.\KG(\ind.\msk,x\|c)$, where $c\la\bit^\ctlen$. Otherwise, it returns $\sk_x$ generated as follows.
                 \begin{itemize}
 \item Generate $c\la\SKE.\E(k,R)$.
\item Returns $\sk_x\la\IND.\KG(\ind.\msk,x\|c)$.
    \end{itemize}
            \end{description}
    \item $\qA$ outputs $\coin^\prime\in \bit$.
\end{enumerate}
\end{description}

\begin{description}
\item[$\hybi{2}$:]This is the same as $\hybi{1}$ except that $\ind.\ct$ is generated as $\ind.\ct\la\IND.\Enc(\ind.\mpk,G[C,k,R],0^\secp)$.
\end{description}

In $\hybi{1}$ and $\hybi{2}$, $\qA$ can query $x$ such that $C(x)=1$ to $\Oracle{\KG,2}$.
However, for such query $x$, $\Oracle{\KG,2}$ returns $\sk_x\la\IND.\KG(\ind.\msk,x\|c)$, where $c\la\SKE.\E(k,R)$.
We see that if $c\la\SKE.\E(k,R)$, $G[C,k,R](x\|c)=0$ since $\SKE.\D(k,c)=R$.
Thus, $\qA$ can obtain decryption keys only for an attribute $x\|c$ such that $G[C,k,R](x\|c)=0$ for the policy $G[C,k,R]$.
Then, from the AD-IND security of $\CPABE$, we have $\abs{\Pr[\hybi{1}=1]-\Pr[\hybi{2}=1]}=\negl(\secp)$.
\begin{description}
\item[$\hybi{3}$:]This is the same as $\hybi{2}$ except that $\ct$ is generated as $\ct\la\CC.\Sim(1^\secp,\pp_{\IND.\Dec},\abs{R})$.
\end{description}

From the security of $\CCObf$, we have $\abs{\Pr[\hybi{2}=1]-\Pr[\hybi{3}=1]}=\negl(\secp)$.

\begin{description}
\item[$\hybi{4}$:]This is the same as $\hybi{3}$ except that $\Oracle{\KG,2}$, given an input $x$, $c$ is generated as $c\la\bit^\ctlen$ even when $C(x)=1$.
\end{description}

From the security of $\SKE$, we have $\abs{\Pr[\hybi{3}=1]-\Pr[\hybi{4}=1]}=\negl(\secp)$.

$\hybi{4}$ is exactly $\expb{\CPABE,\qSim,\qA}{ad}{sim}(\secp, 0)$ where $\qA$ outputs $\msg=0$.
Thus, from the above discussions, we have  $\abs{\Pr[1\la \expb{\CPABE,\qSim,\qA}{ad}{sim}(\secp, 0)]-\Pr[1\la \expb{\CPABE,\qSim,\qA}{ad}{sim}(\secp, 1)]}=\negl(\secp)$ conditioned that $\qA$ outputs $\msg=0$.
This completes the proof.
\end{proof}

	%\input{proof_CP}
	%\input{proof_SIM-PE}
	% !TEX root = main.tex

\section{Succinct CPFE}\label{sec:succinct_CPFE}

We review the definition of hash encryption introduced in \cite{PKC:DGHM18}.
\begin{definition}[Hash Encryption]\label{def-he}
A hash encryption scheme $\HE$ is a four tuple $(\HKG, \allowbreak \Hash, \allowbreak \HEnc, \allowbreak \HDec)$ of PPT algorithms.
\begin{itemize}
\item $\HKG$ is the key generation algorithm that takes as input a security parameter $1^\lambda$ and the input-length $1^\inplen$.
Then, it outputs a hash key $\hk$.

\item $\Hash$ is the (deterministic) hashing algorithm that takes a hash key $\hk$ and a string $x\in \bin^\inplen$ as input, and outputs a hash value $h \in \bin^\lambda$.

\item $\HEnc$ is the encryption algorithm that takes a hash key $\hk$, a triple $(h, j, \alpha) \in \bin^{\lambda} \times [\inplen] \times \bin$, and a message $m\in\bin^*$ as input, and outputs a ciphertext $\hct$.

\item $\HDec$ is the (deterministic) decryption algorithm that takes a hash key $\hk$, a string $x\in\bin^\inplen$, and a ciphertext $\ct$ as input, and outputs a message $m$ which could be the special invalid symbol $\bot$.
\end{itemize}

We require the following properties.

\begin{description}
%\item[Correctness] We require $\HDec(\hk, x, \HEnc(\hk, (h,j,x_j),m)) = m$ for all $\lambda \in \N$, all strings $x = (x_1, \dots, x_{\inplen}) \in \bin^\inplen$, all positions $j\in[\inplen]$, all hash keys $\hk$ output by $\HKG(1^\lambda,\inplen)$, all hash values $h = \Hash(\hk, x)$, and all messages $m$.

\item[Correctness] Let $\hk\la\HKG(1^\lambda,1^\inplen)$.
We have $\HDec(\hk, x, \HEnc(\hk, (\Hash(\hk, x),j,x[j]),m)) = m$ for all strings $x = (x[1], \dots, x[\inplen]) \in \bin^\inplen$, positions $j\in[\inplen]$, and plaintext $m\in\bit^*$.
%all hash values $h = \Hash(\hk, x)$,
and all messages $m$.

\item[Security]
Consider the following security experiment $\expa{\HE,\qA}{he}(\secp,\coin)$ between a challenger and an adversary $\qA$.
\begin{enumerate}
	\item $\qA$ sends $x=(x[1], \dots, x[\inplen]) \in\bin^\inplen$ to the challenger.
\item The challenger generates $\hk \la \HKG(1^\lambda,1^\inplen)$ and sends $\hk$ to $\qA$.

\item $\qA$ sends a position $j\in[\inplen]$ and a pair of messages $(m_0,m_1)$ of the same length to the challenger.
The challenger computes $h \la \Hash(\hk,x)$ and $\hct \la \HE(\hk, (h,j,1 \oplus x_j),m_\coin)$,
and returns $\hct$ to $\qA$.
\item $\qA$ outputs $\coin' \in \bin$.
\end{enumerate}
For any QPT $\qA$, we have
\begin{align}
\adva{\HE, \qA}{he}(\lambda) \seteq \abs{\Pr[\expa{\HE,\qA}{he}(\secp,0)=1]-\Pr[\expa{\HE,\qA}{he}(\secp,1)=1]}= \negl(\secp).
\end{align}
\end{description}

\end{definition}

\begin{theorem}[\cite{PKC:DGHM18}]\label{thm:hash_encryption}
If the LWE or exponentially-hard LPN assumption holds, there exists a hash encryption.
\end{theorem}

We present a CPFE scheme that satisfies 1-bounded security~\cref{def:CPFE_security} and the succinct key property~\cref{def:succinct_key}.
\paragraph{Building blocks.}
\begin{itemize}
\item Hash encryption $\HE=(\HKG,\Hash,\HEnc,\HDec)$.
\item Grabled circuit $(\GC.\Garble,\GC.\Eval,\GC.\Sim)$.
\end{itemize}

Our CPFE scheme $\CPFE$ is as follows.
\begin{description}
 \item[$\Setup(1^\secp,x)$:] $ $
 \begin{itemize}
\item Generate $\hk \gets \HKG(1^\secp,1^\inplen)$.
\item Compute $h_x \seteq \Hash(\hk,x)$
\item Output $\MPK \seteq (\hk,h_x)$ and $\sk_x \seteq x$.
 \end{itemize}
 \item[$\Enc(\MPK,C)$:] $ $
 \begin{itemize}
 \item Parse $\MPK = (\hk,h_x)$.
 \item Generate $(\tlC,\setbk{\lbl_{i,\beta}}_{i\in [\inplen], \beta\in\zo{}})\gets \GC.\Garble(1^\secp,C)$.
 \item Generate $\hct_{i,\beta} \gets \HEnc(\hk,(h_x,i,\beta),\lbl_{i,\beta})$ for $i\in [\inplen]$ and $\beta\in\zo{}$.
 \item Output $\ct \seteq (\tlC,\setbk{\hct_{i,\beta}}_{i\in[\inplen],\beta\in\zo{}})$.
 \end{itemize}
\item[$\Dec(\sk_x,\ct)$:] $ $
\begin{itemize}
\item Parse $\sk_x = x$ and $\ct = (\tlC,\setbk{\hct_{i,\beta}}_{i\in[\inplen],\beta\in\zo{}})$.
\item Compute $\lbl_{i,x[i]} \seteq \HDec(\hk,x,\hct_{i,x[i]})$ for $i\in[\inplen]$.
\item Output $y^\prime \seteq \GC.\Eval(\tlC,\setbk{\lbl_{i,x[i]}}_{i\in[\inplen]})$.
\end{itemize}
\end{description}

\begin{theorem}\label{thm:CPFE_from_HE_GC}
If $\HE$ is a secure hash encryption and $\GC$ is a secure garbling, then $\CPFE$ is 1-bounded secure.
\end{theorem}

\begin{proof}
We define a sequence of games to prove the theorem.
\begin{description}
\item[$\hybi{0}$:] This is $\expa{\HE,\qA}{he}(\secp,\coin)$ where $\coin\la\bit$.
\item[$\hybij{0}{j}$:] This is the same as $\hybi{0}$ except that for $i \in [j]$ we generate $\hct_{i,1\xor x[i]} \gets \HEnc(\hk,(h_x,i,1\xor x[i]),0^{\abs{\lbl}})$ instead of generating $\hct_{i,1\xor x[i]} \gets \HEnc(\hk,(h_x,i,1\xor x[i]),\lbl_{i,1\xor x[i]})$.
\item[$\hybi{1}$:] This is the same as $\hybij{0}{n}$.
\item[$\hybi{2}$:] This is the same as $\hybi{1}$ except that we generate $(\tlC_{\coin},\setbk{\lbl_{i,x[i]}}_{i\in[\inplen]}) \gets \GC.\Sim(1^\secp,1^{\abs{C}},C_{\coin}(x))$ instead of generating $(\tlC_{\coin},\setbk{\lbl_{i,\beta}}_{i\in[\inplen],\beta\zo{}}) \gets \GC.\Garble(1^\secp,C_{\coin})$.
\end{description}
We define $\SUC_i$ (resp. $\SUC_i^j$) be the event that $\qA$ outputs $\coin^\prime=\coin$ in $\hybi{i}$ (resp. $\hybij{i}{j}$).

We have $\hybij{0}{0}=\hybi{0}$.
It holds that $\abs{\Pr[\SUC_0^{j-1}] - \Pr[\SUC_0^{j}]}=\negl(\secp)$ for all $j \in [n]$ due to the security of $\HE$.

Due to the security of $\GC$, it holds that $\abs{\Pr[\SUC_1]-\Pr[\SUC_2]}=\negl(\secp)$ since $\setbk{\lbl_{i,1\xor x[i]}}_{i\in [n]}$ is never used in $\hybi{1}$.

It trivially holds $\Pr[\SUC_2]=\frac{1}{2}$ since $C_0(x)=C_1(x)$ and $\tlC_{\coin}$ is generated by $\GC.\Sim(1^\secp,1^{\abs{C}},C_{\coin}(x))$ in $\hybi{2}$. This complete the proof.
\end{proof}

\begin{theorem}\label{thm:succinct_key_from_HE}
$\CPFE$ satisifes the succinct key property.
\end{theorem}
\begin{proof}
It trivially holds since the setup algorithm $\Setup$ runs $\HKG(1^\secp,1^n)$ and $\Hash(\hk,x)$, and outputs $\MPK \seteq (\hk,h)$ and $\sk_x\seteq x$.
\end{proof}

We complete the proof of~\Cref{thm:single_key_CPFE_succinct_key} by~\cref{thm:hash_encryption,thm:CPFE_from_HE_GC,thm:succinct_key_from_HE}.

	% !TEX root = main.tex

\section{Injective Commitment with Equivocal Mode}\label{sec:Naor_com}
We present a variant of Naor's commitment where we use a commitment key instead of receiver's first message.
We use injective PRG $\PRG: \zo{\secp} \ra \zo{3\secp}$, which can be constructed from injective OWF (with evaluation key generation algorithm).\footnote{See a remark by Kitagawa and Nishimaki~\cite[Section B.1 in the full version]{EC:KitNis22} for ``injective OWF with evaluation key generation algorithm''.}
Let $\cM\seteq \zo{\ell}$ and $\cR \seteq \zo{\secp\cdot \ell}$.
\begin{description}
\item[$\Setup(1^\secp)$:]$ $
\begin{itemize}
 \item Choose $s_i\chosen \zo{3\secp}$ for $i\in [\ell]$.
 \item Output $\ck \seteq (s_1,\ldots,s_\ell)$.
 \end{itemize}

\item[$\Commit(\ck,m\in \zo{\ell})$:] $ $
\begin{itemize}
    \item Parse $\ck = (s_1,\ldots,s_\ell)$.
\item Choose $r_i \chosen \zo{\secp}$ for $i\in [\ell]$.
\item Compute $x_i \seteq \PRG(r_i)$ for $i\in [\ell]$.
\item Set $y_i \seteq x_i$ if $m_i =0$, otherwise $y_i \seteq x_i \xor s_i$ for $i\in [\ell]$.
\item Output $\com \seteq (y_1,\ldots,y_\ell)$.
\end{itemize}

\item[$\EqSetup(1^\secp)$:] $ $
\begin{itemize}
 \item Choose $\tlr_{b,i} \chosen \zo{\secp}$ for $i\in [\ell]$ and $b\in\zo{}$.
 \item Set $s^\ast_i \seteq \PRG(\tlr_{0,i}) \xor \PRG(\tlr_{1,i})$ for $i\in[\ell]$.
 \item Set $y^\ast_i \seteq \PRG(\tlr_{0,i})$ for $i\in [\ell]$.
 \item Output $\ck^\ast \seteq (s^\ast_1,\ldots,s^\ast_\ell)$, $\com^\ast \seteq (y^\ast_1,\ldots,y^\ast_\ell)$, and $\td \seteq \setbk{\tlr_{b,i}}_{i\in[\ell],b\in\zo{}}$.
 \end{itemize}
\item[$\Open(\td,m\in\zo{\ell},\com^\ast)$:] $ $
\begin{itemize}
\item Parse $\td = \setbk{\tlr_{b,i}}_{i\in[\ell],b\in\zo{}}$ and $\com^\ast \seteq (y^\ast_1,\ldots,y^\ast_\ell)$.
\item Set $r^\ast_i \seteq \tlr_{0,i}$ if $m_i =0$, otherwise $r^\ast_i \seteq \tlr_{1,i}$.
\item Output $r^\ast \seteq (r^\ast_1,\ldots,r^\ast_\ell)$.
\end{itemize}
\end{description}

We can verify that a message $m$ and randomness $(r_1,\ldots,r_\ell)$ is a valid opening for a commitment $(y_1,\ldots,y_\ell)$ by checking
\[
y_i = \PRG(r_i) \xor m_i \cdot s_i
\]
for all $i\in [\ell]$.

This is the Naor's commtiement~\cite{JC:Naor91} and it has statistical binding and computational hiding properties.

It is easy to see that the construction above satisfies injectivity since $s_i$ is uniformly random for all $i\in[\ell]$ and $\PRG$ is injective.

It is also easy to see that the construction above satisifes trapdoor equivocality.
Due to the pseudorandomness of $\PRG$, $(\ck^\ast,\com^\ast)$ is computationally indistinguishable from $(\ck,\com)$.\footnote{We can use security of $\PRG(\tlr_{1,i})$ and $\PRG(\tlr_{0,i})$ if we open to $m_i =0$ and $m_i=1$, respectively.} In addition, it holds that $y_i^\ast  = \PRG(\tlr_{m_i,i}) \xor s_i^\ast \cdot m_i$ since $s_i^\ast = \PRG(\tlr_{0,i})\xor \PRG(\tlr_{1,i})$.
Thus, the trapdoor equivocality holds.

\else
%%%%% LNCS-style submission version region %%%%%
\ifnum\submission=1
	\newpage
	 	\appendix
	 	\setcounter{page}{1}
 	{
	\noindent
 	\begin{center}
	{\Large SUPPLEMENTAL MATERIALS}
	\end{center}
 	}
	
	\setcounter{tocdepth}{2}
	 	\ifnum\noaux=1
 	\else
{\color{red}{We attached the full version of this paper as a separated file (auxiliary supplemental material) for readability. It is available from the program committee members.}}

\fi

\fi

	\setcounter{tocdepth}{2}
	\tableofcontents

	\fi
	\else
	%%%%%%% Camera-ready region (no appendix) %%%%%%
\fi

\end{document}